\newcommand{\subversion}[1]{}

\documentclass[11pt,pdfa]{article}
\usepackage[in]{fullpage}

\newcommand{\eps}{\varepsilon}

\usepackage{amsmath}
\usepackage{amsfonts}
\usepackage{amssymb}
\usepackage{amsthm}

\usepackage[colorlinks]{hyperref}
  \hypersetup{linkcolor=blue,filecolor=blue,citecolor=blue,urlcolor=blue}
\usepackage{cleveref}

\usepackage[T1]{fontenc}
\usepackage{comment}
\usepackage{enumitem}
\usepackage{braket}

\usepackage{xcolor}
\newcommand{\alice}{\mathcal{A}}
\newcommand{\bob}{\mathcal{B}}
\newcommand{\charlie}{\mathcal{C}}
\newcommand{\abc}{\left(\mathcal{A}, \mathcal{B}, \mathcal{C}\right)}
\newcommand{\ch}{\mathsf{ch}}
\newcommand{\aliceprime}{\alice'}
\newcommand{\bobprime}{\bob'}
\newcommand{\charlieprime}{\charlie'}
\newcommand{\abcprime}{\brackets{\aliceprime, \bobprime, \charlieprime}}

\newcommand{\abcnum}[1]{\brac{ \alice_{#1}, \bob_{#1}, \charlie_{#1}} }

\newtheorem{theorem}{Theorem}
\newtheorem{lemma}[theorem]{Lemma}
\newtheorem{fact}[theorem]{Fact}
\newtheorem{claim}{Claim}

\newtheorem{corollary}{Corollary}
\newtheorem{remark}{Remark}
\newtheorem{definition}{Definition}

\newcommand{\bit}{\left\{0,1\right\}}
\newcommand{\uniform}{\xleftarrow{\$}}

\newcommand{\pr}[1]{\Pr \left[ #1 \right] }
\newcommand{\from}{\leftarrow}
\newcommand{\secparam}{\lambda}
\newcommand{\poly}{\mathsf{poly}}

\newcommand{\ketbraX}[1]{\ket{#1}\!\!\bra{#1}}

\newcommand{\negl}{\mathsf{negl}}
\newcommand{\tr}{\mathsf{Tr}}
\newcommand{\given}{\mid}
\DeclareMathOperator*{\E}{\mathbb{E}}

\newcommand{\sk}{\mathsf{sk}}

\newcommand{\distr}{\cD}

\newcommand{\identity}{I}
\newcommand{\opNorm}[1]{\left\| #1 \right\|_{\mathsf{op}}}
\newcommand{\abs}[1]{\left| #1 \right|}
\newcommand{\inner}[2]{\langle #1 , #2 \rangle}
\newcommand{\tracedist}[2]{\mathsf{TD}\brackets{#1, #2}}
\newcommand{\density}[1]{\cD(\cH_{#1})}

\newcommand{\brackets}[1]{\left( #1 \right)}
\newcommand{\bracketsSquare}[1]{\left[ #1 \right]}
\newcommand{\bracketsCurly}[1]{\left\{ #1 \right\}}
\newcommand{\brac}{\brackets}
\newcommand{\bracC}{\bracketsCurly}
\newcommand{\bracS}{\bracketsSquare}

\newcommand{\enc}{\mathsf{Enc}}
\newcommand{\dec}{\mathsf{Dec}}
\newcommand{\setup}{\mathsf{Setup}}
\newcommand{\gen}{\mathsf{Gen}}

\newcommand{\sign}{\mathsf{Sign}}
\newcommand{\ver}{\mathsf{Ver}}

\newcommand{\cp}{\mathsf{CP}}
\newcommand{\copyprotect}{\cp}
\newcommand{\eval}{\mathsf{Eval}}

\newcommand{\sati}{\mathsf{SATI}}
\newcommand{\ati}{{\sf ATI}}

\newcommand{\fatih}[1]{}
\newcommand{\pnote}[1]{}
\definecolor{airforceblue}{rgb}{0.36, 0.54, 0.66}
\newcommand{\qipeng}[1]{}

\newcommand{\highlight}[1]{{\textcolor{blue}{#1}}}

\newcommand{\N}{\mathbb{N}}
\newcommand{\Z}{\mathbb{Z}}

\newcommand{\R}{\mathbb{R}}
\newcommand{\F}{\mathbb{F}}

\newcommand{\cA}{\mathcal{A}}
\newcommand{\cB}{\mathcal{B}}
\newcommand{\cC}{\mathcal{C}}
\newcommand{\cD}{\mathcal{D}}
\newcommand{\cE}{\mathcal{E}}
\newcommand{\cF}{\mathcal{F}}

\newcommand{\cH}{\mathcal{H}}

\newcommand{\cK}{\mathcal{K}}
\newcommand{\cM}{\mathcal{M}}
\newcommand{\cO}{\mathcal{O}}
\newcommand{\cU}{\mathcal{U}}
\newcommand{\cP}{\mathcal{P}}
\newcommand{\cR}{\mathcal{R}}
\newcommand{\cS}{\mathcal{S}}
\newcommand{\cX}{\mathcal{X}}
\newcommand{\cY}{\mathcal{Y}}

\newcommand{\aux}{\mathsf{aux}}
\newcommand{\game}{\mathcal{G}}
\newcommand{\gamesetup}{(\setup, \tokengen, \gench, \ver)}
\newcommand{\gench}{\mathsf{GenC}}

\newcommand{\ans}{\mathsf{ans}}
\newcommand{\tokengen}{\mathsf{GenT}}

\newcommand{\referee}{\mathsf{Ref}}
\newcommand{\cloninggame}{\mathfrak{CE}}

\newcommand{\alicecor}[1]{\alice_{#1}}

\newcommand{\identical}{\mathsf{id}}
\newcommand{\independent}{\mathsf{ind}}

\newcommand{\augcloninggame}{\mathfrak{ACE}}
\newcommand{\gamesetupprime}{(\setup', \tokengen', \gench',\ver')}

\newcommand{\Bs}{B}
\newcommand{\Cs}{C}

\newcommand{\distrtild}{\widetilde{\distr}}

\newcommand{\assumption}{\mathfrak{A}}

\newcommand{\adversary}{{\cal A}}
\newcommand{\val}{\mathsf{val}}
\newcommand{\adv}{\mathsf{adv}}
\newcommand{\prob}{\mathsf{Pr}}
\newcommand{\state}{{\sf state}}
\newcommand{\valest}{\mathsf{ValEst}}
\newcommand{\repair}{\mathsf{Repair}}
\newcommand{\nlvalest}{\mathsf{NLValEst}}
\newcommand{\nlrepair}{\mathsf{NLRepair}}
\newcommand{\bfB}{{\bf B}}
\newcommand{\bfC}{{\bf C}}
\newcommand{\distrpr}{\widetilde{\distr}}
\newcommand{\simulator}{\mathsf{Sim}}
\newcommand{\bfrho}{\mathbf{\rho}}
\newcommand{\bfsigma}{\mathbf{\sigma}}

\usepackage[style=alphabetic,minalphanames=3,maxalphanames=4,maxnames=99,backref=true]{biblatex}

  \DeclareFieldFormat{eprint:iacr}{Cryptology ePrint Archive: \href{https://ia.cr/#1}{\texttt{#1}}}
  \DeclareFieldFormat{eprint:iacrarchive}{Cryptology ePrint Archive: \href{https://eprint.iacr.org/archive/#1}{\texttt{#1}}}
\title{Cloning Games: A General Framework for Unclonable Primitives}
\author{Prabhanjan Ananth\thanks{\texttt{prabhanjan@cs.ucsb.edu}}\\ {\small UCSB} \and Fatih Kaleoglu\thanks{\texttt{kaleoglu@ucsb.edu}}\\ {\small UCSB} \and Qipeng Liu\thanks{\texttt{qipengliu0@gmail.com}}\\ {\small Simons Institute}}
\date{}

\usepackage{tikz}
\usetikzlibrary{matrix}

\begin{document}

\maketitle

\begin{abstract}
\noindent The powerful no-cloning principle of quantum mechanics can be leveraged to achieve interesting primitives, referred to as unclonable primitives, that are impossible to achieve classically. In the past few years, we have witnessed a surge of new unclonable primitives. While prior works have mainly focused on establishing feasibility results, another equally important direction, that of understanding the relationship between different unclonable primitives is still in its nascent stages. Moving forward, we need a more systematic study of unclonable primitives. 
\par To this end, we introduce a new framework called {\em cloning games}. This framework captures many fundamental unclonable primitives such as quantum money, copy-protection, unclonable encryption,  single-decryptor encryption, and many more. By reasoning about different types of cloning games, we obtain many interesting implications to unclonable cryptography, including the following:
\begin{enumerate}
    \item We obtain the first construction of information-theoretically secure single-decryptor encryption in the one-time setting. 
    \item We construct unclonable encryption in the quantum random oracle model based on BB84 states, improving upon the previous work, which used coset states. Our work also provides a simpler security proof for the previous work.
    \item We construct copy-protection for single-bit point functions in the quantum random oracle model based on BB84 states, improving upon the previous work, which used coset states, and additionally, providing a simpler proof.
    \item We establish a relationship between different challenge distributions of copy-protection schemes and single-decryptor encryption schemes.
    \item Finally, we present a new construction of one-time encryption with certified deletion.  
\end{enumerate}

\end{abstract}

\newpage
\tableofcontents
\addtocontents{toc}{\protect\setcounter{tocdepth}{2}}
\newpage

    \section{Introduction}
Unclonable cryptography is a prominent research area that lies at the intersection of quantum computing and cryptography. This research area consists of many fascinating primitives that solve cryptographic problems using quantum information that are impossible to solve using only classical technology. At the heart of this area is the no-cloning principle of quantum mechanics~\cite{WZ82,Dieks82}, which states that no universal cloner can clone arbitrary quantum states. Since Wiesner put forward quantum money in 1983~\cite{Wiesner83}, a novel unclonable primitive that protects digital money against counterfeiting attacks, there have been a myriad of interesting unclonable primitives proposed over the years. They include variants of quantum money~\cite{AC12,Zha17,Shm22}, quantum one-time programs~\cite{BGS13}, copy-protection~\cite{Aar09,AL20,ALLZZ20,CLLZ21}, tokenized signatures~\cite{BDS16,CLLZ21,Shm22}, unclonable encryption~\cite{Got02,BL20}, secure software leasing~\cite{AL20,KNY21,BJLPS21}, encryption with certified deletion~\cite{BI20} and certified zero-knowledge~\cite{HMNY22}. 
\par We discuss three unclonable primitives that are the main focus of this work\footnote{ An (impatient) reader familiar with the above primitives could skip directly to~\Cref{sec:complexity}. We still recommend going through the discussion before reading~\Cref{sec:complexity}.}.

\paragraph{Unclonable Encryption.} Roughly speaking, an unclonable encryption scheme, introduced by \cite{BL20}, is a type of symmetric key encryption scheme that protects ciphertexts, encoded in quantum states, from being illegally distributed. To formalize this, we first consider the following security experiment. The adversary participating in the security experiment is referred to as a cloning adversary, consisting of three algorithms, namely $\alice$, $\bob$ and $\charlie$. $\alice$ receives a quantum state in the \emph{setup phase}. The quantum state is a ciphertext, which is an encryption of a message $m$ computed using a private key $k$. Then, $\alice$ sends a bipartite state to the spatially separated parties $(\bob, \charlie)$ during the \emph{splitting phase}. Finally, $\bob$ and $\charlie$ are asked to simultaneously pass verification in the \emph{challenge phase}. In more detail, in the challenge phase, $\bob$ and $\charlie$ both receive the classical decryption key $k$ and $\bob$ outputs $b_{\bob}$ while $\charlie$ outputs $b_{\charlie}$. We say that $\abc$ wins if $b_{\bob}=b_{\charlie}=m$.

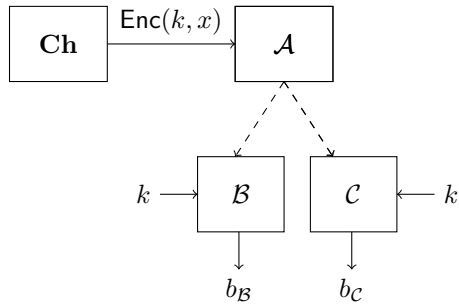
\newcommand{\regA}{{\bf B}}
\newcommand{\regB}{{\bf C}}
\begin{figure}[!htb]
\begin{center}
{\small
\begin{tikzpicture}
\draw (2,-0.5) rectangle (3.3,0.5) node [pos=.5]{${\bf Ch}$}; 
  \draw (5,-0.5) rectangle (6.3,0.5) node [pos=.5]{${\cal A}$}; 
  \draw[->] (3.3,0) -- node[above]{{\small $\mathsf{Enc}(k,x)$}} (5,0);
  \draw[->,dashed] (5.65,-0.5) -- node[left]{} (5,-1.5);
  \draw[->,dashed] (5.65,-0.5) -- node[right]{} (6.3,-1.5);
  \draw (4.5,-1.5) rectangle (5.65,-2.5) node [pos=.5]{${\cal B}$};
  \draw[->] (4,-2) node[left]{$k$} -- (4.5,-2);
  \draw[->] (5.05,-2.5) -- (5.05,-3) node[below]{$b_{\cal B}$};
  \draw[->] (7.65,-2) node[right]{$k$} -- (7.15,-2);
  \draw[->] (6.55,-2.5) -- (6.55,-3) node[below]{$b_{\cal C}$};
  \draw[->,dashed] (5.65,-0.5) -- (6.3,-1.5);
  \draw (6,-1.5) rectangle (7.15,-2.5) node [pos=.5]{${\cal C}$};
  \draw (5,-0.5) rectangle (6.3,0.5) node [pos=.5]{${\cal A}$};
 \end{tikzpicture}

\label{fig:nocloning}
}
\end{center}
\end{figure}
\vspace{-2em}

\noindent With the above security experiment in mind, there are two ways to define security.
\begin{itemize}
    \item {\em Unclonability}: In this case, $m$ is sampled uniformly at random from the message space. We say that the scheme is $\eps$-secure if the probability that the adversary $\abc$ wins is $\eps$. Ideally, we would require that $\eps$ is negligible in $|m|$. 
    \item {\em Unclonable-Indistinguishability}: In this case, $m$ is sampled uniformly at random\footnote{We note that the security in the literature is stated slightly differently. $\alice$ is given encryption of a message $m_b$, where $b$ is picked uniformly at random and $\bob,\charlie$ is expected to simultaneously guess $b$. We note that this formulation is identical to the above formulation.} from some adversarially chosen set $\{m_0,m_1\}$. Similar to the above property, we can define $\eps$-security. Ideally, we would require $\eps$ to be negligibly close to $0.5$.    
\end{itemize}
\noindent Public-key unclonable encryption schemes have also been considered by~\cite{AK21,AKLLZ22}. 

\paragraph{Copy-Protection.} Quantum copy-protection, introduced in~\cite{Aar09}, is a functionality-preserving compiler that transforms programs into quantum states. Moreover, we require that the resulting copy-protected state should not allow the adversary to copy the functionality of the state. 
\par The security experiment against cloning adversaries of the form $\abc$ is formalized as follows\footnote{The original formulation by~\cite{Aar09} is weaker than what is stated here. We follow the game-based definition by~\cite{CMP20}.}. $\alice$ receives an unclonable copy-protected program $\rho_f := {\sf CP}(f)$, which can be used to evaluate a classical function $f$\footnote{We only consider classes of unlearnable functions which are functions that cannot be efficiently learned from its input and output behavior. Copy-protection for learnable functions is impossible.}. In the challenge phase, $\bob$ and $\charlie$ receive inputs $x_{\bob}, x_{\charlie}$, sampled from a {\em challenge distribution} and are asked to output $b_{\bob},b_{\charlie}$. $\abc$ wins if $b_{\bob}=f(x_{\bob})$ and $b_{\charlie}=f(x_{\charlie})$, respectively.

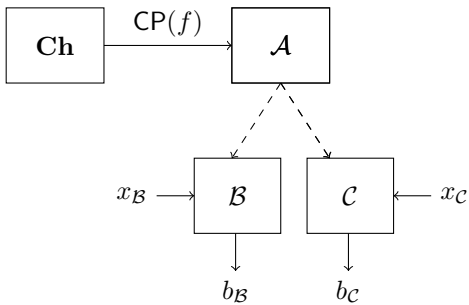
\begin{figure}[!htb]
\begin{center}
{\small
\begin{tikzpicture}
\draw (2,-0.5) rectangle (3.3,0.5) node [pos=.5]{${\bf Ch}$}; 
  \draw (5,-0.5) rectangle (6.3,0.5) node [pos=.5]{${\cal A}$}; 
  \draw[->] (3.3,0) -- node[above]{{\small ${\sf CP}(f)$}} (5,0);
  \draw[->,dashed] (5.65,-0.5) -- node[left]{} (5,-1.5);
  \draw[->,dashed] (5.65,-0.5) -- node[right]{} (6.3,-1.5);
  \draw (4.5,-1.5) rectangle (5.65,-2.5) node [pos=.5]{${\cal B}$};
  \draw[->] (4,-2) node[left]{$x_{\cal B}$} -- (4.5,-2);
  \draw[->] (5.05,-2.5) -- (5.05,-3) node[below]{$b_{\cal B}$};
  \draw[->] (7.65,-2) node[right]{$x_{\cal C}$} -- (7.15,-2);
  \draw[->] (6.55,-2.5) -- (6.55,-3) node[below]{$b_{\cal C}$};
  \draw[->,dashed] (5.65,-0.5) -- (6.3,-1.5);
  \draw (6,-1.5) rectangle (7.15,-2.5) node [pos=.5]{${\cal C}$};
  \draw (5,-0.5) rectangle (6.3,0.5) node [pos=.5]{${\cal A}$};
 \end{tikzpicture}

\vspace{-2em}
\label{fig:nocloning}
}
\end{center}
\end{figure}

\noindent Ideally, we would like to say that a copy-protection scheme is secure if the probability that $\abc$ wins is negligible in the output lengths. However, such a statement would be false if we are not careful in choosing the distributions from which $f$ is sampled and $x_{\bob},x_{\charlie}$ are sampled. For example, if $f$ is sampled from a distribution with support size one, then the adversary $\abc$ clearly knows the function being copy-protected and can thus easily violate the security. Even if $f$ is sampled from a high-entropy distribution, we should also require $x_{\bob}$ and $x_{\charlie}$ to come from high-entropy distributions for the definition to be meaningful. For example, if we set $x_{\charlie}$ to be a fixed element, then $\abc$ can always win by $\alice$ first computing the value of the function $f(x_{\charlie})$ and then handing over the copy-protected state to $\bob$ and then handing over $f(x_{\charlie})$ to $\charlie$. Moreover, our definition should be robust in even handling functions with single-bit outputs. This suggests that we must carefully examine the challenge distribution when evaluating results on constructions of copy-protection schemes.   

\paragraph{Single-Decryptor Encryption.} A single-decryptor encryption, introduced in~\cite{GZ20,CLLZ21}, enables a user to delegate their decryption key, represented as a quantum state, such that the delegated key cannot be used to illegally distribute two or more decryption keys that can decrypt ciphertexts. \par Formally, $\alice$ receives an unclonable decryption key $\rho_{{\sf sk}}$ for an encryption scheme where the encryption procedure and the message space are both classical. In the challenge phase, $\bob$ and $\charlie$ respectively receive ciphertexts ${\sf ct}_{\bob}, {\sf ct}_{\charlie}$ encrypting messages $m_{\bob}, m_{\charlie}$. They are then expected to output $b_{\bob},b_{\charlie}$ respectively. 
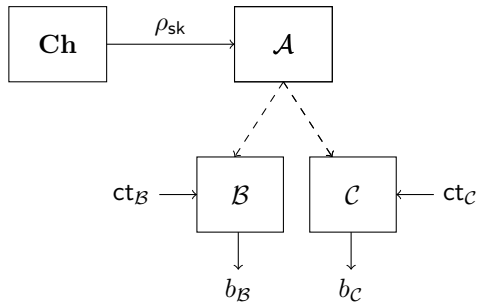
\begin{figure}[!htb]
\begin{center}
{\small
\begin{tikzpicture}
\draw (2,-0.5) rectangle (3.3,0.5) node [pos=.5]{${\bf Ch}$}; 
  \draw (5,-0.5) rectangle (6.3,0.5) node [pos=.5]{${\cal A}$}; 
  \draw[->] (3.3,0) -- node[above]{{\small $\rho_{\sf sk}$}} (5,0);
  \draw[->,dashed] (5.65,-0.5) -- node[left]{} (5,-1.5);
  \draw[->,dashed] (5.65,-0.5) -- node[right]{} (6.3,-1.5);
  \draw (4.5,-1.5) rectangle (5.65,-2.5) node [pos=.5]{${\cal B}$};
  \draw[->] (4,-2) node[left]{${\sf ct}_{\cal B}$} -- (4.5,-2);
  \draw[->] (5.05,-2.5) -- (5.05,-3) node[below]{$b_{\cal B}$};
  \draw[->] (7.65,-2) node[right]{${\sf ct}_{\cal C}$} -- (7.15,-2);
  \draw[->] (6.55,-2.5) -- (6.55,-3) node[below]{$b_{\cal C}$};
  \draw[->,dashed] (5.65,-0.5) -- (6.3,-1.5);
  \draw (6,-1.5) rectangle (7.15,-2.5) node [pos=.5]{${\cal C}$};
  \draw (5,-0.5) rectangle (6.3,0.5) node [pos=.5]{${\cal A}$};
 \end{tikzpicture}
\vspace{-2em}
\label{fig:nocloning}
}
\end{center}
\end{figure}

\noindent Depending on the specification of the distributions from which the ciphertexts are sampled and how $b_{\bob},b_{\charlie}$ are defined, there are many ways to define security for single-decryptor encryption.  
\begin{itemize}
    \item We could require ${\sf ct}_{\bob}={\sf ct}_{\charlie}$ (referred to as {\em identical} ciphertext distribution), in which case $m_{\bob}=m_{\charlie}$, or we could require that ${\sf ct}_{\bob}$ and ${\sf ct}_{\charlie}$ to be sampled independently (referred to as {\em independent} ciphertext distribution).
    \item Analogous to the unclonable encryption setting, we could require that the messages $m_{\bob}$, $m_{\charlie}$ are picked from the uniform distribution, or they are sampled from a set of two messages chosen by $\abc$. 
\end{itemize}

\subsection{Complexity of Unclonable Primitives} 
\label{sec:complexity}
Most prior works on unclonable encryption, copy-protection, single-decryptor encryption, and other unclonable primitives mainly focus on feasibility. A few exceptions include the works of~\cite{CMP20,AK21,SW22}, who make partial progress in understanding the relationship between unclonable encryption,  copy-protection, and single-decryptor encryption. 
\par In order to achieve a deeper understanding of the area, we need to move beyond the feasibility results and investigate how different primitives are related to each other. There are many reasons why we should care about understanding the relationship between unclonable primitives, and we discuss some of them below.

\paragraph{Computational Assumptions.} Firstly, it leads to a better understanding of the computational assumptions necessary in the conception of unclonable primitives. While some primitives require powerful cryptographic tools such as post-quantum indistinguishability obfuscation, some other primitives can even be conceived information-theoretically. It would be interesting to classify the unclonable primitives based on the computational assumptions necessary to construct them. In classical cryptography, via Impagliazzo's five worlds~\cite{impagliazzo5worlds} and numerous black-box separations~\cite{IR}, we have a solid understanding of the minimal computational assumptions necessary for the existence of primitives. We have just begun to understand the assumptions necessary for achieving cryptographic primitives in the quantum world~\cite{AQY22,MY22,BCQ22}. Investigating the implications between the unclonable primitives will help us classify these primitives based on their computational hardness. 

\paragraph{Types of States.} Secondly, not all unclonable primitives use the same types of states. Given any unclonable primitive, it is important to establish the types of states needed to achieve this notion. Some primitives~\cite{Wiesner83,BI20,BL20} use BB84 states~\cite{BB84}, some utilize subspace states~\cite{AC12,Zha17} and others take advantages of coset states~\cite{VZ20,CLLZ21,AKLLZ22}. BB84 states are preferred over subspace and coset states due to two facts: (a) they can be prepared easily (the preparation requires only Hadamard and $X$ gates), and (b) each qubit is unentangled with the other qubits. Since maintaining entanglement has been challenging in the existing quantum systems, understanding the feasibility of cryptographic systems using unentangled states is important. We currently have a limited understanding of whether BB84 states are sufficient for constructing many primitives. For instance, copy-protection for point functions with single-bit output seems to require coset states~\cite{AKLLZ22} whereas copy-protection for point functions with multi-bit output requires only BB84 states~\cite{CMP20}. 

\paragraph{Challenge Distributions.}  Unclonable primitives are often associated with challenge distributions. Thus, different feasibility results on the same unclonable primitive assuming different challenge distributions can be qualitatively incomparable. 
As was seen in the examples of unclonable encryption, copy-protection and single-decryptor encryption, the security of an unclonable primitive can be defined as a game between a challenger and an adversary composed of three parts $({\cal A},{\cal B},{\cal C})$. First, ${\cal A}$ receives an unclonable state (a.k.a. a quantum token) from the challenger, and it outputs a bipartite state shared by ${\cal B}$ and ${\cal C}$. Then, ${\cal B}$ and ${\cal C}$ receive samples from a distribution, called a \emph{challenge distribution}, and they output answers. 
\par It is often the case that security proven with respect to one challenge distribution does not necessarily imply security proven for a different challenge distribution. For instance, as was discussed earlier in the context of copy-protection, the choice of challenge distribution can qualitatively affect the type of result we get. Discerning the relationship between security notions of different challenge distributions will enable us to compare different results based on the challenge distributions they consider. Indeed, even in the literature, constructions of copy-protection for point functions have considered different challenge distributions~\cite{CMP20,BJLPS21,AKLLZ22}, which makes their results difficult to compare. Besides point functions, copy-protection was only known under certain distributions (product distributions).  

\paragraph{Porting Classical Techniques.} It turns out to be challenging to adopt many standard techniques employed to prove the security of cryptographic systems in the classical cryptography literature to the unclonable setting. Let us take an example. Traditionally, encrypting multiple bits can be generically reduced to encrypting single-bit messages in parallel using a simple hybrid argument. The same transformation fails when applied to the setting of unclonable encryption. Even standard search-to-decision reductions, such as Goldreich-Levin~\cite{GL89}, commonly used in the classical cryptography literature, cannot be directly ported to the unclonable setting. In the context of unclonable encryption,~\cite{AKLLZ22} discuss the challenges associated with using Goldreich-Levin and more in Section 1.1 of their work.

\subsection{Our Contributions}

\par In order to better understand the relationship between the unclonable primitives, we propose a new framework called {\em Cloning Games}. Firstly, we observe that many\footnote{As far as we know, all unclonable primitives can be cast as cloning games by making reasonable minor modifications to the framework.}

fundamental unclonable primitives can be cast as cloning games. We establish the relationship between large classes of cloning games. There are two directions we undertake to establish the relationship between cloning games. 
\begin{enumerate}
    \item In the first approach, we show that, under some conditions, the relationship between different cloning games can be reduced to the existence of classical reductions between two non-interactive assumptions. This approach gives a new toolkit to help us use classical techniques and computational assumptions to build unclonable primitives. We give an overview of this approach in~\Cref{sec:part1:overview}. 
    \item In the second approach, using new techniques, we refurbish existing constructions of primitives into generic transformations between cloning games. This approach leads to new constructions of primitives with improved features over prior works. An overview of this approach is given in~\Cref{sec:unclsrchtounclind}. 
\end{enumerate}

\noindent As a consequence of the above two approaches, we obtain new results in unclonable cryptography.

\paragraph{Single-Decryptor Encryption.} Existing constructions of single-decryptor encryption in the public-key setting, are based on post-quantum indistinguishability obfuscation~\cite{GZ20,CLLZ21}. It is worth investigating whether we can achieve single-decryptor encryption in the private-key setting based on well-studied assumptions. Indeed, even in the one-time setting, it was not known how to achieve single-decryptor encryption without relying on strong assumptions. By one-time setting, we mean that the adversary only gets one ciphertext computed using the private key. We show the following. 

\begin{theorem}[Informal] \label{thm:intro_sde}
There exists an information-theoretically secure one-time single-decryptor encryption scheme for single-bit messages.
\end{theorem}

\noindent The ciphertext distribution we consider in the above result is the following:

The challenger chooses the messages $m_{\bob} \xleftarrow{\$} \{0,1\}$ and $m_{\charlie} \xleftarrow{\$} \{0,1\}$. It then encrypts $m_{\bob}$ (resp., $m_{\charlie}$) and gives the ciphertext to $\bob$ (resp., $\charlie$). Then, $\bob$ and $\charlie$ are supposed to simultaneously guess which bit was encrypted. The security of our construction states that the success probability of any adversary is negligibly close to 0.5. 
\par Although our construction is only for 1-bit messages, we hope the toolkit we develop (\Cref{thm:main_lifting_reductions_to_cloning_games}) can be applied to obtain single-decryptor encryption for multi-bit messages in future work. 

\paragraph{Unclonable Encryption and Copy-Protection.} We revisit recent works that leveraged coset states to achieve unclonable primitives. Specifically, we focus on two constructions of unclonable-indistinguishable encryption and copy-protection for point functions by~\cite{AKLLZ22}. We show that these two constructions can be obtained from {\em any} encryption scheme satisfying the unclonability property in the quantum random oracle model. Ours is the first work to formally establish the relationship between unclonability and unclonable-indistinguiushability properties. 

\begin{theorem} \label{thm:intro_qrom}
Assuming the existence of one-time encryption satisfying unclonability, there exists an encryption scheme satisfying unclonable-indistinguishability in the quantum random oracle model. 
\par Assuming the existence of one-time encryption satisfying unclonability, there exists a copy-protection scheme for 1-bit output point functions in the quantum random oracle model. 
\end{theorem}

\noindent Unclonable encryption can be constructed from BB84 states~\cite{BL20}, and hence, as a consequence, we can obtain unclonable-indistinguishable encryption and copy-protection for point functions leveraging just BB84 states.

\begin{corollary}[Informal]\label{thm:BB84_construction}
There exists a (one-time) encryption scheme satisfying unclonable indistinguishability property, based on BB84 states, in the quantum random oracle model. 
\par There exists copy-protection for 1-bit output point functions, based on BB84 states, in the quantum random oracle model. 
\end{corollary}

\noindent In fact,~\cite{AK21} showed that encryption satisfying unclonability can be obtained from a variety of monogamy of entanglement games~\cite{TFKW13}. Consequently, we obtain both unclonable-indistinguishable encryption and copy-protection schemes based on a variety of quantum states, not just BB84 states.
\par Moreover, by plugging in the generic transformation from~\cite{AK21}, we achieve public-key unclonable encryption based on BB84 states. 

\paragraph{Relationship between Challenge Distributions.} In both copy-protection and single-decryptor encryption, the choice of challenge distribution plays a role in determining the usefulness of constructions. This makes comparing results difficult. For instance, a priori, it is unclear how to compare two different works constructing copy-protection for the same class of functions but with different challenge distributions. Similarly, even for single-decryptor encryption, schemes with different ciphertext distributions might be incomparable. 
\par We make progress in understanding the relationship between different challenge distributions. Although our result is more general, for the current discussion, let us focus on two types of distributions:
\begin{itemize}
    \item {\em Identical}: Both $\bob$ and $\charlie$ get the same challenge (challenge refers to input in the case of copy-protection and ciphertext in the case of single-decryptor encryption), drawn from some distribution. 
    \item {\em Independent}: $\bob$ and $\charlie$ each get two challenges chosen independently from some distribution.  
\end{itemize}
Although being quite similar, the relationship between security under identical-challenge cloning experiments and independent-challenge cloning experiments was not known, as all the security proofs of general copy-protection schemes~\cite{ALLZZ20,CLLZ21} were established with respect to independent-challenge distributions\footnote{With the exception of copy-protection of point functions~\cite{CMP20,AKLLZ22}.}, and their security with respect to identical-challenge distributions was not analyzed. Indeed, it turns out that the proof techniques in \cite{ALLZZ20,CLLZ21} were tailored to the independent challenge setting and they did not generalize to the identical challenge setting. 
\par We address this issue by showing the following. 

\begin{theorem}[Informal]
\label{thm:intro:cp:sde}
A copy-protection scheme secure for a class of multi-bit output functions in the independent challenge distribution setting is also secure in the identical challenge distribution setting. 
\par A single-decryptor encryption scheme in the independent challenge distribution setting is also secure in the identical challenge distribution setting. 
\end{theorem}

\noindent For the result on copy-protection, we remark that besides the fact that the output length of the functions is large (more precisely, depends on the security parameter), our result is general and applies to any class of functions. For the result of single-decryptor encryption, we consider the definition where the adversary is given the encryption of a message chosen from the uniform distribution and is supposed to predict the entire message.
\par In fact, in the technical sections, we prove a stronger theorem that generalizes for arbitrary correlated distributions instead of just identical distributions! More precisely, suppose $\distr_{\bob}$ (resp., $\distr_{\charlie}$) is the challenge distribution for $\bob$ (resp., $\charlie$). Let $\distr$ be the challenge distribution on $\bob$'s and $\charlie$'s challenge spaces, as long as the marginal distribution on $\bob$ (or $\charlie$ respectively) of $\distr$ corresponds to $\distr_{\bob}$ (or $\distr_{\charlie}$, respectively). We show that a secure copy-protection scheme when the challenge distribution is $\distr_{\bob} \times \distr_{\charlie}$, is also secure when the challenge distribution is $\distr$. Similar conclusions also hold for single-decryptor encryption schemes.

\paragraph{Encryption with Certified Deletion.} Another well-studied unclonable primitive is encryption with certified deletion \cite{BI20}. Certified deletion can be thought of as a weaker form of unclonability, where the adversary is asked to provide a classical certificate of deletion before learning the secret key.\footnote{In contrast, unclonability allows $\bob$ and $\charlie$ to both learn the secret key before passing verification. Note that we use the word "weaker" qualitatively in this sentence, and do \emph{not} claim that unclonability implies the existence of certified deletion in general.} While it is unknown whether unclonable encryption is information theoretically possible, encryption with certified deletion is known to be information theoretically possible \cite{BI20,BK22}. We give an alternate construction and proof of security of this construction is based on the techniques used for bounding monogamy-of-entanglement games \cite{TFKW13}. Our techniques are conceptually different from the existing works~\cite{BI20,BK22} who used entropic arguments to argue the same. En route, we formally define the notion of \emph{deletion games}, a subclass\footnote{Although this is not true for the initial definition we use to introduce cloning games, deletion games are captured after considering a natural extension of cloning games, where $\bob$ and $\charlie$ are not treated symmetrically.} of cloning games.

\section{Technical Overview}
\label{sec:overview}
\noindent We first discuss our definition of cloning games, why it captures many existing unclonable primitives, and then present techniques to relate different cloning games.

\newcommand{\cAg}{\cA_{\game}}
\subsection{Definitional Contribution: Cloning Games}
 
A cloning game consists of the following four procedures $(\setup, \tokengen, \gench, \ver)$: 
\begin{itemize}\setlength\itemsep{0em}
    \item A setup procedure $\setup$, on input a security parameter, outputs a secret key $\sk$. 
    
    \item A token generation procedure $\tokengen$ that takes as input the secret key $\sk$, a message $m$, and outputs a quantum state $\rho$. As we will see later, $\rho$ is expected to have some unclonability properties.

    \item A challenge generation procedure $\gench$, which takes $\sk, m$ together with random coins $r$, and outputs a challenge ${\sf ch}$. 
    
    \item Finally, a verification procedure $\ver$, that takes $\sk, m, {\sf ch}$ together with an alleged answer (which can be either a classical string or a quantum state) and outputs either 0 (reject) or 1 (accept). 
    
    We also consider another (stateful) variant where $\ver$ gets as input $r$, which are the random coins used in $\gench$. 
\end{itemize}

\noindent We require that a cloning game satisfies two properties: correctness and security. First, we discuss correctness. 

\paragraph{Correctness.} The correctness property says that there always exists an (efficient) quantum algorithm $\cAg$ if all the procedures are executed honestly and in the order of $(\setup, \tokengen, \gench, \cAg)$, the verification procedure should almost always output $1$ (accept). That is, $\cAg$ takes as input the state produced by $\tokengen$ and the challenge produced by $\gench$ and outputs an answer $\ans$ that is accepted by $\ver$ with probability negligibly close to 1. 
\ \\
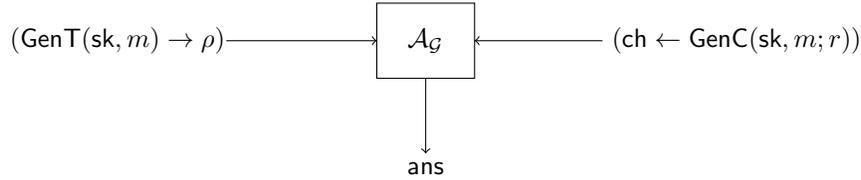
\begin{figure}[!htb]
\begin{center}
{\small
\begin{tikzpicture}
  \draw (5,-0.5) rectangle (6.3,0.5) node [pos=.5]{${\cal A}_{\game}$}; 
  \draw[<-] (6.3,0) --  (8,0) node[right]{{\small $(\ch \leftarrow \gench(\sk,m;r)$)}}; 
  \draw[->] node[right]{{\small $ (\tokengen(\sk,m) \rightarrow \rho)$}} (3,0) --  (5,0);
  \draw[->] (5.65,-0.5) -- (5.65,-1.5) node[below]{$\ans$};
 \end{tikzpicture}

\label{fig:nocloning}
}
\end{center}
\end{figure}

\subsubsection{Instantiations} 
Before we discuss security, we demonstrate the power of cloning games by showing a couple of examples. Below, we show that both unclonable encryption and copy-protection can be cast as cloning games.

\paragraph{Unclonable Encryption.} We cast unclonable encryption as a cloning game $(\setup,\tokengen,\gench,\ver)$ below. 
\begin{itemize}
    \item $\setup$ in the cloning game corresponds to the key generation of the unclonable encryption scheme. That is, $\setup$ produces the secret key, denoted by $\sk$, of the encryption scheme,
    \item $\tokengen$ corresponds to the encryption algorithm, 
    \item $\gench$ produces the challenge $ch=\sk$,
    \item $\ver$ takes as input $(\sk,m,\ch,\ans)$ and outputs 1 if and only if $\ans=m$.
    \item $\cAg$ corresponds to the decryption algorithm. On input the ciphertext state produced by $\tokengen$ and the secret key, i.e., $\ch$, it outputs the message $m$. 
\end{itemize}

\paragraph{Copy-Protection.} We can similarly cast copy-protection using cloning games, as shown below. We do not need to define $\setup$ for copy-protection of classical programs and thus $\sk = \bot$. 
\begin{itemize}\setlength\itemsep{0em}
    \item $\tokengen$ takes $sk=\bot$, message $m=f$, where $f$ is the program to be copy-protection and outputs a quantum state $\rho$. That is, $\tokengen$ corresponds to the copy-protection algorithm, 
    \item $\gench$ takes as input $sk=\bot$, $m=f$ and samples a challenge ${\sf ch}:=x$ according to the distribution.
    \item $\ver$ corresponds to the evaluation algorithm of the copy-protection scheme. That is, it takes input ${\sf ch}:=x$ and $f$, and tests whether $f(x) = \ans$. 
    \item $\cAg$ corresponds to the evaluation algorithm. On input the copy-protected state of $f$ and the challenge input $x$, it outputs $f(x)$. 
\end{itemize}
Similarly, single decryptor encryption, tokenized signatures, primitives with certified deletion, and many others can be cast as cloning games. We refer the reader to the main body for more details. 

\subsubsection{Security} 

\medskip

\noindent In the security experiment, we consider cloning adversaries of the form $\abc$. $\alice$ receives as input a quantum state $\rho$ generated using $\tokengen(\sk,m)$. $\alice$ then computes a bipartite state and sends it to $\bob$ and $\charlie$. Both $\bob$ and $\charlie$ then receive $\ch$, where $\ch$ is produced by $\gench(\sk,m)$. $\abc$ wins if $\ans_{\bob}$ produced by $\bob$ and $\ans_{\charlie}$ produced by $\charlie$ are such that $\ver$ accepts both $\ans_{\bob}$ and $\ans_{\charlie}$. 

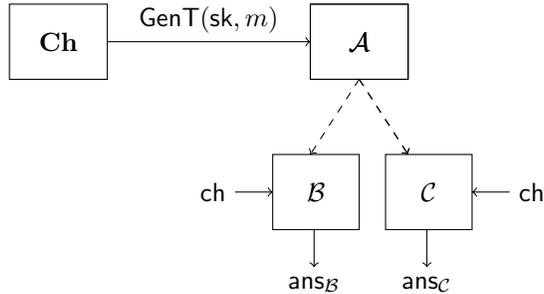
\begin{figure}[!htb]
\begin{center}
{\small
\begin{tikzpicture}
\draw (1,-0.5) rectangle (2.3,0.5) node [pos=.5]{${\bf Ch}$}; 
  \draw (5,-0.5) rectangle (6.3,0.5) node [pos=.5]{${\cal A}$}; 
  \draw[->] (2.3,0) -- node[above]{{\small $\tokengen(\sk,m)$}} (5,0);
  \draw[->,dashed] (5.65,-0.5) -- node[left]{} (5,-1.5);
  \draw[->,dashed] (5.65,-0.5) -- node[right]{} (6.3,-1.5);
  \draw (4.5,-1.5) rectangle (5.65,-2.5) node [pos=.5]{${\cal B}$};
  \draw[->] (4,-2) node[left]{$\ch$} -- (4.5,-2);
  \draw[->] (5.05,-2.5) -- (5.05,-3) node[below]{$\ans_{\cal B}$};
  \draw[->] (7.65,-2) node[right]{$\ch$} -- (7.15,-2);
  \draw[->] (6.55,-2.5) -- (6.55,-3) node[below]{$\ans_{\cal C}$};
  \draw[->,dashed] (5.65,-0.5) -- (6.3,-1.5);
  \draw (6,-1.5) rectangle (7.15,-2.5) node [pos=.5]{${\cal C}$};
  \draw (5,-0.5) rectangle (6.3,0.5) node [pos=.5]{${\cal A}$};
 \end{tikzpicture}

\label{fig:nocloning}
}
\end{center}
\end{figure}

To define the security, we first define the {\em trivial success probability} of the adversaries in the cloning game. We say that the cloning game is secure as long as any cloning adversary cannot succeed with probability significantly larger than the trivial success probability. The trivial success probability is calculated as follows: $\alice$ gives the quantum token to $\bob$, and then $\bob$ computes the correct answer $\ans_{\bob}$. On the other hand, $\charlie$ outputs its best guess $\ans_{\charlie}$. The probability that $\abc$ wins is precisely the trivial success probability.
\par The trivial success probability in an encryption scheme satisfying unclonability\footnote{Please refer to the definition of unclonability of an unclonable encryption scheme in the introduction.} is $\frac{1}{{|{\cal M}|}}$, where ${\cal M}$ is the message space, and the trivial success probability in a scheme satisfying unclonable-indistinguishability is $\frac{1}{2}$. 

\paragraph{Computational Complexity of the Attackers.} We did not remark on the computational complexity of $\abc$. In this work, we mainly work with attackers where $\alice,\bob$, and $\charlie$ are all computationally unbounded adversaries. Nevertheless, we can consider more general settings, where all of them run in quantum polynomial time. 

\paragraph{Message Distributions.} In the security experiment, $m$ is sampled from some distribution $\distr$. There are two types of distributions we consider in this work: (1) $\distr$ is uniform and, (2) $\distr_{m_0,m_1}$ is a distribution parameterized by two messages $m_0,m_1$ and it outputs $m_0$ or $m_1$ with equal probability $1/2$. When considering $\distr_{m_0, m_1}$, we allow the adversary to choose the messages $m_0, m_1$.

\paragraph{Search and Decision Games.} We consider a specific type of cloning games, called {\em search} games, where the verification algorithm $\ver$ is defined as follows: on input $(\sk,m,\ch,\ans)$, it outputs 1 (or Valid) if and only if $\ans=m$. We can consider two different security notions of search games.
\begin{itemize}
    \item {\em Unclonable-Search} security: the message distribution $\distr$ is uniform and, 
    \item {\em Unclonable-Indistinguishability} security: the message distribution is $\distr_{m_0,m_1}$, where $(m_0,m_1)$ is the pair of messages chosen by the cloning adversary. 
\end{itemize}
In the context of unclonable encryption, the above two notions correspond to unclonability and unclonable-indistinguishability properties.  
\par We also define {\em decision} games, where $\ans_{\bob},\ans_{\charlie} \in \{0,1\}$.

\paragraph{Extensions and Stateful games.}
For some applications, we need to generalize the algorithms of the cloning games further. Firstly, we can generalize the challenge generation phase to the asymmetric setting, where both $\bob$ and $\charlie$ do not necessarily receive the same challenge. This is formalized by defining an extended algorithm $\widetilde{\gench}$ which samples two random strings $r_{\bob}$ and $r_{\charlie}$ such that $\bob$ (resp., $\charlie$) receives a challenge generated using $r_{\bob}$ (resp., $r_{\charlie}$). Furthermore, we generalize the verification algorithm $\ver$ to also take as input the randomness used in the challenge generation algorithm. This way, the pair of algorithms $(\gench, \ver)$ acts as a \emph{stateful} verifier, hence the term stateful games.

\paragraph{Challenge Distributions.} Finally, we need to remark on how the randomness for the challenge generation is generated. There are two popular options:
\begin{itemize}
    \item {\em Identical challenge distribution}: in this case, $\widetilde{\gench}$ generates $r_{\bob}=r_{\charlie}$.
    \item {\em Independent challenge distribution}: in this case, $\widetilde{\gench}$ generates $r_{\bob},r_{\charlie}$ such that $r_{\bob}$ and $r_{\charlie}$ are chosen independently.  
\end{itemize}
\noindent We also consider more general challenge distributions where $r_{\bob}$ and $r_{\charlie}$ are arbitrarily correlated.

\subsection{Part I: Implications via Classical Reductions}
\label{sec:part1:overview}
\noindent In the classical cryptography literature, there is an abundance of techniques developed to show the relationship between different primitives. Ideally, we would like to draw inspiration from these techniques and/or rehash them to develop new relationships between unclonable primitives. 
\par We develop a new framework to relate cloning games using classical reductions. This new framework presents a new approach of using classical techniques to build unclonable primitives.
\par Specifically, we show that the implication of a cloning game ${\cal G}=(\setup,\tokengen,\gench,\ver)$ to another cloning game ${\cal G}'=(\setup',\tokengen',\allowbreak \gench',\allowbreak \ver')$ can be based on a classical reduction that transforms a probabilistic polynomial-time solver for one assumption into a solver for a different assumption, where the assumptions are closely related to the games $\game,\game'$. 
\par For the implication to hold, we require some extra (and mild) conditions. In the simplest case, $\setup=\setup'$ and $\tokengen'=\tokengen$. More generally, we require that the distribution of the states generated by $\tokengen$ is close (in trace distance) to the distribution of the states generated by $\tokengen'$. Additionally, we require that in both games, the trivial success probability is negligible\footnote{Our theorem is more general than what is stated here; refer to~\Cref{thm:main_lifting_reductions_to_cloning_games} for more details.}. 

\paragraph{Implications.} As a result of the above implication, we obtain two interesting sets of results. 
\par Firstly, we can show that many unclonable primitives (for instance, copy-protection and single-decryptor encryption schemes secure with respect to independent challenge distribution are also secure with respect to an arbitrary challenge distribution, as long as the marginals of the latter distribution correspond to the independent challenge distribution. This follows from the fact that changing the challenge distribution corresponds to only modifying the algorithms $\gench$ and $\ver$. 
\par Secondly, we show that any unclonable encryption scheme generically implies the existence of single-decryptor encryption. The transformation leverages the classic Goldreich-Levin technique~\cite{GL89}. The setup and token generation of single-decryptor encryption are the same as the setup and token generation of unclonable encryption. In particular, to generate the unclonable decryption key in a single-decryptor scheme, we sample a long random message $x$ and encrypt $x$ to get $\ket{{\sf ct}_x}$, using the unclonable encryption scheme. To encrypt a message $m$ in the single-decryptor scheme, one first sample random coins $r$ (of the same length as $m$) and let the ciphertext be $(r, \langle r, x \rangle \oplus m)$ together with the key to recover $x$ from the encryption $\ket{{\sf ct}_x}$. Since the setup and the token generation algorithms remain the same, and only $\gench$ and $\ver$ need to be modified, we obtain this implication. 

\subsubsection{From Classical Reductions to Reductions between Cloning Games}
We establish the relationship between cloning games using classical reductions in the following steps: 
\begin{itemize}
    \item In the first step, we define a new notion of classical reductions called classical {\em non-local} reductions. We then show that many natural classical (local) reductions can be upgraded to classical non-local reductions. 
    \item In the second step, we show how to generically lift classical non-local reductions into reductions between cloning games. Specifically, we obtain a reduction between two games ${\cal G}$ and ${\cal G}'$ such that a cloning adversary for ${\cal G}$ can be converted into a cloning adversary for ${\cal G}'$. The transformation only works in the setting when the challenge distribution associated with ${\cal G}$ corresponds to an independent challenge distribution.  
    \item In the third and final step, we show that, for any cloning game ${\cal G}$, a cloning adversary succeeding in violating the security of ${\cal G}$ with respect to an arbitrary challenge distribution ${\cal D}$ can also succeed in violating the security of ${\cal G}$ with respect to independent challenge distribution, corresponding to the marginals of ${\cal D}$.
\end{itemize}
Using the third step, we can now get an improved result in Step 2. Specifically, the reduction between ${\cal G}$ and ${\cal G}'$ holds even if the challenge distribution associated with ${\cal G}$ corresponds to an arbitrary challenge distribution, as long as the marginal distributions for $\bob$ and $\charlie$ remain the same. 
\par We remark that in the third step, we only consider cloning games with trivial success probability to be negligible. Thus, the resulting reductions between cloning games only hold for this setting. Alternately, if we start with a cloning game $\game$ with respect to the independent challenge distribution then we can still apply the first and second step to obtain a reduction between $\game$ and $\game'$ even if the trivial success probability is not negligible.

\paragraph{Step I: From Classical (Local) Reductions to Classical Non-Local Reductions.} A reduction transforms a solver for a non-interactive assumption $P$ into a solver for another non-interactive assumption $Q$. Henceforth, we refer to reductions commonly studied in the literature, as {\em local} reductions. 
\par In this work, we consider a notion of reductions called {\em non-local reductions}. First, we need to define non-local solvers. Suppose ${\mathfrak A}$ is a non-interactive assumption (for example, learning with errors).  Then, a non-local solver for ${\mathfrak A}$ consists of two algorithms $(\bob,\charlie)$ such that each of $\bob$ and $\charlie$ receives samples/challenges $\ch_{\bob},\ch_{\charlie}$ from ${\mathfrak A}$ and is supposed to solve the samples they receive. Throughout the process, $\bob$ and $\charlie$ cannot speak to each other, although they could have exchanged some common information, denoted by $\rho$, before receiving $\ch_{\bob},\ch_{\charlie}$. The samples $\ch_{\bob},\ch_{\charlie}$ can be arbitrarily correlated. We denote the distribution that samples $(\ch_{\bob},\ch_{\charlie})$ to be $\distr_{{\mathfrak A}}$. 

\par A {\em non-local} reduction is a transformation that converts a {\em non-local} $\distr_{P}$-solver $(\bob_{P},\charlie_P)$ for assumption $P$ into a {\em non-local} $\distr_{Q}$-solver $(\bob_{Q},\charlie_Q)$ for assumption $Q$, for some challenge distributions $\distr_{P}$ and $\distr_{Q}$. It turns out that we can lift local reductions into non-local reductions in the classical setting (i.e., when the solvers are classical algorithms) as long as the distribution $\distr_{Q}$ is an independent challenge distribution and the trivial success probability of $Q$ is small (for example, negligible)\footnote{One example of trivial success probability being large is non-local decision games, where $\bob$ and $\charlie$ try to produce binary answers simultaneously.}. 

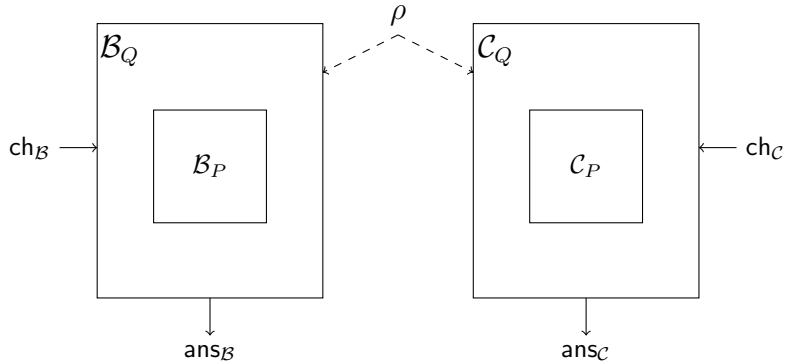
\begin{figure}[!htb]
\begin{center}
{\small
\begin{tikzpicture}
  \draw[->,dashed] (6,1) -- node[left]{} (5,0.5);
  \draw[->,dashed] (6,1) -- node[right]{} (7,0.5);
  \draw (6,1) node[above]{{\large {$\rho$}}};
  \draw (2,1.15) rectangle (5,-2.5) node [pos=.1]{{\large ${\cal B}_Q$}};
  \draw (2.75,0) rectangle (4.25,-1.5) node [pos=.5]{{${\cal B}_P$}};
  \draw (7,1.15) rectangle (10,-2.5) node [pos=.1]{{\large ${\cal C}_Q$}};
  \draw (7.75,0) rectangle (9.25,-1.5) node [pos=.5]{{${\cal C}_P$}};
  \draw[->] (1.5,-0.5) node[left]{$\ch_{\bob}$} -- (2,-0.5);
  \draw[->] (3.5,-2.5) -- (3.5,-3) node[below]{$\ans_{\cal B}$};
  \draw[->] (10.5,-0.5) node[right]{$\ch_{\charlie}$} -- (10,-0.5);
  \draw[->] (8.5,-2.5) -- (8.5,-3) node[below]{$\ans_{\cal C}$};
 \end{tikzpicture}
}
\end{center}
\label{fig:nocloning}
\end{figure}
\vspace{-2em}

\paragraph{Step II: Lifting Classical Non-Local Reductions to Reductions Between Cloning Games.} To lift classical non-local reductions to reductions between cloning games, we take inspiration from a recent work by~\cite{BBK22} (henceforth, referred to as BBK), who showed a lifting theorem that lifts classical reductions into post-quantum reductions. Suppose we would like to convert a solver for assumption $P$ into a solver for assumption $Q$. The difficulty in porting classical reductions into post-quantum reductions stems from the fact that the $Q$-solver could run the $P$-solver multiple times. Since the state of the $P$-solver could drastically change from one execution to the other (due to the difficulty of rewinding), potentially, there could no longer be any guarantees from the $P$-solver after the first execution. 
\par To solve this issue, BBK prove a novel lifting theorem in three steps.
\begin{itemize}
    \item Persistence theorem: in the first step, they show how to transform a $P$-solver into another one, where the success probability of the $P$-solver does not decrease a lot even after executing it multiple times. In other words, the $P$-solver does not lose the ability to solve instances of $P$ even after multiple executions. In more detail, suppose the $P$-solver, on input $\rho$, solved an instance of $P$ with probability $p$. Then we can convert the $P$-solver into another one, whose success probability is at least $p - \epsilon$, for some small $\epsilon > 0$, even after multiple executions. 
    
    Ideally, we would like the $P$-solver to be stateless, i.e., it does not know whether it has ever been executed in the past, in order for us to successfully reduce to the problem of solving $Q$. 
    \item From persistence to memoryless: In the next step, they convert a persistent solver into another one that is indistinguishable from a $P$-solver that is memoryless. A solver is memoryless if the only thing it can remember is the number of times it has been executed so far. 
    \item From memoryless to stateless: In the final step, they convert the solver from the second step into another solver that is indistinguishable from a stateless solver. Roughly speaking, a stateless solver is one that does not remember any information from one execution to the next. 
\end{itemize}
\newcommand{\introvalest}{{\sf ValEst}}
\newcommand{\introrepair}{{\sf Repair}}
\noindent Our strategy to lift non-local reductions into reductions between cloning games is to use the BBK approach. Similar to their work, we can define the notion of persistent, memoryless, and stateless non-local solvers. Due to some nice structural properties of their transformation, it turns out that the persistent to stateless transformation (the second and third steps above) extends directly to the non-local setting.
\par Showing the non-local version of their persistence theorem (first bullet above) requires more work. To see why, let us first recall the BBK approach to prove the persistence theorem. They use two procedures, namely value estimation ($\introvalest$) and repair ($\introrepair$) procedures, first defined by~\cite{CMSZ21}. 
\begin{itemize}
    \item $\introvalest$ has the guarantee that given an input state $\rho$ and a verification algorithm $\ver$, it outputs a number (probability) $p$ such that $\mathbb{E}[p] = p_{\sf acc}$ and $p_{{\sf acc}}$ is the probability that $\ver$ accepts $\rho$. If the output of $\introvalest$ is $p$, let the leftover state be $\rho_p$.

    \item Suppose we have computed the $P$-solver on $\rho_p$. The residual state $\rho'_p$ could be far from $\rho_p$ and more importantly, might not provide any guarantees. We would like to restore the success probability on the residual state $\rho'_p$ obtained after running the $P$-solver. The procedure $\introrepair$ does just that. It takes as input potentially disturbed state $\rho'_p$ and outputs another state $\rho^*$ such that the success probability on $\rho^*$ is close to the success probability on the original state $\rho$. 
\end{itemize}
\noindent The persistence theorem is proven as follows: each time before computing the $P$-solver, first run $\introvalest$ procedure, and then after the execution of the $P$-solver run the $\introrepair$ procedure. Roughly speaking, by the guarantees of $\introvalest$ and $\introrepair$, we have that the underlying $P$-solver does not lose its ability to solve the assumption $P$ even after executing it multiple times. \\

\noindent {\em Non-local Persistence Theorem.} Before we describe the non-local persistence theorem, we first set up some terminology. We start with a non-local {\em classical} non-local reduction which reduces a $P$-non-local solver to a $Q$-non-local solver. A $P$-solver consists of $(\bob_{P},\charlie_P)$ and is  associated with the challenge distribution $\distr_P$. On the other hand, a $Q$-solver consists of $(\bob_{Q},\charlie_{Q})$ and is associated with the challenge distribution $\distr_Q$. We want to upgrade this {\em classical} non-local reduction to the setting when both the $P$-solver and $Q$-solver can be quantum. For simplicity, we consider the case when both $\distr_P$ and $\distr_Q$ are both product distributions. 
\par We now consider a non-local version of the persistence theorem. Informally speaking, we require that the $P$-non-local solver continues to be a good solver for $P$ even after multiple executions. A natural approach to prove this theorem would be to extend the BBK approach to the non-local setting: 
\begin{itemize}
    \item Before computing the $P$-non-local solver on its state, first run $\introvalest$ procedure. 
    \item After computing on the state, run the $\introrepair$ procedure. 
\end{itemize}
Unfortunately, we do not know how to execute the above steps. The reason is that the $Q$-solver itself is non-local and hence, cannot perform any global operations on the state. However, what it can do is to alternately apply value estimation and repair procedures locally. That is, $\bob_{Q}$ (resp., $\charlie_{Q}$) applies the value estimation and repair procedures on $\bob_P$ (resp., $\charlie_P$). While this sounds promising, this leads to a new issue: we need the guarantee that the $P$-solver $(\bob_{P},\charlie_P)$ is {\em simultaneously} persistent. Even if we locally apply the procedures above on $(\bob_P,\charlie_P)$ such that both $\bob_P$ and $\charlie_P$ are persistent, this does not mean that they are simultaneously persistent! It could very well be the case when $\bob_P$ succeeds, then $\charlie_P$ does not (or vice versa), but still both of them are persistent. 
\par To address this issue, let us first consider a simple case when the state shared by ${\bob}_P$ and $\charlie_P$ are unentangled. In this special case, there is a clear relationship between the local and global value estimation and repair procedures. In particular, the following holds:
\begin{itemize}
    \item Suppose applying $(\introvalest \otimes \introvalest)$ on the shared state of $(\bob_P,\charlie_P)$ yields $(p_{\bob},p_{\charlie})$ then it holds that $\mathbb{E}[p_{\bob} \cdot p_{\charlie}]$ equals the output of the (global) $\introvalest$ on the initial state of $(\bob_P,\charlie_P)$.  
\end{itemize}
Using this, we can relate global persistence to local persistence. 
\par To generalize this to the case when the initial states of $(\bob_{P},\charlie_P)$ could be entangled, we look at the specific implementation details of the estimation procedure $\introvalest$ by~\cite{CMSZ21}. The value estimation procedure $\introvalest$ is a sequence of alternating projections, denoted by $\Pi_1,\Pi_2$, followed by a computational basis measurement determining the success probability $p$. 
\par Suppose the initial state of $(\bob_{P},\charlie_P)$ is in the Hilbert space ${\cal H}={\cal H}_{\bob} \otimes {\cal H}_{\charlie}$. We decompose both ${\cal H}_{\bob}$ and ${\cal H}_{\charlie}$ 
into subspaces that are invariant under the projections $\Pi_1,\Pi_2$ using Jordan's lemma~\cite{Jordan1875}. Therefore, we can rewrite the initial state of $(\bob_{P},\charlie_P)$ to be in the span of $\{\ket{\psi_i^{\bob}}\ket{\psi_j^{\charlie}}\}$, where $\{\ket{\psi_i^{\bob}}\}$ (resp., $\{\ket{\psi_j^{\charlie}}\}$) is in the corresponding Jordan subspaces of ${\cal H}_{\bob}$ (resp., ${\cal H}_{\charlie}$). 
\par Using an observation made by~\cite{CMSZ21}, we can think of $\introvalest$ as first performing a Jordan subspace measurement (that projects the state onto one of the Jordan subspaces) followed by performing a sequence of alternating measurements $\Pi_1,\Pi_2$. In other words, we can think of applying value estimation locally, i.e., $(\introvalest \otimes \introvalest)$, as first performing the Jordan subspace measurement to obtain a joint state $\ket{\psi_i^{\bob}}\ket{\psi^{\charlie}_{j}}$, for some $i,j$, followed by alternating measurements. Notice that once we apply the Jordan subspace measurement, the states become unentangled! Thus, we reduce to the above simple case, and the rest of the analysis follows.

\paragraph{Step III: Relating Challenge Distributions: From Independent to Identical.} In Step II, in order to be able to run the value estimation and the repair procedures locally, it was crucial that the underlying $P$-solver was defined for an {\em independent} challenge distribution. It would be interesting to generalize to the case when the underlying challenge distribution is arbitrary. For this overview, we focus on the case when the challenge distribution is identical, although the proof generalizes to arbitrary challenge distributions as well. Specifically, we demonstrate a reduction from a cloning game $\game$ satisfying unclonable security with respect to independent challenges to $\game$ satisfying unclonable security with respect to identical challenges. For the reduction to work, we crucially use the fact that the trivial success probability in both the games is negligible. An interesting point to note here is that the reduction does not change the description of the game. 

We give an overview of our reduction. 
Let $\ket{\sigma}_{\mathbf{BC}}$ be the (entangled) quantum state shared by Bob and Charlie (the two non-local quantum adversaries) after Alice's (the splitting adversary) stage. We additionally define projections $\Pi^B_r$ and $\Pi^C_r$ for every possible random coins $r$: 
\begin{description}
    \item $\Pi^B_r$: Run Bob on its own register $\sigma[\mathbf{B}]$ with the challenge corresponding to random coins $r$, project onto Bob's output being accepted and uncompute;
    \item $\Pi^C_r$: Run Charlie on its own register $\sigma[\mathbf{C}]$ with the challenge corresponding to random coins $r$, project onto Charlie's output being accepted and uncompute;
\end{description}

By definition, the success probability in the independent challenge case is:
\begin{align} \label{eq:prob_ind}
    \delta = {\sf Tr}\left[ \left(\frac{1}{|R|}\sum_r \Pi^B_r \right) \otimes \left( \frac{1}{|R|} \sum_r \Pi^C_r \right) \ket{\sigma}\bra{\sigma} \right],
\end{align}
where $R$ is the random coin space. 

Since $\Pi^B := \frac{1}{|R|} \sum_r \Pi^B_r$ is a POVM, let $\{\ket{\phi_p}\}_{p \in \mathbb{R}}$ be the set of eigenvectors with eigenvalues $p \in [0,1]$\footnote{There can be multiple eigenvectors with the same eigenvalues. In the overview, we assume that eigenvalues are unique.}. Similarly, let $\{\ket{\psi_q}\}_{q \in \mathbb{R}}$ be the set of eigenvectors with eigenvalues $q \in [0,1]$ for $\Pi^C := \frac{1}{|R|} \sum_r \Pi^C_r$. Therefore, we can always write $\ket{\sigma}_{\mathbf{BC}}$ under the bases $\{\ket{\phi_p}\}$ and $\{\ket{\psi_q}\}$\footnote{There is a one-to-one mapping between $\{\ket{\phi_p},\{\ket{\psi_q}\}$ and the vectors $\{\ket{\psi^\bob_i},\{\ket{\psi^\charlie_j}\}$ defined in the Jordan's lemma.}:
\begin{align*}
    \ket{\sigma} = \sum_{p, q} \alpha_{p, q} \ket{\phi_p} \ket{\psi_q}. 
\end{align*}
From the above decomposition of $\ket{\sigma}$ and \Cref{eq:prob_ind}, we have $\delta = \sum_{p, q} |\alpha_{p, q}|^2 p q$. 

\medskip

Let $\eta \in [0,1]$ be a threshold we will pick later. The quantum state can be written as the summation of three terms: 
\begin{align*}
    \ket{\sigma} = \sum_{q < \eta} \alpha_{p, q} \ket{\phi_p} \ket{\psi_q} + \sum_{p < \eta, q > \eta} \alpha_{p, q} \ket{\phi_p} \ket{\psi_q}  + \sum_{p > \eta, q > \eta} \alpha_{p, q} \ket{\phi_p} \ket{\psi_q}.
\end{align*}
We denote the first term by $\ket{\sigma_\cC}$, indicating that Charlie's success probability is bounded by $\eta$; the second term by $\ket{\sigma_\cB}$, indicating that Bob's success probability is bounded by $\eta$; and the last term by $\ket{\rho}$, none of the probabilities is below $\eta$. Thus, $\ket{\sigma} = \ket{\sigma_\cC} + \ket{\sigma_\cB} + \ket{\rho}$. 

First, we note that the success probability when executed on the state $\ket{\sigma_\cC} + \ket{\sigma_\cB}$ is at most $\eta$ under both independent challenges and identical challenges. However, $\frac{\ket \rho}{|| \ket{ \rho} ||^2}$ could be such that the success probability when executed on this state maybe large (even as large as $1$). In the next step, we show that although $\rho$ may have a large probability under identical challenges, $\|\ket\rho\|^2$ is relatively small. Because $\delta := \sum_{p, q} |\alpha_{p, q}|^2 p q$, we have: 
\begin{align*}
    \|\ket\rho\|^2 = \sum_{p > \eta, q > \eta} |\alpha_{p, q}|^2 
    \quad \Longrightarrow \quad \|\ket\rho\|^2 \leq \delta/\eta^2. 
\end{align*}

\noindent Therefore, the success probability of $\ket{\sigma}$ under identical challenges is: 
\begin{align*}
    \delta &= \frac{1}{|R|} \sum_r \left\| \Pi^B_r \otimes \Pi^C_r \ket{\sigma} \right\|^2 \\
    & \leq \frac{3}{|R|} \sum_r \left(  \left\| \Pi^B_r \otimes \Pi^C_r \ket{\sigma_\cC} \right\|^2  + \left\| \Pi^B_r \otimes \Pi^C_r \ket{\sigma_\cB} \right\|^2   + \|\rho\|^2  \right)  \\
    &\leq 3 \left( \eta + \delta / \eta^2 \right).
\end{align*}
By picking $\eta = \delta^{1/3}$, the resulting probability is $6 \cdot \delta^{1/3}$. Specifically, if $\delta$ is negligible then so is the resulting quantity.

\subsection{Part II: Generalizing Existing Results}
\label{sec:unclsrchtounclind}

\subsubsection{Unclonable search to unclonable indistinguishability}

Our first focus is a cloning game with unclonable search security (a concrete example is unclonable encryption with standard unclonable security) whose distribution $\cD$ is uniform over all possible messages. We show a generic reduction that turns such a game into another cloning game with unclonable indistinguishability security whose underlying distribution is $\cD_{m_0,m_1}$ for any $m_0, m_1$ in the message space, in the quantum random oracle model (QROM, introduced by \cite{boneh2011random}). Since unclonable encryption with standard unclonable security exists \cite{BL20}, this gives a direct corollary for unclonable encryption with unclonable indistinguishability security in the QROM, from BB84/Wiesner states, improving the previous result by \cite{AKLLZ22}. 

\paragraph{Unclonable Security for High-Entropy Message Distributions.} As a first step in the reduction, we make the following observation. Suppose we start with a cloning game satisfying unclonable security. If the message is sampled from a high min-entropy distribution instead of being sampled from random, unclonable security still holds. In particular, we prove that when $m$ is sampled from a source with min-entropy $h$ instead of from a uniform source, its unclonable search security will be $2^h \cdot \delta$; where $\delta$ is the unclonable search security under the uniform message distribution. For instance, if $m$ is sampled uniformly at random from a set $S$ then by appropriately choosing $|S|$ (for example, it is exponential sized), we can prove that $2^h \delta$ is still negligible and thus establish its augmented unclonable security. 
\par As a concrete example, we obtain the following corollary: the unclonable encryption with standard unclonable security in \cite{BL20} also satisfies this augmented unclonable security. In other words, even if $P_m$ is provided as oracle, it will not help Bob and Charlie to simultaneously recover $m$. 

\paragraph{Augmented Security.} Next, we first define stronger unclonable search security, which we call {\em augmented unclonable security}. The cloning game is defined in the same way, except now all attackers have oracle access to a point function $P_m(\cdot)$, which outputs $1$ if and only if the input equals to $m$, where $m$ is the message used to generate the token given to the adversary. 
We claim that the definition of unclonable search security can be generically upgraded to obtain augmented security.

Our first observation is that, we can enlarge the set of all accepting inputs of $P_m(\cdot)$ (originally only $m$) to a large random set $S$ consisting of $m$, with its security staying roughly the same. More concretely, $S$ will be defined as an exponentially large (but negligibly small compared to the number of all possible messages) set consisting of a single $m$, and the rest are random messages. As $P_m$ and $P_S$ only differ on exponentially many but sparse random inputs, query-bounded adversaries can not distinguish which oracle is given. 

Next, the augmented unclonable security is then argued under a random message $m$ and oracle access to $P_S$. In this case, we can instead think of an alternate but equivalent process of sampling $m$: first sample an exponentially large random set $S$ then sample a message $m$ is uniformly at random from $S$. After changing the sampling order, we can argue that even if the adversary is given the description of the set $S$, unclonable security still holds. This holds from our earlier observation that unclonable security holds even if the message is sampled from a high min-entropy distribution.

\paragraph{From Augmented Security to Unclonable-Indistinguishability Security.} Finally, we show that starting from a cloning game satisfying augmented unclonability property $\game$, we can obtain a game $\game'$ satisfying unclonable indistinguishability property. The token generation of $\game'$ on input a message $m$, first samples a long random message $x$, runs the token generation of $\game$ on $x$ and then outputs this token along with $H(x) \oplus m$, where $H$ is a hash function.  In the proof of security, $H$ is treated as a random hash function that the adversary has oracle access to. 
\par To prove unclonable indistinguishability, we look at the state $\ket{\psi}_{\bob \charlie}$ output by Alice, where Alice has oracle access to $H$ punctured at the input $x$. For the sake of the proof, we treat the hash function both Bob and Charlie have access to, separately. We use $H_{\bob}$ to denote the hash function Bob has access to and $H_{\charlie}$ to denote the hash function Charlie has access to. Correspondingly, we can define the POVM $\Pi^{\bob}$ that runs $\bob$ with oracle access to $H_{\bob}$ that is programmed on $x$ to output $0$ or $1$ with equal probability, projects onto the output being correct and then uncomputes. Similarly, we define $\Pi^{\charlie}$ as well. In order to make sure we can implement $\Pi^{\bob}$ and $\Pi^{\charlie}$ efficiently, we give the adversary oracle access to $P_x(\cdot)$. 
\par Let $\{\ket{\phi_p}\}_{p \in \mathbb{R}}$ be the set of eigenvectors with respect to $\Pi^{\bob}$ with eigenvalues $p \in [0,1]$. Similarly, let $\{\ket{\psi_q}\}_{q \in \mathbb{R}}$ be the set of eigenvectors with eigenvalues $q \in [0,1]$ with respect to $\Pi^C$.  We can then rewrite $\ket{\psi}_{\bob \charlie}$ in terms of the eigenbases of $\Pi^{\bob}$ and $\Pi^{\charlie}$. 
\begin{align*}
    \ket{\psi} = \sum_{q \approx 0.5} \alpha_{p, q} \ket{\phi_p} \ket{\psi_q} + \sum_{p \approx 0.5, |q - 0.5| \gg 0} \alpha_{p, q} \ket{\phi_p} \ket{\psi_q}  + \sum_{|p - 0.5| \gg 0,|q - 0.5| \gg 0} \alpha_{p, q} \ket{\phi_p} \ket{\psi_q}.
\end{align*}
Once we do this, we show the following:
\begin{itemize}
\item $\sum_{p,q} |\alpha_{p,q}|^2$ is negligible. We show this by reducing to the unclonability property. 
\item Once we show bullet 1, we can then show that Bob and Charlie cannot simultaneously succeed with probability significantly better than 0.5 in the case when it receives as input $\ket{\psi}$. The analysis of this was shown in~\cite{AKLLZ22}. 
\end{itemize}

\subsection{Generalized Cloning Games}
Another way we can extend the notion of cloning games is by allowing asymmetric verification for $\bob$ and $\charlie$ by allowing different algorithms $(\gench_\bob, \gench_\charlie, \ver_\bob, \ver_\charlie)$ in the verification phase. We call this more general class of games \emph{asymmetric cloning games}.

\subsubsection{Deletion Games} 
 As a special case, we define \emph{deletion games}, in which $\gench_\bob$ outputs no challenge, so that $\bob$ is effectively supposed to produce a classical certificate of deletion. In this context, we can define search games based on the algorithms $(\gench_\charlie, \ver_\charlie)$, with $\charlie$ understood to be the party tasked to perform the intended functionality of the token. With these modifications, unclonable search security and unclonable indistinguishable security are defined the same as before. We show how to how to go from the former to the latter using the Quantum Goldreich-Levin Lemma\footnote{Unlike our result on single-decryptor encryption, which asks for the usual, stronger property of unclonability, here we do not need the simultaneous version of the Goldreich-Levin Lemma because we are in the weaker, certified deletion setting.} \cite{AC02}. In order to achieve unclonable search security, we show that the construction\footnote{A simplified version of it without additional properties. The authors show in \cite{BI20} that the construction already satisfies the stronger notion of unclonable indistinguishable security, yet the proof is more involved.} of \cite{BI20} satisfies unclonable search security using monogamy-of-entanglement games \cite{TFKW13}, which have been commonly used in unclonable cryptography \cite{BL20,CLLZ21,CV21}. Specifically, we show that the success probability of any adversary $\abc$ in the following game is exponentially small in $\secparam$: \begin{itemize}
    \item $\abc$ prepares a bipartite state $\rho$ shared between $\alice$ and the referee $\referee$. $\alice$ splits the state between $\bob$ and $\charlie$.
    \item $\referee$ makes a measurement in basis $H^\theta$ for a random $\theta \in \bit^\secparam$
    \item $\bob$ outputs $x_\bob$. $\charlie$ receives $\theta$ and outputs $x_\charlie$.
    \item $\abc$ wins if $x_C = x$ and $x_{B,i} = x_i$ whenever $\theta_i = 1$.
\end{itemize}

\noindent This suffices due to a well-known reduction from cloning games to monogamy-of-entanglement games using EPR pairs. Thus, we provide a different method to achieve information theoretic encryption with certified deletion. Although our method is incomparable to previous methods for achieving the same result \cite{BI20,BK22}, we believe our approach may be more intuitive for some readers. 
\subsubsection{Relating Search and Decision Games}
We give one more transformation, which starts with a cloning search game and ends up with a cloning decision game. The first one uses augmented security above and applies it to the construction of \cite{AKLLZ22} for copy-protection in the QROM. We generalize the proof for a class of cloning games, and as a special case, we achieve copy-protection for point functions using BB84 states in the QROM. Since the ideas employed in this part are similar to~\Cref{sec:unclsrchtounclind}, we omit the details. 

\subsection{Future Directions} \label{sec:open_problems}

\paragraph{Relationship Between Challenge Distributions for Decision Games.} In this work, we show that when a cloning game has negligible soundness (similar to a search game) with an independent distribution, the cloning game with the corresponding identical distribution is also secure. We leverage this theorem to many applications, including copy-protection and single-decryptor encryption schemes. However, this theorem does not apply when the soundness is a constant. An interesting open problem is to generalize the result to the case with constant security error (for example, unclonable-indistinguishability). Generalizing this result would present a pathway towards achieving unclonable encryption scheme with unclonable-indistinguishability in the plain model, that is currently open. 

\paragraph{Removing Random Oracles from BB84-based Constructions.} Another approach to obtain unclonable encryption with unclonable-indistinguishability in the plain model is to remove the need for random oracle in the \Cref{thm:BB84_construction}. Currently, the random oracle is essential, and we do not know how to get rid of it. Still, we believe that removing QROM in the theorem statement is a promising direction and will help us understand the relationship between various unclonable primitives and the computational assumptions they need. 

\paragraph{Domain Extension.} Suppose we have a cloning game for messages of $n$ bits. Is it possible to generically transform this into another cloning game for $2n$ bits?  Naive repetition does not work well with cloning games and hence, it would be interesting to come up with interesting techniques for domain extension. One application of this is domain extension for unclonable encryption. Suppose we have an encryption scheme that can encrypt $n$ bits and we would like to transform this into a different scheme encrypting $2n$ bits. It would also be interesting to study domain extension for the challenge space as well. This would have implications to domain extension for single-decryptor encryption.

\paragraph{Generalizing the Non-Local Lifting Theorem} Our non-local lifting theorem \Cref{thm:nl_clas_to_qtm_lifting} is restricted in that the classical reductions need to be black-box and non-adaptive. These restrictions propagate from the work of \cite{BBK22}, and removing them will allow for more classical reductions to be lifted to the quantum setting.

\subsection{Organization}

In \Cref{sec:prelims}, we define basic terminology and give relevant results from previous work that are used in this work. In \Cref{sec:def_cg}, we introduce the cloning games framework and provide formal definitions. In \Cref{sec:nonlocal_rd}, we describe how to lift classical reductions to quantum reductions in the non-local setting, as well as applications to cloning games. In \Cref{sec:search_to_indist,sec:search_to_dec}, we show how to use our framework to generalize existing results in unclonable cryptography via generic transformations between cloning search/decision games. In \Cref{sec:asym_cg}, we show how to extend our framework to capture asymmetric unclonable primitives, and give an alternative construction of unclonable encryption with certified deletion.

\paragraph{Acknowledgments.} PA and FK are supported by a gift from Cisco.

\newcommand{\Hs}{\cH}
\newcommand{\weight}{w}
\newcommand{\trivguess}{\mathsf{OPT}}
\newcommand{\jord}{\mathsf{Jor}}

\section{Preliminaries} \label{sec:prelims}
We denote the security parameter by $\lambda$. We say a classical algorithm $\alice$ is efficient if it is a probabilistic polynomial-time (PPT) algorithm. We write $\alice(x;r)$ to mean that $\alice$ runs on input $x$ with random coins $r \in \bit^{\poly(\secparam)}$. We say a quantum algorithm $\alice$ is efficient if it runs in quantum polynomial time (QPT). 
\par We write $\alice^\cO$ to denote an oracle algorithm $\alice$ that makes queries to an oracle $\cO$. If $\alice$ is a quantum algorithm and $\cO$ is a classical oracle, then it is understood that $\alice$ can make superposition queries. We call $\alice$ \emph{query-bounded} if it makes only polynomially many queries.
\par We denote by $\cU_X$ the uniform distribution over a set $X$. $\negl(\cdot)$ denotes a negligible function and $\poly(\cdot)$ is a function upper-bounded by a polynomial. We say that an event occurs with \emph{overwhelming probability} if it happens with probability $1 - \negl(\lambda)$.
\par We denote by $P_x(\cdot)$ the point function, defined as \begin{align*}
        P_x(x') := \begin{cases} 1, \quad x' = x \\ 0, \quad x' \ne x \end{cases}.
    \end{align*}

\paragraph{Trivial Guess with Auxiliary Information:} We define \begin{align*}
    \trivguess(X \given Y) := \sum_{y} \pr{Y = y} \cdot \max_{x} \pr{X = x \given Y = y}.
\end{align*}
It is the optimal probability of guessing the value of variable $X$ after observing the value of variable $Y$. Clearly, $\trivguess(X \given Y) \ge \trivguess(X \given 0) = 2^{-H_{\min}(X)}$, where $H_{\min}$ denotes the min-entropy, defined as follows: \begin{align*}
    H_{\min}(X) := - \log_2 \bracS{\max_x \pr{X = x} }
\end{align*}

\subsection{Quantum Computing Basics \& Query Bounds}

\par Given Hilbert space $\cH$, $\cD(\cH)$ denotes the set of density operators on $\cH$. We write $\cH_X$ to denote the Hilbert space associated with a quantum register $X$. We write $\rho[X]$ to denote the $X$ register of a quantum state $\rho$. Given two quantum states $\rho, \sigma$, we denote the (normalized) trace distance between them by \begin{align*}
    \tracedist{\rho}{\sigma} := \frac{1}{2}\left\|\rho - \sigma\right\|_{\mathsf{tr}}.
\end{align*} 

We say that two states $\rho, \sigma$ are \textit{$\delta$-close} if $\tracedist{\rho}{\sigma} \le \delta$. \\

\par A $k$-outcome (general) \emph{quantum measurement} is a $k$-tuple of quantum operators $\cM = \brac{M_i}_{i \in [k]}$ satisfying $\sum_{i \in [k]} M_i^\dagger M_i = I$. The probability of obtaining outcome $i$ after measuring a mixed state $\rho$ is given by $\tr \brac{ M_i \rho M_i^\dagger }$, and the post-measurement state is given by $M_i \rho M_i^\dagger / \tr \brac{M_i \rho M_i^\dagger}$. If $M_i$ is a projector for each $i \in [k]$, then we call $\cM$ a projective measurement. A \emph{positive operator valued measurement} (POVM) is $k$-tuple of positive semi-definite operators $\brac{E_i}_{i \in [k]}$. It is used to describe quantum measurement when the post-measurement state is irrelevant in the context. The probability of outcome $i$ equals $\tr \brac{ E_i \rho }$. Every general quantum measurement defines a POVM by setting $E_i = M_i^\dagger M_i$.

\par A common technique in quantum computation is \emph{uncomputing}~\cite{BBBV97}. A quantum algorithm $\alice$ can be modeled as a unitary $U$ acting on some hilbert space $\Hs$, followed by a measurement on output registers without loss of generality. We refer to applying $U$ as running $\alice$ \emph{coherently}, and to uncomputing $\alice$ as applying $U^\dag$ on $\Hs$.

\paragraph{Quantum Random Oracle Model}

In the quantum random oracle model (QROM), we assume that there exists a random function $H : \bit^m \to \bit^n$, where $m = \poly(\secparam), n = \poly(\secparam)$, such that all parties (honest and malicious) have oracle access to the unitary $\cO^H$, defined as $\cO^H \ket{x}\ket{y} = \ket{x}\ket{y \oplus H(x)}$. Such a random oracle $H$ can be efficiently simulated on the fly for a query-bounded adversary \cite{Zha19}, or if a query bound $t$ was known beforehand, it can be simulated efficiently by using a $2t$-wise independent hash function in lieu of the random oracle \cite{Zha12}.

\par The following theorem, paraphrased from \cite{BBBV97}, will be used for reprogramming oracles without adversarial detection on inputs that are not queried with large weight:

\begin{theorem}[\cite{BBBV97}] \label{thm:bbbv}
 Let $\alice$ be an oracle algorithm which makes at most $T$ oracle queries to a function $H: \bit^m \to \bit^n$ . Define $\ket{\phi_i}$ as the global state after $\alice$ makes $i$ queries, and $W_y(\ket{\phi_i})$ as the sum of squared amplitudes in $\ket{\phi_i}$ of terms in which $\alice$ queries $H$ on input $y$. Let $\epsilon>0$ and let $F \subseteq \bracC{0,1,\dots,T-1} \times \bit^m$ be a set of time-input pairs such that $\sum_{(i,y) \in F} W_y(\ket{\phi_i}) \le \epsilon^2 / T $.
 \par For $i \in \bracC{0,1,\dots,T-1}$, let $H_i'$ be an oracle obtained by reprogramming $H$ on inputs in $\bracC{y \in \bit^m \; : \; (i,y) \in F}$ to arbitrary outputs. Let $\ket{\phi_T'}$ be the global state after $\alice$ is run with oracle $H_i'$ on the $i$th query (instead of $H$). Then, $\tracedist{\ket{\phi_T}}{\ket{\phi_{T}'}} \le \epsilon / 2$.
\end{theorem}

Note that the theorem can be straightforwardly generalized to mixed states by convexity. \\

\par We will typically use \Cref{thm:bbbv} by contrapositive, i.e., if a query-bounded adversary $\alice$ outputs states with non-negligible trace distance when given oracle access to $H$ or $H'$, then $\alice$ must have non-negligible query weight on inputs for which $H$ and $H'$ differ. Hence one can extract such input by measuring a random query. We list a particular corollary of interest below.

\begin{corollary}[Subset Hiding] \label{cor:subset_hiding}
Let $S \subset \bit^\secparam$. Let $m > |S|$ such that $m / (2^\secparam - |S|)$ is negligible. Let $T_S^m$ be the set of all $T \subset \bit^\secparam$ such that $|T| = m + |S|$ and $S \subset T$. Then, for any query-bounded algorithm $\alice$, we have \begin{align}
    \abs{ \pr{ 1 \from \alice^{P_S(\cdot)}(1^\secparam) } - \pr{ 1 \from \alice^{P_T(\cdot)}(1^\secparam) \; : \; T \uniform T_S^m } } \le \negl(\secparam). \label{eq:4}
\end{align}

\end{corollary}

\begin{proof}
    Suppose $\alice$ violates \cref{eq:4}. By \Cref{thm:bbbv}, $\alice^{P_S(\cdot)}$ must have non-negligible query weight on $T \setminus S$, which is a random subset of $\bit^\secparam \setminus S$ of size $m$. By measuring a random query of $\alice$, a query-bounded adversary $\alice'^{P_S(\cdot)}$ can output a value $t \in T \setminus S$ with non-negligible probability. However, $T \setminus S$ is information-theoretically hidden from $\alice^{P_S(\cdot)}$, this probability is upper-bounded by $m/(2^{\secparam} - |S|)$, a contradiction.
\end{proof}

\subsection{Jordan's Lemma} \label{sec:prelims_jorlem}
We state Jordan's Lemma, paraphrased from \cite{CMSZ21}.

\begin{lemma}[\cite{Jordan1875}] \label{lem:jordans_lemma}
    Let $\cH$ be a Hilbert space and let $\Pi_A, \Pi_B$ be two orthogonal projectors on $\cH$. There exists an orthogonal decomposition $\cH = \bigoplus_j S_j$ into one-dimensional or two-dimensional subspaces, where each $S_j$ is invariant under both $\Pi_A$ and $\Pi_B$. Moreover: \begin{itemize}
        \item If $\dim S_j = 1$, then $\Pi_A$ and $\Pi_B$ act as rank-0 or rank-1 projectors on $S_j$.
        \item If $\dim S_j = 2$, then there exist distinct orthogonal bases $\bracC{\ket{v^A_{j,1}}, \ket{v^A_{j,0}}}$ and $\bracC{\ket{v^B_{j,1}}, \ket{v^B_{j,0}}}$ of $S_j$,
        where $\Pi_A = \sum_j \ketbraX{v^A_{j,1}}$ and $\Pi_B = \sum_j \ketbraX{v^B_{j,1}}$.
    \end{itemize}
\end{lemma}

We will denote by $\Pi^{\jord} = \bracC{\Pi^{\jord}_j}_{j}$ the projective measurement that measure the index $j$ of $S_j$, i.e. $\Pi^{\jord}_j = \ketbraX{v^A_{j,0}} + \ketbraX{v^A_{j,1}} = \ketbraX{v^B_{j,0}} + \ketbraX{v^B_{j,1}}$. An important fact is that $\Pi^{\jord}$ commutes with both the projective measurements $(\Pi_A, I - \Pi_A)$ and $(\Pi_B, I - \Pi_B)$.

\noindent We cite the following lemma from \cite{AKLLZ22}, which is a corollary of Jordan's Lemma.
\begin{lemma} \label{lem:jordans}
    For any two projectors $\Pi_0, \Pi_1$ and $w \in [0,1]$, let $\ket{\phi_0}$ and $\ket{\phi_1}$ be two eigenvectors of $\weight\Pi_0 + (1 - \weight)\Pi_1$ with eigenvalues $\lambda_0, \lambda_1$. If $\lambda_0 + \lambda_1 \ne 1$ and $\lambda_0 \ne \lambda_1$, then 
    \begin{align*}
        \langle \phi_0 |\Pi_0|\phi_1 \rangle = \langle \phi_0 |\Pi_1|\phi_1 \rangle = 0. 
    \end{align*}
\end{lemma}

\subsection{Applications of Jordan's Lemma}

In this section, we state two applications of Jordan's lemma~\cite{Jordan1875} which use the techniques of \cite{Zha20}, and date back to the QMA amplification techniques of \cite{MW05}. While they use similar techniques, the applications are different in their syntax and flavor. Threshold measurement (first application) aims to project a state onto eigenstates with eigenvalues larger than some threshold, whereas state repair (second application) involves estimating the average eigenvalue, and restoring the state to another state with similar value after a collapsing measurement occurs.

\paragraph{Threshold Measurement}

We cite the following theorems regarding how to test the success probability of a quantum token from \cite{AKLLZ22}, originally due to \cite{Zha20}. 

\begin{theorem}[Inefficient Threshold Measurement] \label{thm:inefficient_threshold_measure}
    Let $\cP = (P, Q)$ be a binary outcome POVM. Let $P$ have eigenbasis $\{\ket {\psi_i}\}$ with eigenvalues $\{\lambda_i\}$. Then, for every $\gamma \in (0,1)$ there exists a projective measurement $\cE_{\gamma} = (E_{\leq \gamma}, E_{> \gamma})$ such that:
    \begin{itemize}
        \item[(1)] ${E}_{\leq \gamma}$ projects a quantum state into the subspace spanned by $\{\ket{\psi_i}\}$ whose eigenvalues $\lambda_i$ satisfy $\lambda_i \leq \gamma$; 
        \item[(2)] ${E}_{> \gamma}$ projects a quantum state into the subspace spanned by $\{\ket{\psi_i}\}$ whose eigenvalues $\lambda_i$ satisfy $\lambda_i > \gamma$.
    \end{itemize}
    
    Similarly, for every $\gamma \in (0, 1/2)$, there exists a projective measurement ${\cE'}_\gamma = (\widetilde{E}_{\leq \gamma}, \widetilde{E}_{> \gamma})$ such that:
    \begin{itemize}
        \item[(1)] $\widetilde{E}_{\leq \gamma}$ projects a quantum state into the subspace spanned by $\{\ket{\psi_i}\}$ whose eigenvalues $\lambda_i$ satisfy $|\lambda_i - \frac{1}{2}| \leq \gamma$; 
        \item[(2)] $\widetilde{E}_{> \gamma}$ projects a quantum state into the subspace spanned by $\{\ket{\psi_i}\}$ whose eigenvalues $\lambda_i$ satisfy $|\lambda_i - \frac{1}{2}| > \gamma$.
    \end{itemize}
\end{theorem}

\begin{theorem}[Efficient Threshold Measurement] \label{thm:thres_approx_asymmetric}
    Let $\cP_b = (P_b,Q_b)$ be a binary outcome POVM over Hilbert space $\Hs_b$ that is a mixture of projective measurements for $b \in \{1,2\}$. Let $P_b$ have eigenbasis $\{\ket {\psi_i^b}\}$ with eigenvalues $\{\lambda_i^b\}$. For every $\gamma_1, \gamma_2 \in (0, 1), 0 < \epsilon < \min(\gamma_1/2, \gamma_2/2, 1 - \gamma_1, 1 - \gamma_2)$ and $\delta > 0$, there exist efficient binary-outcome quantum algorithms, interpreted as the POVM element corresponding to outcome 1, $\ati_{\cP_b, \gamma}^{\epsilon, \delta}$ such that for every quantum program $\rho \in \density{1} \otimes \density{2}$ the following are true about the product algorithm $\ati_{\cP_1, \gamma_1}^{\epsilon, \delta} \otimes \ati_{\cP_2, \gamma_2}^{\epsilon, \delta}$:
    \begin{itemize}
        \item[(0)] Let $ (E^b_{\leq \gamma}, E^b_{> \gamma})$ be the inefficient threshold measurement in \Cref{thm:inefficient_threshold_measure} for $\Hs_b$. 
        \item[(1)] The probability of measuring 1 on both registers satisfies $$\tr\bracS{\brac{ \ati_{\cP_1, \gamma_1}^{\epsilon, \delta} \otimes \ati_{\cP_2, \gamma_2}^{\epsilon, \delta} } \rho } \ge \tr\bracS{ \brac{ E^1_{> \gamma_1 + \epsilon} \otimes E^2_{> \gamma_2 + \epsilon} } \cdot \rho} - 2\delta.$$
        \item[(2)] The post-measurement state $\rho'$ after getting outcome (1,1) is $4\delta$-close to a state in the support of $\bracC{\ket{\psi^1_i}\ket{\psi^2_j}}$ such that $\lambda_i^1 > \gamma_1 - 2\epsilon$ and $\lambda^2_j > \gamma_2 - 2\epsilon$. 
        
        \item[(3)] The running time of the algorithm is polynomial in the running time of $P_1, P_2$, ${1}/{\epsilon}$ and $\log(1/\delta)$. 
    \end{itemize}
\end{theorem}

\begin{theorem}[Efficient Symmetric Threshold Measurement] \label{thm:thres_approx}
    Let $\cP_b = (P_b,Q_b)$ be a binary outcome POVM over Hilbert space $\cH_b$ that is a mixture of projective measurements for $b \in \{1,2\}$. Let $P_b$ have eigenbasis $\{\ket {\psi_i^b}\}$ with eigenvalues $\{\lambda_i^b\}$. For every $\gamma_1, \gamma_2 \in (0, 1/2), 0 < \epsilon < \min(\gamma_1/2,\gamma_2/2)$, and $ \delta > 0$, there exist efficient binary-outcome quantum algorithms, interpreted as the POVM element corresponding to outcome 1, $\sati_{\cP_b, \gamma}^{\epsilon, \delta}$ such that for every quantum program $\rho \in \density{1} \otimes \density{2}$ the following are true about the product algorithm $\sati_{\cP_1, \gamma_1}^{\epsilon, \delta} \otimes \sati_{\cP_2, \gamma_2}^{\epsilon, \delta}$:
    \begin{itemize}
        \item[(0)] Let $ (\widetilde{E}^b_{\leq \gamma_b}, \widetilde{E}^b_{> \gamma_b})$ be the inefficient threshold measurement in \Cref{thm:inefficient_threshold_measure} for $\cH_b$. 
        \item[(1)] The probability of measuring 1 on both registers satisfies $$\tr\bracS{\brac{ \sati_{\cP_1, \gamma_1}^{\epsilon, \delta} \otimes \sati_{\cP_2, \gamma_2}^{\epsilon, \delta} } \rho } \ge \tr\bracS{ \brac{ \widetilde{E}^1_{> \gamma_1 + \epsilon} \otimes \widetilde{E}^2_{> \gamma_2 + \epsilon} } \cdot \rho} - 2\delta.$$

        \item[(2)] The post-measurement state $\rho'$ after getting outcome (1,1) is $4\delta$-close to a state in the support of $\bracC{\ket{\psi^1_i}\ket{\psi^2_j}}$ such that $|\lambda_i^1 - 1/2| > \gamma_1 - 2\epsilon$ and $|\lambda^2_j - 1/2| > \gamma_2 - 2\epsilon$.

        \item[(3)] The running time of the algorithm is polynomial in the running time of $P_1, P_2$, ${1}/{\epsilon}$ and $\log(1/\delta)$. 
    \end{itemize}
\end{theorem}

\paragraph{State Repair.}
We state the state-repairing results (Lemma 4.9 and Lemma 4.10) from \cite{CMSZ21} below to be used for achieving persistence. We adapt the formulation by~\cite{BBK22} with some additional modifications: unlike~\cite{BBK22}, we need some additional structural properties of value estimation and state repair procedures of~\cite{CMSZ21} formalized in the third bullet below.

\begin{lemma} \label{lem:cmsz}
\label{lem:cmsz} There exist efficient quantum algorithms $\valest$ and $\repair$ with the following syntax and guarantees: 
\begin{itemize}
    \item $\valest_{V, \alice}(\bfrho, 1^{1/\varepsilon})$ takes as input the description of a verifier algorithm $V : \bit^d \times \bit^n \to \bit$, a quantum algorithm $\alice$, a quantum state $\bfrho$, and an accuracy parameter $\varepsilon$. It outputs a quantum state $\bfrho^*$ and a value $p^* \in [0,1]$.
    \item $\repair_{V,\alice,\Pi}(\bfsigma, y, p, 1^{1/\varepsilon}, 1^k)$ takes as input a verifier algorithm $V$, a quantum algorithm $\alice$, a $k$-outcome quantum measurement $\Pi$ with outcomes $\cY$, an outcome $y \in \cY$, a value $p \in [0,1]$, and an accuracy parameter $\varepsilon$. It outputs a quantum state $\bfsigma^*$.
\end{itemize}

\begin{enumerate}
    \item {\bf Value Estimation:} \begin{align} \label{eq:valest}
        \E_{(\bfrho^*, p^*) \from \valest_{V, \alice}(\bfrho, 1^{1/\varepsilon})} \bracS{ p^* } = \pr{ V(y; r) = 1 \; : \; \substack{ r \uniform \bit^d \\ y \from \alice(\bfrho, r) } }.
    \end{align}
    \item {\bf Almost-Projective Estimation:} For any $\varepsilon \ge \varepsilon' > 0$, \begin{align*}
        \pr{ \abs{p^* - p^{**}} \ge \varepsilon \; : \; \substack{(\bfrho^*, p^*) \from \valest_{V, \alice}(\bfrho, 1^{1/\varepsilon}) \\ (\bfrho^{**}, p^{**}) \from \valest_{V, \alice}(\bfrho^*, 1^{1/\varepsilon'}) }  } \le \varepsilon.
    \end{align*}
    \item {\bf 2-Projection Implementation:} For every $(V, \alice, \varepsilon)$, there exist projective measurements $\cM_0 = (\Pi_A, I - \Pi_A)$ and $ \cM_1 = (\Pi_B, I - \Pi_B)$ and classical deterministic algorithms $f,g$ such that the execution of $\valest_{V, \alice}(\bfrho, 1^{1/\varepsilon})$ does the following: \begin{enumerate}
        \item Initialize an empty database $L = \emptyset$.
        \item Initialize an auxiliary register as $\ket{\psi_0}$, so that the current mixed state is $\rho \otimes \ketbraX{\psi_0}$.
        \item For $1 \le i \le \poly(|V|, |\alice|, 1/\varepsilon)$: \begin{itemize}
            \item Compute $b = f(i,L) \in \{0,1,\bot\}$
            \item If $b \in \bit$, apply $\cM_b$ and obtain outcome $\ell_i$. Set $L = L \sqcup \{\ell_i\} $, where $\sqcup$ denotes disjoint union.
            \item If $b = \bot$, end the loop. Output the current residual state $\rho^*$ and the value $p^* = g(L)$.
        \end{itemize}
    \end{enumerate}
    
    Furthermore, the measurements satisfy the following: \begin{itemize}
        \item $\cM_0(\rho')$ can be described as: \begin{itemize}
            \item Pick $r \uniform \bit^d$.
            \item Compute $y \from \alice(\rho', r)$ coherently.
            \item Measure if $V(y;r) = 1$.
            \item Uncompute. 
           
        \end{itemize}
        \item $\Pi_B = I \otimes \ketbraX{\psi_0}$.
    \end{itemize}
    
    \item {\bf State Repair:} For any $\varepsilon > 0$, \begin{align} \label{eq:repairing}
        \pr{ \abs{p^* - p^{**}} \ge \varepsilon \; : \; \substack{ (\bfrho^*, p^*) \from \valest_{V, \alice}(\bfrho, 1^{1/\varepsilon}) \\ (\bfsigma, y) \from \Pi(\bfrho^*) \\ \bfsigma^* \from \repair_{V,\alice,\Pi}(\bfsigma, y, p, 1^{1/\varepsilon}, 1^k) \\ (\bfrho^{**}, p^{**}) \from \valest_{V, \alice}(\bfsigma^*, 1^{1/\varepsilon})}  } \le \varepsilon.
    \end{align}

\end{enumerate}

\end{lemma}

\subsection{Unlearnable Distributions}

\begin{definition}[Unlearnability]\label{def:unlearnable}
A distribution $\distr = \distr(\secparam)$ is called \emph{unlearnable} if for any query-bounded adversary $\alice^{P_y(\cdot)}$ with oracle access to $P_y(\cdot)$, we have: \begin{align*}
    \pr{y' = y: \substack{y \from \distr \\ y' \from \alice^{P_y(\cdot)}(1^\secparam) }} \le \negl(\secparam).
\end{align*}

\end{definition}

\section{Cloning Games - Definitions} \label{sec:def_cg}

\newcommand{\oraclegen}{\mathsf{GenOracle}}
\newcommand{\trivprob}{p^{\mathsf{triv}}}
\newcommand{\trivattack}{\mathsf{TRIV}}

We would like to capture all cryptographic games where the adversary needs to clone a particular functionality of a given quantum token. The quantum token could be a copy-protected program, signature token, unclonable ciphertext, unclonable decryption key, or any quantum state that serves some functionality which could only be used by one party at a given time. We start off with basic definitions and give generalizations in \Cref{sec:prelims_extensions}.

\begin{definition}[Cloning Game] \label{def:cloning_game} A \emph{cloning game} consists of a tuple of efficient algorithms $\game = (\setup,\allowbreak \tokengen,\allowbreak \gench,\allowbreak \ver)$:

\begin{itemize}
    \item {\bf Key Generation: } $\setup(1^\secparam)$ is a PPT algorithm which takes as input a security parameter $1^\secparam$ in unary. It outputs a secret key $\sk \in \{0,1\}^*$. We will assume without loss of generality\footnote{This is in order to simplify the notation for the rest of the algorithms. We will sometimes make this inclusion explicit, and other times it is understood implicitly.} that $\sk$ always contains the security parameter $1^\secparam$.

    \item {\bf Token Generation:} $\tokengen(\sk, m)$ is a QPT algorithm that takes as input a secret key $\sk$ and a message $m \in \{0,1\}^*$. It outputs a quantum token $\rho$. 
    
    \item {\bf Challenge Generation: } $\gench(\sk, m)$ takes as input a secret key $\sk$ and a message $m$. It outputs a classical challenge $\ch \in \{0,1\}^*$.
    \item {\bf Verification: } $\ver(\sk, m, \ch,\ans)$ takes as input a secret key $\sk$, a message $m$, a challenge $\ch$, and an answer $\ans$. It outputs either $0$ (reject) or $1$ (accept).
\end{itemize}

\end{definition}

\subsection{Correctness} Before we talk about cloning experiments, we should specify what property of a quantum token $\rho$ we would like to be unclonable. Intuitively, the property will be captured by the ability to honestly pass verification using the token $\rho$. This brings us to the definition of correctness for a cloning game:

\begin{definition}[Correctness] \label{def:correctness}
Let $\delta : \Z^+ \to [0,1]$. We say that $\game$ has $\delta$-correctness if there exists an efficient quantum algorithm $\alicecor{\game}$ such that for all messages $m \in \cM$:

\begin{align*}
    \pr{ \substack{\ver(\sk, m, \ch, \ans) = 1} \; : \; \substack{ \sk \from \setup(1^\secparam) \\ \rho \from \tokengen(\sk, m) \\ \ch \from \gench(\sk, m) \\ \ans \from \alicecor{\game}( \rho, \ch)} } \ge \delta(\secparam)
\end{align*}

If $\delta = 1$ (or $\delta(\secparam) = 1 - \negl(\secparam)$), we say $\game$ has \emph{perfect} (or \emph{statistical}) correctness. In this work, we will mainly focus on statistically correct cloning games. \\
\par \emph{Note: In the correctness definition above, $\alicecor{\game}$ should be considered an honest user of the primitive.}

\end{definition}

\subsection{Special Types of Cloning Games}

\par Next, we define some special cases, with terminology borrowed from classical security notions.

\begin{definition}[Cloning Search Game] \label{def:cloning_search_game}
    Let $\game = (\setup, \tokengen, \gench, \ver)$ be a cloning game such that $\ver(\sk, m, \ch, \ans)$ accepts if and only if $\ans = m$. Then, $\game$ is called a \emph{cloning search game}.
\end{definition}

\begin{definition}[Cloning Decision Game] \label{def:decision_game}
Let $\game = (\setup, \tokengen, \gench, \ver)$ be a cloning game such that the answer $\ans$ taken as input by $\ver$ is one bit, i.e. $\ans \in \bit$. Then, $\game$ is called a \emph{cloning decision game}.
\end{definition}

\noindent We additionally define the notion of a \emph{cloning encryption game} when we discuss unclonable encryption in \Cref{sec:prelims_examples_ue}.

\subsection{Security} \label{sec:definitions_security}
\newcommand{\genchprime}{\widetilde{\gench}}

\paragraph{Cloning Experiment.} We will define notions of security for a cloning game in terms of a security experiment. Given a token $\rho$, an adversary should not be able to generate two (possibly entangled) quantum tokens which can simultaneously pass verification. We will formalize this intuition below.

\begin{definition}[Cloning Experiment] \label{def:cloning_exp}
A \emph{cloning experiment}, denoted by $\cloninggame_{\game, \distr}$, is a security game played between a referee $\referee$ and a cloning adversary $\abc$. It is parameterized by a cloning game $\game = ( \setup, \allowbreak \tokengen, \allowbreak \gench, \ver)$ and a distribution $\distr$ over the message space $\cM$. The experiment is described as follows:
\begin{itemize}
    \item {\bf \em Setup Phase: }
    \begin{itemize}
        \item All parties get a security parameter $1^\secparam$ as input.
        \item $\referee$ samples a message $m \from \distr$.
        \item $\referee$ computes $\sk \from \setup(1^\secparam)$ and $\rho \from \tokengen(\sk, m)$.
        \item $\referee$ sends $\rho$ to $\alice$.
    \end{itemize}
    \item {\bf \em Splitting Phase: } \begin{itemize}
        \item $\alice$ computes a bipartite state $\rho'$ over registers $B,C$.
        \item $\alice$ sends $\rho'[B]$ to $\bob$ and $\rho'[C]$ to $\charlie$.
    \end{itemize}
    \item {\bf \em Challenge Phase: } \begin{itemize}
        \item $\referee$ independently samples $\ch_\bob, \ch_\charlie \from \gench(\sk, m)$.
        \item $\referee$ sends $\ch_\bob$ to $\bob$ and $\ch_\charlie$ to $\charlie$.
        \item $\bob$ and $\charlie$ send back answers $\ans_\bob$ and $\ans_\charlie$, respectively.
        \item $\referee$ computes bits $b_\bob \from \ver(\sk, m, \ch_\bob, \ans_\bob)$ and $b_\charlie \from \ver(\sk, m, \ch_\charlie, \ans_\charlie)$.
        \item The outcome of the game is denoted by $\cloninggame_{\game, \distr}(1^\secparam,\abc)$, which equals 1 if $b_\bob = b_\charlie = 1$, indicating that the adversary has won, and 0 otherwise, indicating that the adversary has lost.
    \end{itemize}
\end{itemize}

\end{definition}

\paragraph{Trivial Success.} As a baseline for unclonable security, we will consider \emph{trivial} attacks that do not require any cloning operation. The best we can hope is that such attacks are optimal, hence the definitions below.

\begin{definition}[Trivial Cloning Attack] \label{def:triv_cloning_attack}

We say that $\abc$ is a \emph{trivial cloning attack} against a cloning experiment $\cloninggame_{\game, \distr}$ if $\alice$ upon receiving a token $\rho$, sends the product state $\ketbraX{\bot} \otimes \rho$ to $\bob$ and $\charlie$. In other words, only $\charlie$ gets the token $\rho$. We denote by $\trivattack(\cloninggame_{\game, \distr})$ the set of trivial attacks against $\cloninggame_{\game, \distr}$.
\end{definition}

\begin{remark}
Note that due to the symmetry between $\bob$ and $\charlie$, the definition of trivial cloning attack could be equivalently defined so that only $\bob$ gets the token $\rho$.
\end{remark}

\begin{definition}[Trivial Success Probability for Cloning Games] \label{def:triv_success_prob}
We define the trivial success probability of a cloning experiment $\cloninggame_{\game, \distr}$ as \begin{align*}
    \trivprob(\game, \distr) := \sup_{\abc \in \trivattack(\cloninggame_{\game, \distr})} \pr{ 1 \from \cloninggame_{\game, \distr}(1^\secparam, \abc) }.
\end{align*}

\end{definition}

\paragraph{Unclonable Security.} We present the security definition of cloning games below.  

\begin{definition}[Unclonable Security] \label{def:unclonable_sec}
Let $\game$ be a cloning game, $\distr$ be a distribution over the message space $\cM$, and $\varepsilon : \Z^+ \to [0,1]$. We say that $\game$ has \emph{$(\distr, \epsilon)$ unclonable security} if for all QPT cloning adversaries $\abc$ we have:
\begin{align*}
    \pr{ 1 \from \cloninggame_{\game, \distr}(1^\secparam, \abc) } \le \trivprob(\game, \distr) + \varepsilon(\secparam).
\end{align*}

If $|\cM| = 1$, we will simply write $\varepsilon$ unclonable security.

\end{definition}

\subsubsection{Security for Search Games} For the special case of search games, we consider two definitions below. 

\begin{definition}[Unclonable Search Security] \label{def:unclonable_search_sec}
If $\game$ is a cloning search game with $(\distr, \varepsilon)$ unclonable security, we additionally say that $\game$ has \emph{$(\distr, \varepsilon)$ unclonable search security}.
\end{definition}

\begin{remark} \label{rem:triv_success}
Note that even though the definitions above are valid for any distribution $\distr$, to get meaningful security one needs to choose $\distr$ appropriately for the context. For instance, if the cloning game $\game$ represents copy-protection for point functions, it is appropriate to pick $\distr$ in a balanced way so that the trivial success probability $\trivprob(\game, \distr)$ is bounded away from 1. As long as this is the case, $(\distr, \varepsilon)$ unclonable security (for small $\varepsilon$) is non-trivial\footnote{We assume statistical correctness here.} in the sense that it is classically impossible and it uses the power of no-cloning. On the other hand, when $\trivprob(\game, \distr) \approx 1$ unclonable security becomes trivial and achieved by uninteresting constructions including classical games. 

\end{remark}

\begin{definition}[Unclonable Indistinguishable Security] \label{def:unclonable_dist_sec}
Let $\distr_{m_0, m_1}$ denote the distribution that outputs messages $m_0$ and $m_1$ with probability $1/2$ each. We say that a search game $\game$ has \emph{$\varepsilon$ unclonable indistinguishable security} if it has $(\distr_{m_0, m_1}, \varepsilon)$ unclonable search security
for any pair of messages $m_0, m_1 \in \cM$.
\end{definition}

In \Cref{def:unclonable_sec,def:unclonable_search_sec,def:unclonable_dist_sec}, if $\abc$ is not required to be efficient, then we say that $\game$ has \emph{information theoretic} unclonable (search/indistinguishable) security.

\subsection{Extended Definitions} \label{sec:prelims_extensions}

\paragraph{Stateful Games.}
Most cloning games can be captured by \Cref{def:cloning_game}. Yet, some cloning games are \emph{stateful} in the sense that verification takes as additional input the random coins used in the challenge generation. With this in mind, we define a generalization of the cloning game below, highlighting the differences in \textcolor{blue}{blue}. Throughout this section, we will assume that $\gench$ is a classical algorithm.\footnote{If $\gench$ is quantum, one can similarly define statefulness by having $\gench$ output some random coins that it sampled during its execution. One would need this because unlike classical algorithms, the randomness of a quantum algorithm can inherently result from collapsing measurements and hence cannot be modeled as an auxiliary random input string.}

\begin{definition}[Stateful Cloning Game] \label{def:cloning_game_stateful} A \emph{stateful cloning game} consists of a tuple of efficient algorithms $\game = \brac{\setup, \tokengen, \gench, \ver}$.

\begin{itemize}
    \item {\bf Key Generation: } $\setup(1^\secparam)$ takes as input a security parameter $1^\secparam$ in unary. It outputs a secret key $\sk$.

    \item {\bf Token Generation:} $\tokengen(\sk, m)$ takes as input a secret key $\sk$ and a message $m$. It outputs a quantum token $\rho$. 
    
    \item {\bf Challenge Generation: } $\gench(\sk, m; \highlight{r_{\gench}})$ takes as input a secret key $\sk$ and a message $m$. 

    It outputs a classical challenge $\ch$.
    \item {\bf Verification: } $\ver(\sk, m, \ch,\ans \highlight{, r_{\gench}})$ takes as input a secret key $\sk$, a message $m$, a challenge $\ch$, an answer $\ans$, \highlight{random coins $r_\gench$ used by $\gench$ when generating $\ch$}. It outputs either $0$ (reject) or $1$ (accept).
\end{itemize}

When we talk about cloning games, we will always implicitly mean stateful cloning games. In fact, all of our results easily generalize to stateful cloning games. However, we will omit $ r_{\gench}$ above and use the syntax in \Cref{def:cloning_game} when appropriate, for simplicity.

\end{definition}

\paragraph{Security Against Correlated Distributions.}
When we defined security in \Cref{sec:definitions_security}, we assumed that $\bob$ and $\charlie$ in the security experiment receive independently generated challenges. We will define security more broadly and refer to the aforementioned definition as \emph{independent-challenge security}. For simplicity, we will assume that challenge generation is classical, i.e., $\gench$ is a PPT algorithm, which is true for all the primitives considered in this work.

\begin{definition}[Challenge Extension] \label{def:extension}
    Let $\gench$ be a challenge generation algorithm that takes as input randomness from $\cR = \bit^{\poly(\secparam)}$. We say that $\genchprime$ is an \emph{extension} of $\gench$ if: \begin{itemize}
        \item On input a secret key $\sk$ and a message $m$, it outputs a pair of random strings $(r_\bob, r_\charlie) \in \cR^2$.
        \item For any $(\sk, m)$, if $(r_\bob, r_\charlie) \from \genchprime(\sk, m)$, then the marginal distributions of both $r_\bob$ and $r_\charlie$ are equal to $\cU_\cR$.
    \end{itemize}
\end{definition}

We will also refer to $\genchprime$ as an extension of a game $\game$ whenever $\gench$ is the challenge generation algorithm for $\game$. We sometimes will omit and not specify the extension $\genchprime$, in which case it is either clear from the context or assumed to be $\genchprime = \gench_{\independent}$ by default (see \Cref{def:id-ind_ch}).

\paragraph{Extended Cloning Experiment.} Next, we give a more general definition of unclonable security, highlighting the differences to the corresponding definitions in \Cref{sec:definitions_security} in \highlight{blue}.

\begin{definition}[Extended Cloning Experiment] \label{def:cloning_exp_ext}

An \emph{(extended) cloning experiment}, denoted by $\cloninggame^\highlight{{\genchprime}}_{\game, \distr}$, is a security game played between a referee $\referee$ and a cloning adversary $\abc$. It is parameterized by a cloning game $\game = ( \setup, \allowbreak \tokengen, \allowbreak \gench, \ver)$, a distribution $\distr$ over the message space $\cM$, \highlight{and an extension $\genchprime$ of $\gench$}. The experiment is described as follows:
\begin{itemize}
    \item {\bf \em Setup Phase: }
    \begin{itemize}
        \item All parties get a security parameter $1^\secparam$ as input.
        \item $\referee$ samples a message $m \from \distr$.
        \item $\referee$ computes $\sk \from \setup(1^\secparam)$ and $\rho \from \tokengen(\sk, m)$.
        \item $\referee$ sends $\rho$ to $\alice$.
    \end{itemize}
    \item {\bf \em Splitting Phase: } \begin{itemize}
        \item $\alice$ computes a bipartite state $\rho'$ over registers $B,C$.
        \item $\alice$ sends $\rho'[B]$ to $\bob$ and $\rho'[C]$ to $\charlie$.
    \end{itemize}
    \item {\bf \em Challenge Phase: } \begin{itemize}
        \item \highlight{$\referee$ samples $(r_\bob, r_\charlie) \from \genchprime(\sk, m)$ and then computes $\ch_\bob = \gench(\sk, m; r_\bob)$, $\ch_\charlie = \gench(\sk, m; r_\charlie)$.}
        \item $\referee$ sends $\ch_\bob$ to $\bob$ and $\ch_\charlie$ to $\charlie$.
        \item $\bob$ and $\charlie$ send back answers $\ans_\bob$ and $\ans_\charlie$, respectively.
        \item $\referee$ computes bits $b_\bob \from \ver(\sk, m, \ch_\bob, \ans_\bob\highlight{, r_\bob})$ and $b_\charlie \from \ver(\sk, m, \ch_\charlie, \ans_\charlie\highlight{, r_\charlie})$.
        \item The outcome of the game is denoted by $\cloninggame_{\game, \distr}^{\highlight{\genchprime}}(1^\secparam,\abc)$, which equals 1 if $b_\bob = b_\charlie = 1$, indicating that the adversary has won, and 0 otherwise, indicating that the adversary has lost.
    \end{itemize}
\end{itemize}

\end{definition}

\par Next, we discuss two important special cases for the extension $\genchprime$. In the first case, $\bob$ and $\charlie$ get the same challenge, whereas in the second case, they get independently generated challenges. Keep in mind that the second case is the default assumption when we do not mention extensions in \Cref{sec:definitions_security}.

\begin{definition}[Identical/Independent-Challenge Cloning Experiment] \label{def:id-ind_ch}
Let $\game = (\setup, \tokengen, \gench, \ver)$ and define the following extensions of $\gench$, which takes randomness from the set $\cR$: \begin{itemize}
    \item $\gench_\identical(\sk, m)$ samples $r \from \cU_\cR$ and outputs $(r, r)$.
    \item $\gench_\independent(\sk, m)$ computes $r_\bob, r_\charlie \from \cU_\cR$ independently. It outputs $(r_\bob, r_\charlie)$.
\end{itemize}

Then, we call the cloning experiments $\cloninggame^{\gench_\identical}_{\game, \distr}$ and $\cloninggame^{\gench_\independent}_{\game, \distr}$ an \emph{identical-challenge cloning experiment} or an \emph{independent-challenge cloning experiment}, respectively.

\end{definition}

\paragraph{Trivial Success Probability.} We will slightly modify the definition of trivial attacks to account for the potential asymmetry introduced by a challenge extension.

\begin{definition}[(Extended) Trivial Cloning Attack] \label{def:triv_cloning_attack_ext}
We say that $\abc$ is a \emph{\highlight{$\bob$-}trivial cloning attack} against a cloning experiment $\cloninggame^\highlight{{\genchprime}}_{\game, \distr}$ if $\alice$ upon receiving a token $\rho$, sends the product state $\ketbraX{\bot} \otimes \rho$ to $\bob$ and $\charlie$. In other words, only $\charlie$ gets the token $\rho$. We denote by $\trivattack\highlight{_\bob}(\cloninggame^\highlight{{\genchprime}}_{\game, \distr})$ the set of \highlight{$\bob$-}trivial cloning attacks against $\cloninggame^\highlight{{\genchprime}}_{\game, \distr}$. We similarly define $\trivattack\highlight{_\charlie}(\cloninggame^\highlight{{\genchprime}}_{\game, \distr})$ as the set of $\highlight{\charlie}$-trivial attacks. 

Finally, we define \begin{align*} \trivattack(\cloninggame^\highlight{{\genchprime}}_{\game, \distr}) := \trivattack_\bob(\cloninggame^\highlight{{\genchprime}}_{\game, \distr}) \cup \trivattack_\charlie(\cloninggame^\highlight{{\genchprime}}_{\game, \distr}) \end{align*}
as the set of trivial cloning attacks against $\cloninggame^\highlight{{\genchprime}}_{\game, \distr}$.

\end{definition}

\begin{definition}[(Extended) Trivial Success Probability for Cloning Games] \label{def:triv_success_prob_ext}
We define the \highlight{$\bob$-}trivial success probability of a cloning experiment $\cloninggame^\highlight{{\genchprime}}_{\game, \distr}$ as \begin{align*}
    \trivprob_\highlight{\bob}(\game, \distr, \highlight{\genchprime}) := \sup_{\abc \in \trivattack\highlight{_\bob}(\cloninggame^\highlight{{\genchprime}}_{\game, \distr})} \pr{ 1 \from \cloninggame^\highlight{{\genchprime}}_{\game, \distr}(\abc) }.
\end{align*}
We similarly define $\trivprob_\charlie(\game, \distr, \highlight{\genchprime})$ as the $\charlie$-trivial success probability of $\cloninggame^\highlight{{\genchprime}}_{\game, \distr}$. Accordingly, we define the trivial success probability of $\cloninggame^\highlight{{\genchprime}}_{\game, \distr}$ as \begin{align*}
    \trivprob(\game, \distr, \highlight{\genchprime})
    =& \max \brac{ \trivprob_\highlight{\bob}(\game, \distr, \highlight{\genchprime)}, \trivprob_\highlight{\charlie}(\game, \distr, \highlight{\genchprime)} }.
\end{align*}

\end{definition}

\begin{remark}
One may consider mixtures of $\bob$-trivial and $\charlie$-trivial cloning attacks as trivial, but such attacks cannot do better than trivial cloning attacks by convexity. 
\end{remark}

\begin{definition}[(Extended) Unclonable Security] \label{unclonable_sec}
Let $\game$ be a cloning game with extension $\genchprime$, $\distr$ be a distribution over the message space $\cM$, and $\varepsilon : \Z^+ \to [0,1]$. We say that $\game$ has \emph{$(\distr, \epsilon, \highlight{\genchprime})$ unclonable security} if for all QPT cloning adversaries $\abc$ we have \begin{align*}
    \pr{ 1 \from \cloninggame^\highlight{{\genchprime}}_{\game, \distr}(1^\secparam, \abc) } \le \trivprob(\game, \distr, \highlight{\genchprime}) + \varepsilon(\secparam).
\end{align*}

\end{definition}

\par We define $(\distr, \varepsilon, \genchprime)$ unclonable search security and $(\varepsilon, \genchprime)$ unclonable indistinguishable security similarly. If $\genchprime = \gench_\independent$ (resp., $\genchprime = \gench_\identical$), then we say $\game$ has $(\distr, \varepsilon)$ independent-challenge (resp., identical-challenge) unclonable security. Likewise, information theoretic security is defined by removing the efficiency requirement from $\abc$ as before.

\par A good number of cloning games of interest will have the following additional property.

\begin{definition}[Evasiveness] \label{def:evasive}
    A cloning game $\game$ is called \emph{$\distr$-evasive} if $\trivprob(\game, \distr)$ is a negligible function in $\secparam$.
\end{definition}
Note that the definition above is independent of the extension $\genchprime$ for statistically correct games, i.e. $\trivprob(\game, \distr, \genchprime)$ is negligible for any extension $\genchprime$. We keep the definition simple given that we only use it for statistically correct games.

\paragraph{Asymmetric Verification.} Another way to generalize the cloning games is to allow asymmetric verification for $\bob$ and $\charlie$, which we will define and discuss in \Cref{sec:asym_cg}.

\subsection{Examples} \label{sec:prelims_examples}

\newcommand{\functclass}{\cF}

In this section, we demonstrate the comprehensiveness of cloning games by casting popular unclonable primitives as cloning games. We restrict our attention to primitives with symmetric verification, and those with asymmetric verification, such as secure software leasing or certified deletion, require a slightly more general syntax, which will be defined in \Cref{sec:asym_cg}.

\subsubsection{Copy-Protection}
Let $\cF$ be the class of functions of the form $f : \cX \to \cY$, parameterized implicitly by a security parameter $\secparam$, and let $\distr$ be a distribution over $\cF$. A copy-protection scheme for $\distr$ is a pair of efficient algorithms $(\cp, \eval)$: \begin{itemize}
    \item $\cp(1^\secparam, d_f)$ takes as input description $d_f$ of a function $f : \cX \to \cY$ and outputs a copy-protected quantum program $\rho_f$.
    
    \item $\eval(1^\secparam, \rho_f, x)$ takes as input a quantum program $\rho_f$ and an input $x \in \cX$. It outputs a value $y \in \cY$.
\end{itemize}

\noindent $(\cp, \eval)$ defines a cloning game $\game_\cp^{\distr'} = \gamesetup$ for any family of distributions $\distr' = \brac{\distr'_f}_{f \in \cF}$ over $\cF$ as follows. Note that, $\distr'$ defines a distribution on challenge inputs, therefore, it specifies $\gench$.

\begin{itemize}
    \item The message space $\cM$ is the set of function descriptions $d_f$ for all $f \in \cF$.
    \item $\setup(1^\secparam)$ outputs $\sk = 1^\secparam$, i.e. there is no secret key. 
    \item $\tokengen(\sk, m)$ parses the input as $m = d_f$, then it computes $\rho_f \from \copyprotect(1^\secparam, d_f)$ and outputs $\rho_f$. 
    \item $\gench(\sk, m)$ parses $m = d_f$ and samples input $x \from \distr'_f$.\footnote{Here we make the natural assumption that correctness and security are defined with respect to the same distribution $\distr'_f$. Intuitively, the scheme should protect against cloning the functionality of the honest evaluator.}
    \item $\ver(\sk, m, \ch,\ans)$ parses $m = d_f$, $\ch = x$. It accepts if and only if $\ans = f(x)$.
    
\end{itemize}

\paragraph{Correctness:} We require that $\game_\cp$ has statistical correctness, to ensure that the copy-protected program is reusable. More specifically, $\alice_{\game_{\cp}}(\rho, \ch)$ runs $\eval(\rho, \ch)$. \footnote{Note that this captures the average-input correctness as opposed to per-input correctness.}

\paragraph{Security:} We consider a game-based definition of copy-protection, first defined by~\cite{CMP20,BJLPS21}. We say that $(\cp, \eval)$ is secure for a class of distributions $\distr'$ if $\game_\cp^{\distr'}$ has $(\distr, \varepsilon)$ unclonable security. For optimal security, we require $\varepsilon$ to be negligible.

\subsubsection{Unclonable Encryption} \label{sec:prelims_examples_ue}

Below, we define unclonable encryption~\cite{BL20} as a cloning game. We focus our attention to one-time secret-key setting, in which case unclonable encryption is synonymous with cloning encryption games defined in \Cref{def:cloning_enc_game}. It is known in the literature that construction in this simple setting can be generically lifted to achieve unclonable encryption with additional properties, such as public-key encryption \cite{AK21}. We note, however, that unclonable encryption with such properties can still be expressed as a cloning game by modifying the syntax of a cloning encryption game. We state the correspondence below, which is easy to verify.

\begin{definition}[Cloning Encryption Game] \label{def:cloning_enc_game}
    A cloning search game $\game = \gamesetup$ is called a \emph{cloning encryption game} if $\gench(\sk, m)$ outputs $\sk$ with probability 1 for all $(\sk, m)$. 
\end{definition}

\begin{fact}[Informal] \label{lem:def_ue}
    An unclonable encryption scheme for a message space $\cM$ exists with unclonable (unclonable indistinguishable) security if and only if a cloning encryption game $\game$ for $\cM$ with unclonable (unclonable indistinguishable) security exists.
\end{fact}

There are two types of security we will consider for unclonable encryption: (1) $(\cU_\cM, \varepsilon)$ unclonable security and (2) $\varepsilon$ unclonable indistinguishable security. These security definitions together with $\delta$-correctness are on par with the original definitions of \cite{BL20}.\footnote{Although (2) was defined in a slightly different way in \cite{BL20}, the difference is inconsequential, and our version has been used in follow-up works such as \cite{AKLLZ22}. We also mention that \cite{BL20} considered perfect correctness.} Note that since $\gench$ is deterministic, it has a unique extension.

\par A particular case of interest is adapted\footnote{We omit a classical one-time-pad on the message $m$, which is irrelevant for the purposes of unclonability.} from the \emph{conjugate encryption} of \cite{BL20} and uses Wiesner (BB84) states:
\begin{definition}[BB84 Cloning Game] \label{def:bb84_cloning_game}
    $\game_{\mathsf{BB84}} = \gamesetup$ is a cloning encryption game with message space $\cM = \bit^\secparam$, defined as follows: \begin{itemize}
        \item $\setup(1^\secparam)$ outputs $\theta \uniform \bit^\secparam$
        \item $\tokengen(\theta, m)$ takes as input $\theta, m \in \bit^\secparam$ and outputs $\rho = \ketbraX{m^\theta}$, where $\ket{m^\theta} = H^\theta \ket{m}$
        \item $\gench$ and $\ver$ are defined as part of a cloning encrpytion game.
    \end{itemize}
\end{definition}

\begin{lemma}[Security of BB84 Cloning Game \cite{BL20}] \label{lem:bb84}
    The game $\game_{BB84}$ above has $(\cU_\cM, |\cM|^{-\delta})$ unclonable security for some constant $\delta > 0$.
\end{lemma}

\newcommand{\qdk}{\rho_{dk}}

\subsubsection{Single Decryptor Encryption}
We define single-decryptor encryption as a tuple of efficient algorithms $(\gen, \tokengen', \enc, \dec)$, adapted from the definition of (secret-key) single-decryptor encryption (with honestly generated keys) in \cite{GZ20}: \begin{enumerate}
    \item $\gen(1^\secparam)$ takes as input a security parameter and outputs a classical secret key $sk$.
    \item $\tokengen'(sk)$ takes as input a classical secret key and it outputs a quantum decryption key $\qdk$.
    \item $\enc(sk, x)$ takes as input a secret key and a classical message. It outputs a classical ciphertext $ct$.
    \item $\dec(\rho, ct)$ takes as input a quantum decryption key and a classical ciphertext. It outputs a classical message $x'$.
\end{enumerate}
$(\gen, \tokengen', \enc, \dec)$ defines a stateful cloning game $\game^{\distr_\cX}_{\mathsf{SDE}} = \gamesetup$, parameterized by a distribution $\distr_\cX$, where $\cX$ is the set of classical messages encrypted by this scheme, as follows: \begin{itemize}
    \item $\setup(1^\secparam)$ runs $sk \from \gen(1^\secparam)$ and outputs $\sk = sk$.
    \item There is no message, i.e. $m = \bot$.
    \item $\tokengen(\sk, m)$ computes $\qdk \from \tokengen'(sk)$ and outputs $\qdk$.
    \item $\gench(\sk, m; r_\gench)$ samples $x \uniform \cX$ using random coins $r_\gench$. It outputs $c \from \enc(sk, x)$.
    \item $\ver(\sk, m, \ch, \ans, r_\gench)$ computes $x$ as above using $r_\gench$. Then it accepts if and only if $\ans = x$.
\end{itemize}

\paragraph{Correctness.}
We say that $(\gen, \tokengen', \enc, \dec)$ has correctness if $\game_{\mathsf{SDE}}^{\distr_\cX}$ has perfect correctness for any distribution $\distr_\cX$. More specifically, $\alice_{\game_{\mathsf{SDE}}}(\qdk, \ch)$ runs $\dec(\qdk, \ch)$.

\paragraph{Security.}
We say that $(\gen, \tokengen', \enc, \dec)$ has $\varepsilon$ unclonable security if $\game^{\cU_\cX}_{\mathsf{SDE}}$ has $\varepsilon$ unclonable security\footnote{We omit the message distribution due to the lack of message.}.

In other words, the ability to decrypt a random classical message is the unclonable property of the quantum decryption key. For optimal security, we require that $\varepsilon$ is negligible.

\newcommand{\mint}{\mathsf{Mint}}
\newcommand{\vertoken}{\mathsf{VerToken}}
\newcommand{\qm}{\mathsf{QM}}

\subsubsection{Quantum Money}
Next, we give examples of cloning games with quantum verification. We focus on quantum money, first introduced by Wiesner~\cite{Wiesner83}. We consider a public-key variant of quantum money considered by~\cite{AC12,Zha17}. We note that the description below can be suitably adapted to case private-key quantum money as a cloning game. A public-key quantum money scheme is a tuple of efficient algorithms $(\gen, \mint, \vertoken)$: \begin{itemize}
    \item $\gen(1^\secparam)$ takes as input a security parameter and outputs a public-secret key pair $(pk, sk)$.
    \item $\mint(sk)$ takes as input a secret key and outputs a classical serial number $s$ and a quantum banknote $\rho_s$.
    \item $\vertoken(pk, s, \rho)$ takes as input a public key, a serial number, and a quantum state. It outputs 0 (reject) or 1 (accept).
\end{itemize}

\noindent $(\gen, \mint, \vertoken)$ defines a cloning game $\game_{\qm} = \gamesetup$ as follows:

\begin{itemize}
    \item We set $m = \bot$, i.e. there is no message.
    \item $\setup(1^\secparam)$ runs $(pk, sk) \from \gen(1^\secparam)$ and outputs $\sk = (pk, sk)$
    \item $\tokengen(\sk, m)$ parses the input as $\sk = (pk, sk)$, runs $(s, \rho_s) \from \mint(sk)$, and outputs $\rho \otimes \ketbraX{s} \otimes \ketbraX{pk}$.
    \item $\gench(\sk, m)$ outputs $\ch = \bot$, i.e. no challenge.
    \item $\ver(\sk, m, \ch, \sigma_\ans)$ parses the input as $\sk = (pk, sk)$ and outputs $ b \from \vertoken(pk, \sigma_{\ans})$ 
\end{itemize}

\paragraph{Correctness.} We say that the quantum money scheme has correctness if $\game_\qm$ has statistical correctness. Note that $\alice_{\game_\qm}$ can simply output the quantum banknote it receives to satisfy \Cref{def:correctness} (correctness). Furthermore, it can be assumed without loss of generality that the optimal $\alice_{\game_\qm}$ acts as identity (i.e. outputs $\rho$ as is) since there is no challenge. Therefore, this fully captures the usual definition of correctness for quantum money schemes.

\paragraph{Security.} We say that the quantum money scheme is secure if $\game_\qm$ has $\varepsilon$

unclonable security. For optimal security, we require that $\varepsilon$ is negligible. Note that unclonable security as we defined only gives 1-to-2 unclonability, but it can be generalized to $k$-to-$k+1$ unclonability. Alternatively, one can define a quantum-money mini scheme in our framework, which is necessary and sufficient for constructing public-key quantum money \cite{AC12}.

\newcommand{\bfP}{{\bf P}}
\newcommand{\bfQ}{{\bf Q}}

\section{Constructive Post-Quantum Reductions: The Non-Local Setting} \label{sec:nonlocal_rd}
\noindent We present a new toolkit to understand the relationship between different cloning games. We first define a notion of non-local classical reductions. Roughly speaking, a classical non-local reduction transforms a non-local classical algorithm\footnote{Refer to~\Cref{sec:prelims} for the definition of a non-local algorithm.}, solving a problem $P$ to a non-local classical algorithm solving a problem $Q$. The reason we consider non-local classical reductions is that it turns out to be simpler to come up with non-local reductions in the classical setting. We then show how to generically upgrade some classes of classical non-local reductions to quantum non-local reductions, the analogous notion for quantum non-local algorithms. The resulting quantum non-local reductions are useful in analyzing the relationship between different cloning games.

\subsection{Definitions}
In this section, we borrow definitions from Section 3 of \cite{BBK22} and adapt them to the non-local setting\footnote{For simplicity, we omit explicit purification of quantum solvers (see Definition 3.5), as we will make nearly black-box use of the results of \cite{BBK22}.}.

\paragraph{Extension.} Similar to \Cref{def:extension}, we define the notion of an extension over random coins in this context.

\begin{definition}[$d$-extension]
We say that a distribution $\distrpr$ over $\bit^d \times \bit^d$ is a \emph{$d$-extension} if for $(r_\bob, r_\charlie) \from \distrpr $, the marginal distributions of both $r_\bob$ and $r_\charlie$ are $\cU_{\bit^d}$.

\end{definition}

\subsubsection{Classical Assumptions}
\noindent The first two definitions below are commonly used in the cryptography literature. Note that we use $\gench, \ver$ to denote the algorithms in order to point out the connection between a non-interactive assumption and the challenge phase of a cloning game (see \Cref{def:cloning_game,def:cloning_exp}). 

\begin{definition}[Non-Interactive Assumption] 
A non-interactive assumption $\assumption=(d,n,m,\gench,\ver,c)$ is associated with polynomials $d(\secparam),n(\secparam),m(\secparam)$ and a tuple $P=(\gench,\allowbreak \ver,\allowbreak c)$ with the following syntax. Here, $\gench$ and $\ver$ are classical algorithms, and $c : \N \to \R^+ \cup \{0\}$ is the assumption's threshold.
\begin{itemize}
    \item Challenge generator, $\gench(1^{\secparam};r)$: on input security parameter $\secparam$ and random coins $r \in \{0,1\}^d$, outputs a challenge $x \in \{0,1\}^n$. 
    \item Verifier, $\ver(1^{\secparam},r, y)$: on input security parameter $\secparam$, random coins $r \in \{0,1\}^d$, and answer $y \in \{0,1\}^m$, outputs $1$ (accept) or $0$ (reject).
\end{itemize}

\par We will sometimes use $P$ and $\assumption$ interchangeably, with the understanding that the polynomials $d,n,m$ are implicit. 
We say that $P$ is \emph{evasive} if $c$ is a negligible function. 
\end{definition}

Another important property is called \emph{verifiably polynomial image}, and informally it requires that it is possible to verify if a solution lies in a polynomial-size superset of valid solutions without the random coins $r$ of $\gench$.

\begin{definition}[Verifiably Polynomial Image] \label{def:ver_poly_image}

A non-interactive assumption $\assumption=(d,n,m,\gench,\ver,c)$ has a \emph{verifiably polynomial image} if there exists a polynomial $k(\secparam)$ and an efficient verifier $K$ such that for every $x \in \bit^n$, the set $Y_x := |\{y\ :\ K(1^{\secparam},x,y)=1\}|$ has size at most $k$ and for any valid challenge $x = G(1^\secparam; r)$ and answer $\ver(1^\secparam, r, y) = 1$, it holds that $y \in Y_x$.
\end{definition}

\begin{remark}
At a high level, the above definition states that a non-interactive assumption with verifiably polynomial range has a public verification algorithm $K$ that accepts all the solutions also accepted by $\ver$ (and possibly, more) such that for any $x$, the set of strings accepted by $K$ should be of polynomial size.  
\end{remark}

\subsubsection{Solvers}

\paragraph{Local Solvers.} A local solver for a non-interactive assumption $P$, which we call a \emph{$P$-solver}, is a pair $\alice = (\widetilde{alice}, \state_0)$, where $\widetilde{\alice}$ is an algorithm which takes as input a challenge $x \in \bit^n$ and outputs an answer $y \in \bit^m$, whereas $\state_0$ is an initial state. We define the value and advantage of a local solver below: 

\begin{definition}[Value and Advantage of a Local Solver (Definition 3.3 in \cite{BBK22})] \label{def:local_solver}
       Let $\assumption=(d,n,m,\gench,\ver,c)$ be a non-interactive assumption, with $P = (\gench, \ver, c)$, and let $\alice = (\widetilde{\alice}, \state_0)$ be a $P$-solver. We define the (one-shot) value and advantage of $\alice$, parameterized by the security parameter $\secparam$, as \begin{align*}
           \val_P[\alice] = \pr{ \ver(1^\secparam, r, y) = 1 \; \Bigg\vert \; \substack{ r \uniform \bit^d \\ x = \gench(1^\secparam; r) \\ y \from \widetilde{\alice}(1^\secparam, x, \state_0)} }, \quad \adv_P[\alice] = \abs{\val_P[\alice] - c}.
       \end{align*}
\end{definition}

\paragraph{Non-Local Solvers.} Below, we adapt \Cref{def:local_solver} to the non-local setting. A non-local adversary consists of two (possibly inefficient) algorithms $(\bob,\charlie)$, where both $\bob$ and $\charlie$ are given (possibly correlated) instances from a non-interactive assumption. Each adversary is expected to solve the instance they get. Note that both $\bob$ and $\charlie$ are not allowed to communicate with each other. In the end, the adversary wins if both $\bob$ and $\charlie$ win. In general, for algorithms $\bob$ and $\charlie$ acting on separate registers, we call the algorithm $\bob \otimes \charlie$ a non-local algorithm\footnote{A non-local algorithm can be implemented by two spatially separated and non-communicating parties, hence the name.}. We consider the non-uniform setting, where both $\bob$ and $\charlie$ could share some auxiliary information (either a string or a quantum state, depending on whether the adversary is classical or quantum) in the very beginning.

\noindent In order to define reductions, we need to define stateful solvers. Suppose in a reduction, we use the solver for a problem $P$ to design a solver for problem $Q$. Suppose the $Q$ solver runs the $P$ solver multiple times. In the classical setting, the $P$ solver could be stateless, whereas in the quantum setting the $P$ solver is inherently stateful and thus, we need to define stateful solvers appropriately below. Formally, we adapt Definition 3.4 from~\cite{BBK22} to the non-local setting.

\begin{definition}[Non-Local Stateful Solvers: Syntax]

Let $P$ be a non-interactive assumption with $d$-bit random coins.
\par Let $\ell=\ell(\secparam),\ell_{\bob}=\ell_{\bob}(\secparam),\ell_{\charlie}=\ell_{\charlie}(\secparam)$ be functions such that $\ell_{\bob}(\secparam) + \ell_{\charlie}(\secparam) = \ell(\secparam)$
and let $\distrpr$ be a $d$-extension.
An $(\ell,\ell_{\bob},\ell_{\charlie})$-stateful non-local $(P, \distrpr)$-solver $\adversary=(\bob,\charlie,\state_0=\{\state_{0,\secparam}\}_{\secparam})$ is defined as follows. 
\begin{itemize}
    \item $\state_0=\{\state_{0,\secparam}\}_{\secparam}$ is a sequence of bipartite $\ell$-qubit states (with the partitions being $\ell_{\bob}$-qubit register ${\bf B}$ and $\ell_{\charlie}$-qubit register ${\bf C}$). 
    
    \item $\adversary$ is a quantum algorithm that takes as input a security parameter $1^{\secparam}$, a step $1^t$, a pair of challenges $x_\bob, x_\charlie \in \{0,1\}^n$, and $\state=(\state_{\bob},\state_{\charlie})$, which is an $\ell$-qubit bipartite state.
    It runs $(\bob(1^{\secparam},x_\bob,\cdot) \otimes \charlie(1^{\secparam},x_\charlie,\cdot)) (\state)$, with $\bob$ getting as input register ${\bf B}$ and $\charlie$ getting as input ${\bf C}$, to obtain $(y_{\bob},y_{\charlie},\state')$, where $y_{\bob} \in \{0,1\}^m$ is $\bob$'s output, $y_{\charlie} \in \{0,1\}^m$ is $\charlie$'s output and $\state'=(\state'_{\bob},\state'_{\charlie})$ is a bipartite $\ell$-qubit state on $\ell_{\bob}$-qubit register ${\bf B}$ and $\ell_{\charlie}$-qubit register ${\bf C}$ with ${\bf B}$ (resp., ${\bf C}$) being the output of $\bob$ (resp., $\charlie$). $\adversary$ outputs $(y_{\bob},y_{\charlie},\state')$. 
    
\end{itemize}

If the states $\state_{i,\secparam}$ are classical strings and $(\bob, \charlie)$ are classical algorithms, then $\alice$ is called a classical non-local stateful solver, otherwise it is called a quantum non-local stateful solver. If $\distrpr = \cU_{\bit^d \times \bit^d}$, then $\alice$ is called an \emph{independent-challenge} $P$-solver.
\end{definition}

\begin{remark}
    Even though the definition above of a non-local $(P, \distrpr)$-solver does not depend on $\distrpr$, it will be used to define the value and advantage of the solver below. 
\end{remark}

\paragraph{Examples of Non-Local Solvers:} For instance, if $P$ is inverting a one-way function $f$, then $\bob$ and $\charlie$ respectively get $x_\bob = f(z_\bob)$ and $x_\charlie = f(z_\charlie)$ as challenges, where $(z_\bob, z_\charlie)$ is pair of uniform inputs arbitrarily correlated depending on $\distrpr$. In order to pass verification, $\bob$ needs to output $y_\bob = z_\bob$ and $\charlie$ needs to output $y_\charlie = z_\charlie$. Another example is when $P$ is distinguishing the output of a PRG $G$ from a random string. In this case, $x_\bob$ and $x_\charlie$ are each either a uniform output of $G$ or a uniformly random string. In order to pass verification, $\bob$ and $\charlie$ need to correctly guess which one. The correct answer for $\bob$ and $\charlie$, may be arbitrarily correlated depending on $\distrpr$. Furthermore, $x_\bob, x_\charlie$ could be correlated still conditioned on the answers. \\ 

\noindent We formalize the interaction between a non-local solver and a non-local algorithm below. Formally, we adapt Definition 3.6 from~\cite{BBK22}.

\begin{definition}[Non-Local Solver Interaction] \label{def:nl_stateful_solver_interaction}
Let $P=(\gench,\ver,c)$ be a non-interactive assumption. For any stateful $(P, \distrpr)$-solver $\adversary=(\bob_{\adversary},\charlie_{\adversary},\state_0)$, and ${\cal S}=(\bob_{\cal S},\charlie_{\cal S})$, where $\bob_{\cal S}$ and non-local algorithm $\charlie_{\cal S}$ are a pair of algorithms, with input $z = (z_\bob, z_\charlie) \in \{0,1\}^*$, we consider the process ${\cal S}^{\adversary}_z(1^{\secparam})$ of the algorithm interacting with the solver. We define this process below. 
\begin{itemize}
    \item ${\cal S}$ is invoked on the input $(1^{\secparam},z_{\bob},z_{\charlie},\state_0)$, where $\bob_{\cal S}$ receives as input $z_{\bob}$ and $\charlie_{\cal S}$ receives as input $z_{\charlie}$. Moreover, $\state_0$ is defined on two registers $\bfB$ and $\bfC$, with ${\bob}_{\cal S}$ receiving the register $\bfB$ and ${\charlie}_{\cal S}$ receiving the register $\bfC$. Initialize $\tau_{\bob}$ to contain $z_{\bob}$ and $\tau_{\charlie}$ to contain $z_{\charlie}$. At every step $i \geq 1$: 
    \begin{enumerate}
        \item ${\bob}_{\cal S}(1^{\secparam},\tau_{\bob})$ makes a query $x_{\bob}^{(i)}$ and ${\charlie}_{\cal S}(1^{\secparam},\tau_{\bob})$ makes a query $x_{\charlie}^{(i)}$,
        \item Run $(y_{\bob}^{(i)},y_{\charlie}^{(i)},\state_{i}) \leftarrow (\bob_{\adversary}(1^{\secparam},1^i,x_{\bob}^{(i)},\cdot) \otimes \charlie_{\adversary}(1^{\secparam},1^i,x_{\charlie}^{(i)},\cdot))(\state_{i-1})$, where $y_{\bob}^{(i)}$ is $\bob_\alice$'s output, $y_{\charlie}^{(i)}$ is $\charlie_\alice$'s output and $\state_{i}$ is the joint output of $\bob_\alice$ and $\charlie_\alice$. 
        \item Add $(x_{\bob}^{(i)},y_{\bob}^{(i)})$ to $\tau_{\bob}$ and  $(x_{\charlie}^{(i)},y_{\charlie}^{(i)})$ to $\tau_{\charlie}$.
    \end{enumerate}
    \item At the end of the interaction, ${\cal S}$ may produce the output $(w_{\bob},w_{\charlie})$, where $w_{\bob}$ is output by ${\bob}_{\cal S}$ and $w_{\charlie}$ is output by ${\charlie}_{\cal S}$.
\end{itemize}
We will sometimes refer to $S$ as a solver-aided non-local algorithm.
\end{definition}

\paragraph{Non-Local Solver Value, Advantage, and Persistence.}
For stateful solvers, we explicitly define their value after interacting with the assumption many times.

\begin{definition}[Non-local Stateful Value and Advantage]
Let $P = (\gench, \ver, c)$ be a non-interactive assumption and let $\alice = (\bob_\alice, \charlie_\alice, \state_0)$ be a stateful $(P, \distrpr)$-non-local solver. Let $\cS$ be a solver-aided non-local algorithm with input $z = (z_\bob, z_\charlie)$. Let $\state_i$ be defined as in \Cref{def:nl_stateful_solver_interaction} for the interaction $\cS^{\alice}_z(1^\secparam)$. Then, we define the (many-shot) value of $\cS^{\alice}_z$ as \begin{align*}
    \val_P^{\distrpr} \bracS{ i, \cS_z^{\alice} } := \pr{ \substack{\ver(1^{\secparam},y_{\bob};r_{\bob})=1\\ \land \\\ver(1^{\secparam},y_{\charlie};r_{\charlie})=1} \ {\Bigg| } \ \substack{ (r_\bob, r_\charlie ) \from \distrpr \\ x_\bob = \gench(1^\secparam; r_\bob) \\ x_\charlie = \gench(1^\secparam; r_\charlie) \\ (y_{\bob},y_{\charlie},\state_{{i+1}}) \leftarrow (\bob_{\adversary}(1^{\secparam},1^{i+1},x_{\bob},\cdot) \otimes \charlie_{\adversary}(1^{\secparam},1^{i+1},x_{\charlie},\cdot))(\state_{i})} },
\end{align*}

and for any threshold $c^* : \N \to \R^+ \cup \{0\}$, we define the (many-shot) advantage of $\cS^{\alice}_z$ as \begin{align*}
    \adv_P^{(\distrpr, c^*)}\bracS{ i, {\alice} } := \abs{ \val_P^{\distrpr} \bracS{ i, \cS_z^{\alice} } - c^*(\secparam) }
\end{align*}

We define the one-shot value of the solver interaction $\cS^{\alice}_z$ as $\val_P^{\distrpr} \bracS{\alice} := \val_P^{\distrpr} \bracS{ 0, \cS_z^{\alice} }$ and the one-shot advantage of $\alice$ with respect to threshold $c^*$ as $\adv_P^{(\distrpr, c^*)}\bracS{\alice} := \adv_P^{(\distrpr, c^*)}\bracS{ 0, \cS_z^{\alice} }$. Note that neither value depends on $(\cS, z)$.

\end{definition}

\begin{remark}[Non-Local vs. Local Threshold]
Above, the non-local threshold $c^*$, which is a parameter of the advantage $\adv_P^{(\distrpr, c^*)}$ of the adversary $\alice$, can depend on the local threshold $c$ as well as the $d$-extension $\distrpr$.
\end{remark}

\begin{definition}[Non-Local Persistence] 
Let $P$ be a non-interactive assumption. A distribution $\mathbb{B}$ on $(P,\widetilde{\distr})$-non-local solvers $\bracC{\alice^\alpha = (\bob^\alpha_{\alice}, \charlie^\alpha_{\alice}, \state_0^\alpha)}_\alpha$ is called $(p, \eta)$-persistent if for any solver-aided non-local algorithm $\cS = (\bob_{\cS}, \charlie_{\cS})$ with input $z = (z_\bob, z_\charlie)$, we have \begin{align*}
    \pr{\sup_{i} \abs{ \val_P^{\distrpr} \bracS{ i, \cS_z^{\alice^\alpha} } - p} \le \eta} \ge 1 - \eta,
\end{align*}
where the probability is taken over the randomness of $\alpha$ and the solver interaction $\cS_z^{\alice^\alpha}$. 

\end{definition}

\subsubsection{Reductions} 
\noindent A reduction is a transformation that converts a solver for one problem to a solver for another (possibly different) problem. More precisely, suppose $P$ and $Q$ be two assumptions. A reduction leverages the solver for $P$ to design a solver for $Q$. We consider the setting when the solver for $Q$ uses the solver for $P$ as a black-box. As in~\cite{BBK22}, we consider reductions where the solver for $Q$ runs the solver for $P$ multiple times. We define analogous notions of one-shot advantage and persistent advantage below. 
\par We first define the traditional notion of reduction before defining non-local reductions.

\newcommand{\reduction}{{\cal R}}
\begin{definition}[Reduction; Definition 3.12 in~\cite{BBK22}]
A classical (resp., quantum) reduction from solving a non-interactive assumption $Q$ to solving a non-interactive assumption $P$ is an efficient classical (resp., quantum) uniform algorithm ${\cal R}$ with the following guarantee. 
\par For any classical (resp., quantum) $P$-solver $\adversary_{P}=(A_{P},\state_0)$ with one-shot advantage $\eps$ and runtime $T$, let $\state'_0=(\state_0,A_P,1^{1/\eps},1^T)$. Then, $\adversary_{Q}=(\reduction,\state'_0)$ is a classical (resp., quantum) $Q$-solver with one-shot advantage $\eps'=\poly(\eps,\allowbreak T^{-1},\allowbreak \secparam^{-1})$ and runtime $\poly(T,\eps^{-1},\allowbreak \secparam)$. An \emph{inefficient} reduction is defined the same way without the runtime requirement and the dependency on $T$.

\end{definition}

\paragraph{Non-Local Reductions.} We now generalize the above definition to the non-local setting. 

\begin{definition}[Non-Local Reduction]
A non-local classical (resp., quantum) $(\distrpr_P, c_P^*, \distrpr_Q, c_Q^*)$-reduction from solving a non-interactive assumption $Q$ to solving a non-interactive assumption $P$ is an efficient classical (resp., quantum) uniform non-local algorithm ${\cal R} = (\cR_1, \cR_2)$ with the following guarantee. 
\par Given any non-local classical (resp., quantum) $(P, \distrpr)$-solver $\adversary_{P}=(\bob_P, \charlie_P,\state_0)$ with one-shot advantage $\eps = \adv_P^{(\distrpr, c_P^*)}\bracS{\alice_P}$ and runtime $T$, let $\state'_0=(\state_0, \bob_P, \charlie_P, 1^{1/\eps},1^T)$.
Then, $\adversary_{Q}=(\reduction_1, \reduction_2,\state'_0)$ is a $(Q, \distrpr_Q)$-solver with one-shot advantage $\eps'= \adv_Q^{(\distrpr, c_Q^*)}\bracS{\alice_Q} = \poly(\eps,T^{-1},\secparam^{-1})$ and runtime $\poly(T,\eps^{-1},\secparam)$. Here, it is understood that description of $\bob_P$ (resp., $\charlie_P$) is included as part of $\reduction_1$'s (resp., $\reduction_2$'s) register. An \emph{inefficient} non-local reduction is defined the same way without the runtime requirement and the dependency of $T$.

\end{definition}

\begin{definition}[Non-Local Black-Box Reduction]
A non-local black-box $(\distrpr_P, c_P^*, \distrpr_Q, c_Q^*)$-reduction $\cR = (\cR_1, \cR_2)$ from solving $Q = (\gench_Q, \ver_Q, c_Q)$ to solving $P = (\gench_P, \ver_P, c_P)$ is a non-local $(\distrpr_P, c_P^*, \distrpr_Q, c_Q^*)$-reduction such that $\cR_1$ (resp., $\reduction_2$) interacts with $\bob_P$ (resp., $\charlie_P$) as defined in \Cref{def:nl_stateful_solver_interaction}.

\par We further say that $\reduction$ is \emph{non-adaptive} if $\cR_1$ (resp., $\reduction_2$) produces its queries to $\bob_P$ (resp., $\charlie_P$) independent of the answers to its previous queries (or equivalently all at once).

\end{definition}

\subsection{Main Theorem}

In our main theorem, we show how to relate the unclonable security of two cloning games $\game, \game'$ which satisfy a similarity condition in the setup phase as well as some additional properties. Before we state the formal theorem, we first describe how the notions of non-interactive assumption and unclonable security are connected by introducing relevant notation.

\paragraph{Assumption Induced by a Cloning Experiment.} Let $\game = \gamesetup$ be a stateful\footnote{Recall that stateful cloning games (\Cref{def:cloning_game_stateful}) generalize cloning games, so that this definition applies the same to cloning games that are not stateful.} cloning game with message space $\cM$. For every $m \in \cM$, message distribution $\cD_\cM$, and secret key $\sk$ in the support of $\setup(1^\secparam)$, we consider the induced non-interactive assumption $\assumption^{\game, \cD_\cM}_{ {\sk, m} } = (d,n,\ell,\gench^{\game, \cD_\cM}_{ {\sk, m} },\ver^{\game, \cD_\cM}_{ {\sk, m} },c)$. Here $d$ is the length of the random coins used by $\gench$, $n$ is the length of the challenges output by $\gench$, and $\ell$ is the length of the answers received by $\gench$. In addition, $\gench^{\game, \cD_\cM}_{ {\sk, m} }(1^\secparam; r) := \gench(\sk, m; r)$ and $\ver^{\game, \cD_\cM}_{ {\sk, m} } (1^\secparam, r, \ans) := \ver(\sk, m, \ch, \ans, r)$.\footnote{Note that this is well-defined since $\ch$ can be computed deterministically given $\sk, m, r$.} Finally, we will set $c := \trivprob(\game, \cD_\cM)$ for the message distribution $\cM$ considered for unclonable security. As before, we write $P^{\game, \cD_\cM}_{ {\sk, m} } = (\gench^{\game, \cD_\cM}_{ {\sk, m} } ,\ver^{\game, \cD_\cM}_{ {\sk, m} },c) $ to denote the assumption when the parameters $d,n,\ell$ are implicit.

\ \\

\begin{theorem}[Main Theorem] \label{thm:main_lifting_reductions_to_cloning_games}
    Let $\game = \gamesetup$ and $\game' = \gamesetupprime$ be cloning games with the same message space $\cM$. Let $\distr_\cM$ be a message distribution, and $\genchprime$ be an extension of $\gench$. Suppose that the following conditions are satisfied: \begin{enumerate}[label=(\arabic*)]
        \item $\setup = \setup'$ and $\tokengen = \tokengen'$. \label{item:1}
        \item $\game'$ is $\distr_\cM$-evasive. \label{item:2}
        \item Either $\game$ is $\distr_\cM$-evasive or $\genchprime = \genchprime_{\independent}$. \label{item:3}
        
    \end{enumerate}
    For $\sk$ in the support of $\setup$ and message $m \in \cM$, consider the induced non-interactive assumptions \begin{align*}
        P = P^{\game, \cD_\cM}_{ {\sk, m} } = (\gench^{\game, \cD_\cM}_{ {\sk, m} } ,\ver^{\game, \cD_\cM}_{ {\sk, m} },\trivprob(\game, \distr_\cM, \genchprime))
    \end{align*}
    and \begin{align*}
        Q = Q^{\game', \cD_\cM}_{ {\sk, m} } = (\gench^{\game', \cD_\cM}_{ {\sk, m} } ,\ver^{\game', \cD_\cM}_{ {\sk, m} },\trivprob(\game', \distr_\cM)).
    \end{align*}
    Suppose further that the following conditions are satisfied: \begin{enumerate}[label = (\alph*)]
        \item For all $(\sk, m)$, $P$ has a verifiably polynomial image. \label{item:a}
        \item For all $(\sk, m)$, there exists a classical non-adaptive black-box reduction $\cR$ from solving non-interactive assumption $Q$ to solving non-interactive assumption $P$. Furthermore, the description of $\cR$ does not depend on $(\sk, m)$.

        \label{item:b}
        \item $\game'$ has $(\distr_\cM,\negl)$ independent-challenge unclonable security. \label{item:c}
    \end{enumerate}
    
    Then, $\game$ has $(\distr_\cM, \negl, \genchprime)$ unclonable security. 
    
\end{theorem}

\noindent The proof of the main theorem consists of 3 technical steps: (1) going from the classical reduction given in \cref{item:b} to a classical non-local reduction, (2) going from non-local $(Q, \distrpr_Q)$-solver to independent-challenge $Q$-solver, and (3) lifting classical reduction to quantum. We will formally discuss each step in \Cref{sec:proof_of_main_thm}. We include the final proof below to aid the reader in reading \Cref{sec:proof_of_main_thm}.

\begin{proof}[Proof of \Cref{thm:main_lifting_reductions_to_cloning_games}]

Define $c_P = \trivprob(\game, \distr_\cM, \genchprime)$ and $c_Q = \trivprob(\game', \distr_\cM)$. By condition \ref{item:2}, $Q$ is evasive, i.e. $c_Q$ is negligible. By conditions \ref{item:2},\ref{item:b} and \Cref{lem:local_to_nl_clas_lifting}, there exists a classical non-adaptive black-box non-local $(\distrpr_O, c_P, \distrpr_Q, c_Q)$-reduction from solving $Q$ to solving $P$, where $\distrpr_P,\distrpr_Q$ denote uniform extensions. 

\par By \Cref{thm:nl_clas_to_qtm_lifting}, there exists a quantum non-local $(\distrpr_O, c_P, \distrpr_Q, c_Q)$-reduction $\cR'$ from solving $Q$ to solving $P$. Combining this with \Cref{cor:classical} if $\game$ is $\distr_\cM$-evasive, and trivially if $\genchprime = \genchprime_{\independent}$, $\cR'$ is also a quantum non-local $(\genchprime, c_P, \distrpr_Q, c_Q)$-reduction from solving $Q$ to solving $P$.

\par Now, let $\abc$ be a QPT adversary which breaks $(\distr_\cM, \varepsilon)$ unclonable security of $\game$ for a non-negligible function $\varepsilon$. Let $\varepsilon_{\sk, m}$ be the one-shot value of $\abc$ in assumption $P$ defined above with respect to $c_P$, so that $\abs{\E_{\sk, m}\bracS{\varepsilon_{\sk, m} - c_P}} > \varepsilon$. Note that since the description of $\cR$ does not depend on $(\sk,m)$, neither does the description of $\cR'$. 

\par Let $(\alice', \bob', \charlie')$ be the induced QPT adversary obtained by giving the solver $(\bob', \charlie', \alice(\rho_{\sk, m}))$ as input to $\cR$, where $\rho_{\sk, m} \from \tokengen(\sk, m)$ is the token received by $\alice$ in $\game$. That is, $\alice'$ is defined in terms of $\alice$ and how the state $\alice(\rho_{\sk, m})$ is modified by the reduction $\cR'$, whereas $(\bob', \charlie')$ is the non-local algorithm output by $\cR'$. Note that this is well-defined since the description of $\cR'$ does not depend on $(\sk, m)$. Define $\varepsilon'_{\sk, m}$ as the one-shot value of $\abcprime$ in $\game'$.
\par By the guarantee of $\cR'$, we have $\abs{\varepsilon'_{\sk, m} - c_Q} \ge q(\abs{\varepsilon_{\sk, m} - c_P})$ for all $(\sk, m)$ and some polynomial $q$. Without loss of generality, we can take $q(x) = x^{\beta}$ for some constant $\beta > 1$, so that $q$ is a monotone, convex function. Taking the expectation and using Jensen's Inequality, we obtain \begin{align*}
    q(\varepsilon) &< q\brac{\abs{\E_{\sk, m}\bracS{\varepsilon_{\sk, m} - c_P}}} \le q\brac{\E_{\sk, m}\bracS{\abs{\varepsilon_{\sk, m} - c_P}}} \le \E_{\sk, m}\bracS{q\brac{\abs{\varepsilon_{\sk, m} - c_P}}} \le \E_{\sk, m}\bracS{\abs{\varepsilon'_{\sk, m} - c_Q}} \\
    &\le \E_{\sk, m}\bracS{\varepsilon'_{\sk, m} + c_Q} \le 2c_Q + \abs{\E_{\sk, m}\bracS{ \varepsilon'_{\sk, m} - c_Q }},
\end{align*}
hence $\abcprime$ breaks the $(\distr_\cM, q(\varepsilon) - 2c_Q)$ unclonable security of $\game'$, which suffices for the proof since $c_Q$ is negligible and $\varepsilon$ is non-negligible.

\end{proof}

\paragraph{Alternate Versions of the Main Theorem.}
We give two more versions of the main theorem, which could be useful for applications. The first one (\Cref{thm:main_thm_it}) is for information theoretic applications, including our single-decryptor encryption construction in \Cref{sec:sde}. The second one (\Cref{thm:main_thm_general}) concerns the case when $\game, \game'$ do not have identical setup phases, yet they are equivalent modulo a local quantum transformation applied by $\alice$, hence relaxing condition \ref{item:1} of \Cref{thm:main_lifting_reductions_to_cloning_games}.

\par 
\begin{theorem}[Main Theorem - Inefficient Version] \label{thm:main_thm_it}
    Let $P, Q, \game, \game', \distr_\cM, \genchprime$ be defined as in \Cref{thm:main_lifting_reductions_to_cloning_games}. Suppose that the conditions in \cref{item:1,item:2,item:3} and \cref{item:b} are satisfied. If $\game'$ has information theoretic $(\distr_\cM,\negl)$ independent-challenge unclonable security, then $\game$ has information theoretic $(\distr_\cM, \negl, \genchprime)$ unclonable security.
\end{theorem}
\begin{proof}
    The proof follows the same as that of \Cref{thm:main_lifting_reductions_to_cloning_games}. The only difference is that since $P$ has no verifiably polynomial image, the non-local reduction is inefficient as per \Cref{thm:nl_clas_to_qtm_lifting}.
\end{proof}

\begin{theorem}[Main Theorem - General Version] \label{thm:main_thm_general}
    \Cref{thm:main_lifting_reductions_to_cloning_games} holds true if condition \ref{item:1} is replaced with the following: There exists a quantum channel $\Gamma$ such that, $$\tracedist{\rho_{\sk, m}}{\Gamma \brac{\rho'_{\sk, m}}} \le \negl(\secparam)$$
    for all key-message pairs $(\sk, m)$, where $\rho_{\sk, m} \from \tokengen(\sk, m)$ and $\rho'_{\sk, m} \from \tokengen'(\sk, m)$.
\end{theorem}

\begin{proof}[Proof (sketch)]
    Follows by modifying the proof of \Cref{thm:main_lifting_reductions_to_cloning_games} so that the non-local reduction first applies $\Gamma$.
\end{proof}

We note that the same generalization can be applied to \Cref{thm:main_thm_it} in the information theoretic setting.

\subsection{Proof of the Main Theorem.} \label{sec:proof_of_main_thm}

In this section, we give the technical details of the 3 steps of the proof of the main theorem.

\paragraph{Step 1: Lifting Classical Local Reduction to Non-Local Reduction.}
In \Cref{sec:local_to_nl}, we show that any classical reduction from solving $Q$ to solving $P$ implies a classical non-local reduction, where the assumption $Q$ is assigned the independent-challenge distribution.

\paragraph{Step 2: Converting Non-Local Solver to Independent-Challenge Solver.}
In \Cref{sec:step2}, we show that any non-local solver with negligible non-local threshold is an independent-challenge solver. This step is needed because the third step below requires independent-challenge mode to lift a classical non-local reduction to quantum.

\paragraph{Step 3: Lifting Classical Non-Local Reductions to Quantum Reductions.} \label{sec:lifting_nl_red}

In \Cref{sec:main_lemma,sec:nl_persistence,sec:nl_lifting_thm}, we adapt the work of \cite{BBK22}, which shows how to lift local non-adaptive black-box reductions to quantum, to the non-local setting. The main technical contribution occurs in \Cref{sec:main_lemma}, where we do non black-box analysis of the state-repair procedure of \cite{CMSZ21} to show that it can be used in the non-local setting.

\subsubsection{Upgrading Classical Local Reductions to Classical Non-Local Reductions} \label{sec:local_to_nl}
Most classical reductions known are for local assumptions, and we will show below how to obtain a non-local reduction from a local reduction. The non-local reduction simply runs the local reduction in each register.

\begin{lemma} \label{lem:local_to_nl_clas_lifting}
Suppose there exists a classical reduction $\cR$ from solving non-interactive evasive assumption $Q = (\gench_Q, \ver_Q, c_Q)$ to solving non-interactive assumption $P = (\gench_P, \ver_P, c_P)$. Let $\distrpr_Q = \cU_{\bit^{d_Q} \times \bit^{d_Q}}$ be the uniform extension. Then, there exists a classical non-local $(\distrpr_P, c_P, \distrpr_Q, c_Q)$-reduction $\cR' = (\cR'_1, \cR'_2)$ from solving $Q$ to solving $P$ for any $d_P$-extension $\distrpr_P$. Furthermore, $\cR'$ is non-adaptive (resp., black-box) if $\cR$ is non-adaptive (resp., black-box).
\end{lemma}
\begin{proof}
The non-local reduction locally runs the classical reduction, i.e. $\cR'_1 = \cR'_2 = \cR$. Let $\alice = (\bob, \charlie)$ be a $(P, \distrpr_P)$-solver with one-shot advantage $\eps = \adv_P^{(\distrpr_P, c_P)}\bracS{\alice}$. Let $\alice' = (\bob', \charlie')$ be the non-local adversary output by $\cR'$. Assume that $\alice, \alice'$ are both stateless. Since $\varepsilon + c_P$ is the probability of $(\bob, \charlie)$ simultaneously passing verification, it follows that $\bob$ and $\charlie$ are both $P$-solvers with advantage $\varepsilon$. Thus, by assumption $\bob'$ and $\charlie'$ are both $Q$-solvers with advantage $\varepsilon' = \poly(\varepsilon, T^{-1}, \secparam)$, where $T$ is the runtime of $\alice$. Therefore, since $\alice'$ is an independent-challenge non-local $Q$-solver, it follows that the advantage of $\alice'$ is at least $(\varepsilon' + c_Q)^2 - c_Q \ge (\varepsilon')^2 - \negl(\secparam)$ as desired.
\par The argument can be generalized to stateful non-local solvers by convexity.
\end{proof}

\subsubsection{From Non-Local Solvers to Independent-Challenge Solvers} \label{sec:step2}

For the second step of the proof of the main theorem, we will show that any non-local solver is an independent-challenge solver, up to polynomial loss.

\begin{theorem} \label{thm:quantum_ident_indep_ch}
    Let $P$ be a non-interactive assumption and let $c: \N \to \R^+ \cup \{0\}$ be a negligible non-local threshold. Then, any non-local quantum $(P, \distrpr)$-solver $\alice = (\bob, \charlie, \state_0)$ with one-shot advantage $\varepsilon = \adv_P^{(\distrpr, c)}\bracS{\alice}$ is a quantum independent-challenge $P$-solver with advantage $\varepsilon' = \adv_P^{(\distrpr_\independent, c)}\bracS{\alice} = \poly(\varepsilon) - \negl(\secparam)$, where $\distrpr_\independent$ is the uniform extension.
\end{theorem}
\begin{proof}
    Since $c$ is negligible, it suffices to show $\val_P^{(\distrpr_\independent, c)}\bracS{\alice} = \poly\brac{\val_P^{(\distrpr, c)}\bracS{\alice}}$. Since the success of a non-local solver can be implemented as a product measurement, this follows from \Cref{lem:ind-dep} below, with the polynomial above given by $q(x) = x^3/216.$
\end{proof}

\begin{lemma} \label{lem:ind-dep}
Let $\cR$ be a finite set of random coins, $\distrtild$ be a distribution over $\cR \times \cR$ with marginals $(\distr, \distr)$, and $\rho$ be a mixed bipartite state. Define the following quantities: \begin{align*}
        p_{cor} = \E_{(r,r') \from \distrtild} \tr \bracS{ \brac{B_r \otimes C_{r'} }\rho }, \quad p_{ind} = \E_{r, r' \from \distr} \tr \bracS{ \brac{B_r \otimes C_{r'} }\rho },
    \end{align*}
    where for each $r \in \cR$, we have $0 \le B_r \le I$ and $0 \le C_r \le I$. Then, \begin{align*}
        p_{id} \le 6 \sqrt[3]{p_{ind}}
    \end{align*}
\end{lemma}

\newcommand{\eucNorm}[1]{\left\|#1\right\|}

\begin{proof} 
    We can interpret $B_r$ and $C_r$ as POVM elements. Thus, by Naimark Dilation theorem, without loss of generality we can assume that $B_r$ and $C_r$ are projections. We can also assume that $\rho = \ketbraX{\psi}$ is a pure state, and the mixed state case follows by convexity, since $f(x) = 6 \sqrt[3]{x}$ is a concave function. \\
    
    \par Define PSD operators \begin{align*}
        P_B := \E_{r \from \distr} B_r, \quad P_C := \E_{r \from \distr} C_r,
    \end{align*}
    and let \begin{align*}
        \ket{\psi} = \sum_{i,j} \alpha_{i,j} \ket{\phi_i}\ket{\sigma_j}
    \end{align*}
    be the spectral decomposition of $\ket{\psi}$ with respect to $P_B \otimes P_C$, where $\ket{\phi_i}$ is an eigenvector of $P_B$ with eigenvalue $\lambda_i$ and $\ket{\sigma_j}$ is an eigenvector of $P_C$ with eigenvalue $\gamma_j$. Set $\eta = \sqrt[3]{p_{ind}}$ and define subnormalized states \begin{align*}
      \ket{\psi_B} = \sum_{\substack{i \; : \; \lambda_i \le \eta \\ j \; : \; \gamma_j > \eta }} \alpha_{i,j} \ket{\phi_i}\ket{\sigma_j}, \quad \ket{\psi_B} = \sum_{\substack{i \\ j \; : \; \gamma_j \le \eta }} \alpha_{i,j} \ket{\phi_i}\ket{\sigma_j}, \quad
      \ket{\psi_{BC}} = \sum_{\substack{i \; : \; \lambda_i > \eta \\ j \; : \; \gamma_j > \eta }} \alpha_{i,j} \ket{\phi_i}\ket{\sigma_j},
    \end{align*}
    so that $\ket{\psi} = \ket{\psi_B} + \ket{\psi_C} + \ket{\psi_{BC}}$.
    \\
    
    We first show the following claim:
    
    \begin{claim} \label{clm:weightbound}
    $ \eucNorm{\psi_{BC}}^2 \le \frac{p_{ind}}{\eta^2}$, where $\eucNorm{\cdot}$ is the Euclidean Norm.
    \end{claim}
    \begin{proof} 
    \begin{align*}
        p_{ind} &= \E_{r,r' \from \distr} \braket{ \psi | \brac{B_r \otimes C_{r'}} | \psi } \\
        &= \braket{ \psi | P_B \otimes P_C | \psi } \\
        &= \sum_{i,j} \abs{\alpha_{i,j}}^2 \lambda_i \gamma_j \\
        &\ge \sum_{\substack{i \; : \; \lambda_i > \eta \\ j \; : \; \gamma_j > \eta }} \abs{\alpha_{i,j}}^2 \lambda_i \gamma_j \\
        &\ge \eta^2 \sum_{\substack{i \; : \; \lambda_i > \eta \\ j \; : \; \gamma_j > \eta }} \abs{\alpha_{i,j}}^2 \\
        &= \eta^2 \eucNorm{\psi_{BC}}^2
    \end{align*}
        
    \end{proof}
    
    \par We have: \begin{align*}
        p_{cor} &= \E_{(r,r') \from \distrtild} \eucNorm{\brac{B_r \otimes C_{r'}}\brac{\ket{\psi_B} + \ket{\psi_C} + \ket{\psi_{BC}}}}^2 \\
        \text{(Cauchy-Schwartz)} \quad &\le \E_{(r,r') \from \distrtild} 3\brac{ \eucNorm{(B_r \otimes C_{r'})\ket{\psi_B}}^2 + \eucNorm{(B_r \otimes C_{r'})\ket{\psi_C}}^2 + \eucNorm{(B_r \otimes C_{r'})\ket{\psi_{BC}}}^2} \\
        &\le \E_{(r,r') \from \distrtild} 3\brac{ \eucNorm{(B_r \otimes I)\ket{\psi_B}}^2 + \eucNorm{(I \otimes C_{r'})\ket{\psi_C}}^2 + \eucNorm{\ket{\psi_{BC}}}^2} \\
        &= \E_{r, r' \from \distr} 3\brac{ \eucNorm{(B_r \otimes I)\ket{\psi_B}}^2 + \eucNorm{(I \otimes C_{r'})\ket{\psi_C}}^2 + \eucNorm{\ket{\psi_{BC}}}^2} \\
        &= 3 \brac{ \braket{ \psi_B | \brac{P_B \otimes I} \psi_B } + \braket{ \psi_C | \brac{I \otimes P_C} \psi_C } + \eucNorm{\psi_{BC}}^2} \\
        \text{(\Cref{clm:weightbound})} \quad &\le 3 \brac{ \eta\brac{ \eucNorm{\psi_B}^2 + \eucNorm{\psi_C}^2 } + \frac{p_{ind}}{\eta^2} } \\
        &\le 3 \brac{ \eta + \frac{p_{ind}}{\eta^2} } \\
        &= 6 \sqrt[3]{p_{ind}}
    \end{align*}

\end{proof}

\paragraph{Classical Case.}
We list the special case when $\alice$ in \Cref{thm:quantum_ident_indep_ch} is a classical algorithm explicitly below because (1) we will need it in \Cref{sec:cor_ch_dist} for our second application of the main theorem, and (2) we can show a quadratic (instead of cubic) loss for this special case.

\begin{corollary} \label{cor:classical}
    Let $P$ be a non-interactive assumption and let $c: \N \to \R^+ \cup \{0\}$ be a negligible non-local threshold. Then, any non-local classical $(P, \distrpr)$-solver $\alice = (\bob, \charlie, \state_0)$ with one-shot advantage $\varepsilon = \adv_P^{(\distrpr, c)}\bracS{\alice}$ is a classical independent-challenge $P$-solver with advantage $\varepsilon' = \adv_P^{(\distrpr, c)}\bracS{\alice} = \varepsilon^2 - \negl(\secparam)$.
\end{corollary}

\begin{proof}
    Although it follows as a special (classical) case of \Cref{thm:quantum_ident_indep_ch}, we will give a direct proof, as the classical concrete bound trumps the quantum bound.
    We will show that if the one-shot value of $\alice$ with respect to $\distrpr$ is $\delta$, then its one-shot value in the independent-challenge setting is at least $\delta^2$, which suffices for the proof. The classical state $\state_0$ can be modeled as a shared random string. We assume that $\alice$ is stateless and the general case follows by convexity.
    \par Let $\cR$ be the space of random coins used by the challenge generation algorithm $\gench$ of $P$. Let $p_r$ and $q_r$ be the probability that $\bob$ and $\charlie$ pass verification conditioned on the challenge being generated using randomness $r$, respectively. Since $\alice$ is stateless, we can write the one-shot value of $\alice$ in the independent-challenge setting as\footnote{This equation can be thought of as a no-signalling condition.} \begin{align*}
        \delta' = \frac{1}{|\cR|^2}\sum_{r,r'}p_rq_{r'}
    \end{align*} 
    On the other hand, the one-shot value of $\alice$ with respect to the extension $\distrpr$ can be written as \begin{align*}
        \delta = \sum_{r,r'} \alpha_{r,r'}p_rq_{r'},
    \end{align*}
    where $\bracC{\alpha_{r,r'}}_{r,r'}$ are non-negative coefficients satisfying \begin{align*}
        \sum_{r'} \alpha_{r',r} = \sum_{r'} \alpha_{r,r'} = \frac{1}{|\cR|}
    \end{align*}
    for all $r \in \cR$.
    Thus, we have \begin{align*}
        \delta^2 = \sum_{r,r',s,s'}\alpha_{r,r'}\alpha_{s,s'}p_rq_{r'}p_sq_{s'} &\le \sum_{r,r',s,s'}\alpha_{r,r'}\alpha_{s,s'}p_rq_{s'} \\
        &= \sum_{r,s'}p_rq_{s'} \sum_{r',s} \alpha_{r,r'}\alpha_{s,s'} \\
        &= \frac{1}{|\cR|^2} \sum_{r,s'}p_rq_{s'} \\
        &= \delta'
    \end{align*}
    as desired.
    
\end{proof}

\subsubsection{Main Lemma} \label{sec:main_lemma}

We will extend \Cref{lem:cmsz} to the non-local setting.

\begin{lemma}[Main Lemma] \label{lem:nl_cmsz}
Let $V$ be a classical algorithm and let $\adversary=\bob \otimes \charlie$ be a non-local algorithm, where $\bob$ (resp., $\charlie$) acts on register $\bfB$ (resp., $\bfC$). Define the following algorithms $(\nlvalest, \nlrepair)$ using the algorithms $\valest$ and $\repair$ described in~\Cref{lem:cmsz}:
\begin{itemize}
    \item $(\rho^*_{\bfB \bfC},p^*_{\bob},p^*_{\charlie}) \leftarrow  \nlvalest_{V,\adversary}(\rho,1^{1/\eps})$: compute $\left( \valest_{V,\bob}(\cdot,1^{1/(\eps/2)}) \otimes \valest_{V,\charlie}(\cdot,1^{1/(\eps/2)}) \right) (\rho_{\bfB \bfC})$  to obtain $(p^*_{\bob})$ ($\bob$'s output), $p^*_{\charlie}$ ($\charlie$'s output) and $\rho^*_{\bfB \bfC}$ is the joint output of $\bob$ and $\charlie$, where $\bob$ (resp., $\charlie$) output the register $\bfB$ (resp., $\bfC$). 
    \item $\sigma^*_{\bfB \bfC} \leftarrow \nlrepair_{V,\adversary,\Pi_{\bob},\Pi_{\charlie}} \left( \sigma_{\bfB \bfC},y_{\bob},p_{\bob},y_{\charlie},p_{\charlie},1^{1/\eps},1^{k} \right)$: for k-outcome projections $\Pi_\bob, \Pi_\charlie$ on register $\bfB,\bfC$, respectively, compute the following:
    $$\left( \repair_{V,\bob,\Pi_{\bob}} \left( \cdot,y_{\bob},p_{\bob},1^{1/\eps},1^{k} \right) \otimes \repair_{V,\charlie,\Pi_{\charlie}} \left( \cdot,y_{\charlie},p_{\charlie},1^{1/\eps},1^{k})\right) \right)(\sigma_{\bfB \bfC})$$
    Denote the result by $\sigma^*_{\bfB \bfC}$. 
\end{itemize}
\noindent Note that $(\nlvalest, \nlrepair)$ are efficient\footnote{Here, by efficient we mean that the runtime of $\nlvalest$ is polynomial in the size of $\rho$, $1/\varepsilon$, the runtime of $V$ and the runtime of $\alice$. Similarly for $\nlrepair$.} algorithms given that $(\valest, \repair)$ are efficient algorithms by \Cref{lem:cmsz}. In addition, we have the following guarantees:
\begin{enumerate}
    \item {\bf Value Estimation:} For any $\varepsilon > 0$, \begin{align} \label{eq:nl_valest}
        \E_{(\rho^*_{\bfB \bfC},p^*_{\bob},p^*_{\charlie}) \leftarrow  \nlvalest_{V,\adversary}(\bfrho_{\bfB \bfC},1^{1/\eps})} \bracS{ p_\bob^*p_\charlie^* } = \pr{ \substack{ V(y_\bob; r_\bob) = 1 \\ \land \\ V(y_\charlie; r_\charlie) = 1} \; : \; \substack{ r_\bob, r_\charlie \uniform \bit^d \\ (y_\bob, y_\charlie) \from \bob(\cdot, r_\bob) \otimes \charlie(\cdot, r_\charlie) (\bfrho_{\bfB \bfC}) } }.
    \end{align}
    \item {\bf Almost-Projective Estimation:} For any $\varepsilon \ge \varepsilon' > 0$, \begin{align*}
        \pr{ \abs{p_\bob^* - p_\bob^{**}} \geq \eps \ \lor \  \abs{p_\charlie^* - p_\charlie^{**}} \geq \eps \; : \; \substack{(\rho^*_{\bfB \bfC},p^*_{\bob},p^*_{\charlie}) \leftarrow  \nlvalest_{V,\adversary}(\bfrho_{\bfB \bfC},1^{1/\eps}) \\ (\rho^{**}_{\bfB \bfC},p^{**}_{\bob},p^{**}_{\charlie}) \leftarrow  \nlvalest_{V,\adversary}(\bfrho^*_{\bfB \bfC},1^{1/\eps'}) }  } \le \varepsilon.
    \end{align*}
    \item {\bf Repairing}: For any state $\bfrho_{\bfB \bfC}$ and any $\eps > 0$,
    \begin{align} \label{eq:nl_repairing}
    \prob \left[ \abs{p_\bob^* - p_\bob^{**}} \geq \eps \ \lor \  \abs{p_\charlie^* - p_\charlie^{**}} \geq \eps   \ :\ \substack{(\rho^*_{\bfB \bfC},p^*_{\bob},p^*_{\charlie}) \leftarrow  \nlvalest_{V,\adversary}(\bfrho_{\bfB \bfC},1^{1/\eps})\\ \ \\(\sigma_{\bfB \bfC},y_{\bob},y_{\charlie}) \leftarrow (\Pi_{\bob} \otimes \Pi_{\charlie})(\rho^*_{\bfB \bfC})\\ \ \\ \sigma^*_{\bfB \bfC} \leftarrow \nlrepair_{V,\adversary,\Pi_{\bob},\Pi_{\charlie}}\left( \sigma_{\bfB \bfC},y_{\bob},p^*_{\bob},y_{\charlie},p^*_{\charlie},1^{1/\eps},1^{k} \right)\\ \ \\ (\rho^*_{\bfB \bfC},p^*_{\bob},p^*_{\charlie}) \leftarrow  \nlvalest_{V,\adversary}(\bfrho^{*}_{\bfB \bfC},1^{1/\eps}) } \right] \leq \eps. \end{align}
\end{enumerate}
\end{lemma}

\begin{proof}
\par We start with value estimation. Let $\cM^\bob_0 = (\Pi_A, I - \Pi_A)$ and $ \cM^\bob_1 = (\Pi_B, I - \Pi_B)$ be the projective measurements applied by $\valest_{V,\bob}(\cdot, 1^{\varepsilon/2})$. Using \Cref{lem:jordans_lemma}, we decompose the $\bfB$ register into Jordan subspaces as $\cH_\bfB = \bigoplus S_j$. Let $\bracC{\ket{v^A_{j,1}}, \ket{v^A_{j,0}}}$ and $\bracC{\ket{v^B_{j,1}}, \ket{v^B_{j,0}}}$ be the orthogonal eigenbases of $S_j$ with respect to $\Pi_A, \Pi_B$.
\par Similarly for $\valest_{V,\charlie}(\cdot, 1^{\varepsilon/2})$, we consider projective measurements $\cM^\charlie_0 = (\Gamma_A, I - \Gamma_A)$ and $ \cM^\charlie_1 = (\Gamma_B, I - \Gamma_B)$, with the Jordan decomposition $\cH_\bfC = \bigoplus T_{j'}$ and the eigenbases $\bracC{\ket{w^A_{j',1}}, \ket{w^A_{j',0}}}$ and $\bracC{\ket{w^B_{j',1}}, \ket{w^B_{j',0}}}$. For simplicity, we will assume $\dim S_j = \dim T_{j'} = 2$ for all $j,j'$. The degenerate cases can be handled similarly\footnote{See \cite{CMSZ21} Section 4.1 to see how to handle degenerate subspaces generically. }.
\par We will need the following claim:
\begin{claim} \label{clm:jordan_commute}
    Let $\Pi^{\jord}_\bfB = \bracC{\Pi^{\jord}_{\bfB, j}}_j$ be the jordan subspace measurement, with $\Pi^{\jord}_{\bfB, j} = \ketbraX{v^A_{j,0}} + \ketbraX{v^A_{j,1}}$. Similarly define $\Pi^{\jord}_\bfC$. Then, $\Pi^{\jord}_\bfB$ commutes with $\valest_{V, \bob}(\cdot, 1^{\varepsilon/2})$. Similarly, $\Pi^{\jord}_\bfC$ commutes with $\valest_{V, \charlie}(\cdot, 1^{\varepsilon/2})$.
\end{claim}

\begin{proof}
    We will show that $\Pi^{\jord}_\bfB$ commutes with $\valest_{V, \bob}(\cdot, 1^{\varepsilon/2})$. Recall that $\Pi^{\jord}_\bfB$ commutes with both $\cM^\bob_0$ and $\cM^\bob_1$. We write the register $\bfB$ as two registers, $\bfP$ and $\bfQ$, where $\bfP$ stores the input state and $\bfQ$ contains all auxiliary registers used in the computation. In particular, this means that the measurements $\cM^\bob_0, \cM^\bob_0, \Pi^{\jord}_\bfB$ are all restricted to the $\bfP$ register. By \Cref{lem:cmsz} bullet 3, every step of $\valest_{V, \bob}(\cdot, 1^{\varepsilon/2})$ falls into one of the following categories: \begin{enumerate}
        \item Apply $\cM^\bob_b$ on the $\bfP$ register conditioned on a qubit located in the $\bfQ$ register for some $b \in \bit$. This captures all steps adaptively applying a projective measurement $\cM^\bob_b$.
        \item Apply a local operation on the $\bfQ$ register. This captures the steps computing the database of outcomes and the functions $f,g$ described in \Cref{lem:cmsz}.
    \end{enumerate}

    Because $\Pi^{\jord}_\bfB$ is a projection applied to the $\bfP$ register alone, it commutes with operations from both categories above.
\end{proof}

Without loss of generality assume that $\rho$ is a pure state, for the mixed state case follows by convexity. By assumption, the bipartite state after steps (a) and (b) described in the 2-Projection Implementation property in \Cref{lem:cmsz} are performed on both registers has the form \begin{align*}
    \sum_{j,j'} \alpha_{j,j'} \ket{v^B_{j,1}}\ket{w^B_{j',1}}
\end{align*}
Let $\widetilde{\valest}$ denote the execution of $\valest$ after steps (a) and (b) mentioned above. By \Cref{clm:jordan_commute}, the product Jordan subspace measurement $\Pi^{\jord}_\bfB \otimes \Pi^{\jord}_\bfC$ commutes with $\nlvalest_{V,\alice}(\cdot, 1^{1/\eps})$, hence the values $p_\bob^*, p_\charlie^*$ are distributed as follows: \begin{itemize}
    \item Sample $(j,j')$ with probability $|\alpha_{j,j'}|^2$,
    \item Compute $p_\bob^* \from \widetilde{\valest}_{V,\bob}(\ket{v^B_{j,1}}, 1^{1/(\eps/2)})$ and $p_\charlie^* \from \widetilde{\valest}_{V,\charlie}(\ket{v^B_{j',1}}, 1^{1/(\eps/2)})$,
    \item Output $(p_\bob^*, p_\charlie^*)$.
\end{itemize}

Let $p_\bob^j$ be the value of state $\ket{v^B_{j,1}}$ as defined on the RHS of \cref{eq:valest}. By the same equation, we know that $\E \bracS{p_{\bob}^* | j} = p_{\bob}^j$. Similarly define $p_{\charlie}^{j'}$ so that $\E \bracS{p_{\charlie}^* | j'} = p_{\bob}^{j'}$. Let $p^{j,j'}$ denote the value of state $\ket{v^B_{j,1}}\ket{v^C_{j',1}}$ as defined on the RHS of \cref{eq:nl_valest}. Because this is pure product state, we have that $p^{j,j'} = p_{\bob}^j p_{\charlie}^{j'}$, and also that the variables $p_{\bob}^*, p_{\charlie}^*$ are independent conditioned on $j,j'$. On the other hand, since computing the value commutes with the Jordan measurements, the value $p$ for the state $\rho$ is distributed as follows: \begin{itemize}
    \item Sample $(j,j')$ with probability $|\alpha_{j,j'}|^2$,
    \item Output $p^{j,j'}$.
\end{itemize}

Therefore, putting everything together, we have \begin{align*}
    \E \bracS{ p_\bob^*p_\charlie^* } &= \sum_{j,j'} \abs{\alpha_{j,j'}}^2 \E \bracS{ p_\bob^* p_\charlie^* | j,j' } = \sum_{j,j'} \abs{\alpha_{j,j'}}^2 \E \bracS{ p_\bob^* | j} \cdot \E \bracS{ p_\charlie^* | j' } \\
    &= \sum_{j,j'} \abs{\alpha_{j,j'}}^2 p_\bob^j p_\charlie^{j'} = \sum_{j,j'} \abs{\alpha_{j,j'}}^2 p^{j,j'} = p
\end{align*}
as desired. Next, we show almost-projective estimation. Recall that $\nlvalest_{V,\adversary}(\cdot, 1^{1/\eps})$ applies $$\left( \valest_{V,\bob}(\cdot,1^{1/(\eps/2)}) \otimes \valest_{V,\charlie}(\cdot,1^{1/(\eps/2)}) \right),$$ and the operations on the $\bfB, \bfC$ registers commute. Therefore, using condition (2) of \Cref{lem:cmsz} for the state $\tr_{\bfC}\brac{\bfrho_{\bfB \bfC}}$, we have $\pr{ \abs{p_\bob^* - p_\bob^{**}} \geq \eps } \le \eps/2$, and similarly $\pr{ \abs{p_\charlie^* - p_\charlie^{**}} \geq \eps } \le \eps/2$. Using a union bound, we achieve the desired property. The repairing property follows by a similar argument, after observing that the procedure in \cref{eq:nl_repairing} involves two parallel executions of \cref{eq:repairing} on registers $\bfB$ and $\bfC$.
\end{proof}

\subsubsection{Achieving Non-Local Persistence} \label{sec:nl_persistence}
In Theorem 4.1 of \cite{BBK22}, the output $p^*$ of the value estimation $\valest$ estimates the value of the solver in expectation, and can be used to achieve a distribution over $(p^*, \eta)$-persistent solvers. We observe that the estimation $p^*$ could be stored non-locally $(p_\bob^*, p_\charlie^*)$ such that $p_\bob^* p_\charlie^* = p^*$. These local values are enough for non-local state repair by \Cref{lem:nl_cmsz}, and using this the proof of \cite{BBK22} can be straightforwardly applied to the non-local setting. We state the formal theorem statement below.
\begin{theorem}[Non-Local Persistence Theorem] \label{thm:nl_persistence}
Let $P$ be a non-interactive falsifiable assumption with a verifiably polynomial image. Let $\eta = \frac 1\poly$. There exist efficient non-local quantum algorithms $S = S_1 \otimes S_2, R = R_1 \otimes R_2$ with the following syntax and guarantee: \begin{itemize}
    \item $S^{\alice}(\state_0) = \brac{S_1^\bob \otimes S_2^\charlie} (\state_0)$ takes as input a non-local algorithm $\alice = (\bob, \charlie)$ and a bipartite state $\state_0$. It outputs a state $\state_0^*$ and values $(p_\bob^*, p_\charlie^*)$, where $p_\bob^*$ is output by $S_1^\bob$ and $p_\charlie^*$ is output by $S_2^\charlie$. 
    \item $R^\alice(1^\secparam, 1^i, x_\bob, x_\charlie, \state_{i-1}^*) = \brac{ R_1^\bob(1^\secparam, 1^i, x_\bob, \cdot) \otimes R_2^\charlie(1^\secparam, 1^i, x_\charlie, \cdot)}(\state_{i-1}^*)$ takes as input a non-local algorithm $\alice$, a security parameter $\secparam$, a step $i$, a pair of inputs $(x_\bob, x_\charlie)$, and a bipartite state $\state_{i-1}^*$. It outputs a pair of solutions $(y_\bob, y_\charlie)$, where $y_\bob$ is output by $R_2^\bob$ and $y_\charlie$ is output by $R_2^\charlie$, and a bipartite state $\state_i^*$.
\end{itemize}

For any non-local independent-challenge $P$-solver $\alice = (\bob, \charlie, \state_0)$ with one-shot value $p = \val_P^{\distrpr} \bracS{\alice}$, with $\distrpr = \cU_{\bit^d \times \bit^d}$, considering the random variable $(\state_0^*, p_\bob^*, p_\charlie^*) \from S^\alice(\state_0)$, we have: \begin{enumerate}
    \item $\E\bracS{p_\bob^* p_\charlie^*} = p$
    \item $\cR^* = (R^\alice, \state_0^*)$ sampled in this process is a distribution over efficient stateful non-local solvers that is $(p_B^*p_C^*, \eta)$-persistent.
\end{enumerate}

\noindent Moreover, if $P$ does not have a verifiably polynomial image, the same holds, but $S$ is not efficient.

\end{theorem}

\begin{proof}

    The proof will closely follow the proof of Theorem 4.1 in \cite{BBK22}. Let $P = (\gench, \ver, c)$ with polynomial-image verifier $K$ and polynomial-image bound $k$. Let $\eta = 1/\poly(\cdot)$. We will denote by $\bfB, \bfC$ the two registers for the non-local algorithms $S, R, \alice$. For $x \in \bit^n$, we define the $(k+1)$-outcome projective measurement $\Pi^\bob_x$ as follows: \begin{itemize}
        \item Coherently run $\bob$ with input $x$.
        \item First, measure if the output of $\bob$ is in the polynomial image specified by $K$. \begin{itemize}
            \item If not, output $\bot$
            \item If yes, measure the output of $\bob$ and output the answer $y \in \bit^m$.
        \end{itemize}
        \item Uncompute.
    \end{itemize} 

    \noindent We similarly define the projective measurement $\Pi^\charlie_x$. We also define the wrapper non-local solver $\widetilde{\alice} = \widetilde{\bob} \otimes \widetilde{\charlie}$, where $\widetilde{\bob}$ (resp., $\widetilde{\charlie}$) takes as input $r \in \bit^d$, computes $x = \gench(1^\secparam; r)$ and runs $\bob$ (resp., ${\charlie}$) on input $x$, outputting the result $y \in \bit^m$. Next, we move on to describe $S^\alice$ and $R^\alice$ next.

    \paragraph{Description of $S^{\alice}$:} On input $\state_0$, compute $(\state_0^*, p_{\bob}^*, p_{\charlie}^*) \from \nlvalest_{\ver, \widetilde{\alice}}(\state_0, 1^{10/\eta})$. That is, $S_1^\bob$ runs $\valest_{\ver, \widetilde{\bob}}(\cdot, 1^{20/\eta})$ and $S_2^\charlie$ runs $\valest_{\ver, \widetilde{\charlie}}(\cdot, 1^{20/\eta})$. Output $(\state_0^*, p_{\bob}^*, p_{\charlie}^*)$.

    \paragraph{Description of $R^{\alice}$:} On input $(1^\secparam, 1^i, x_\bob, x_\charlie, \state_{i-1}^*)$, where $R_1$ receives $1^\secparam, 1^i,x_{\bob}$ and the register $\bfB$ while $R_2$ receives $1^\secparam, 1^i,x_{\charlie}$ and $\bfC$, do the following:
    \begin{itemize}
    \item Define $\eps_i=\eta/32i^2$, so that $(\varepsilon_i)_i$ is monotonically decreasing and $\sum_{i=1}^\infty \varepsilon_i < \eta/16$.   
    \item Compute $ (\sigma_{i-1}, q_{\bob}^{i-1}, q_\charlie^{i-1}) \from \nlvalest_{\ver, \widetilde{\alice}}(\state_{i-1}^*,1^{1/\eps_i})$.
    \item Compute $(y^i_{\bob}, y_{\charlie}^i, \sigma_i^*) \from (\Pi_{x_\bob}^\bob \otimes \Pi_{x_\charlie}^\charlie)(\sigma_{i-1}) $, where $y_{\bob}^{i}, y_\charlie^{i} \in \bit^m \cup \bracC{\bot}$ and $\sigma_i^*$ is the post-measurement state.

    \item Compute $\rho_i \from \nlrepair_{\ver, \widetilde{\alice}, \Pi^\bob_{x_\bob}, \Pi^\charlie_{x_\charlie}}(\sigma_i^*,y^i_{\bob},q_\bob^{i-1}, y_{\charlie}^i, q_\charlie^{i-1}, 1^{1/\varepsilon_i, 1^{k+1}})$.

    \item Compute $( \state_i^*, p^{i}_{\bob}, p_{\charlie}^{i})\from \nlvalest_{\ver, \widetilde{\alice}}(\rho_i, 1^{1/\varepsilon_i})$. 

    \item Output $(y_{\bob},y_{\charlie},\state_i^*)$, where $y_{\bob}$, register $\bfB$ is $R_1$'s output while $y_{\charlie}$, register $\bfC$ is $R_2$'s output. 
    \end{itemize}
    Note that $R_1$ and $R_2$ are defined implicitly above given that every step is a non-local algorithm; $R_1^\bob$ performs the operations on the $\bfB$ register, and $R_2^\charlie$ on the $\bfC$ register. \\

    \noindent Consider a solver-aided non-local algorithm $\cS^{\cR^*}$ with input $z$, and the random variables sampled by the process below:

\begin{itemize}
    \item $(p_\bob^*, p_\charlie^*, \state_0^*) \from S^\alice(\state_0)$. Equivalently, $(\state_0^*, p_{\bob}^*, p_{\charlie}^*) \from \nlvalest_{\ver, \widetilde{\alice}}(\state_0, 1^{2/\eta})$.
    \item For $i \ge 1$, $(y^i_{\bob}, y_{\charlie}^i, \state_i^*) \from R^\alice(1^\secparam, 1^i, x_\bob^i, x_\charlie^i, \state_{i-1}^*)$, where $(x_\bob^i, x_\charlie^i)$ is the $i$th pair of inputs chosen by $\cS_z^{\cR^*}$. In more detail, \begin{itemize}
        \item $(\sigma_{i-1}, q_\bob^{i-1}, q_\charlie^{i-1}) \from \nlvalest_{\ver, \widetilde{\alice}}(\state_{i-1}^*, 1^{1/\varepsilon_i})$.
        \item $(y^i_{\bob}, y_{\charlie}^i, \sigma_i^*) \from (\Pi_{x_\bob^i}^\bob \otimes \Pi_{x_\charlie^i}^\charlie)(\sigma_{i-1}) $.
        \item $\rho_i \from \nlrepair_{\ver, \widetilde{\alice}, \Pi^\bob_{x_\bob^i}, \Pi^\charlie_{x_\charlie^i}}(\sigma_i^*,y^i_{\bob},q_\bob^{i-1}, y_{\charlie}^i, q_\charlie^{i-1}, 1^{1/\varepsilon_i, 1^{k+1}})$.
        \item $( \state_i^*, p^{i}_{\bob}, p_{\charlie}^{i})\from \nlvalest_{\ver, \widetilde{\alice}}(\rho_i, 1^{1/\varepsilon_i})$.
    \end{itemize}
\end{itemize}

\noindent Over the randomness of this process, we have the following identities: \begin{align}
    &\E\bracS{p_\bob^* p_\charlie^*} = p \label{eq:7} \\ 
    &\pr{\abs{p_\bob^* - q_\bob^0} \ge \eta/10 \lor \abs{p_\charlie^* - q_\charlie^0} \ge \eta/10} \le \eta/10  \label{eq:8} \\
    &\pr{\abs{p_\bob^i - q_\bob^i} \ge \varepsilon_i \lor \abs{p_\charlie^i - q_\charlie^i} \ge \varepsilon_i} \le \varepsilon_i, \quad \forall i \ge 1  \label{eq:9} \\
    &\pr{\abs{q_\bob^{i-1} - p_\bob^i} \ge \varepsilon_i \lor \abs{q_\charlie^{i-1} - p_\charlie^i} \ge \varepsilon_i} \le \varepsilon_i, \quad \forall i \ge 1  \label{eq:10} \\
    &\E\bracS{q_\bob^{i-1} q_\charlie^{i-1}} = \val_P^{\distrpr} \bracS{ i-1, \cS_z^{\cR^*}}, \quad \forall i \ge 1  \label{eq:11} \\
    &\pr{\abs{p_\bob^{i-1} p_\charlie^{i-1} - \val_P^{\distrpr} \bracS{ i-1, \cS_z^{\cR^*} } } \ge \varepsilon_i } \le \varepsilon_i, \quad \forall i \ge 1, \cS, z  \label{eq:12}
\end{align}

\Cref{eq:7} follows from the value estimation property in \Cref{lem:nl_cmsz} and the fact that $\widetilde{\alice}$ is a wrapper solver for $\alice$. This is the first bullet we need to show to prove the theorem. \\
\par Moving on to the second bullet (persistence), \cref{eq:8,eq:9} follow from the almost-projective estimation property in \Cref{lem:nl_cmsz} since $\varepsilon \le \eta/2$ and $\varepsilon_{i+1} \le \varepsilon_i$. \Cref{eq:10} follows from the repairing property. Finally, \cref{eq:12} follows from the value estimation property and from the fact that\footnote{Here we define $\ver(1^\secparam, r, \bot) = 0$ for convenience.} $\val_P^{\distrpr} \bracS{ i-1, \cS_z^{\cR^*} } = \pr{ \ver(1^\secparam, r_\bob^i, y_\bob^i) = \ver(1^\secparam, r_\charlie^i, y_\charlie^i) = 1 }$. The latter is true because $y_\bob^i$ (resp., $y_\charlie^i$) is distributed the same as the output of $\bob$ (resp., $\charlie$) conditioned on $y_\bob^i \ne \bot$, and otherwise the verification fails as expected. \\

\par Therefore, by union bound we have with probability at least $1 - \eta/10 - \sum_{i=1}^\infty 4\varepsilon_i \ge 1 - \eta$ that \begin{enumerate}
    \item $\abs{p_\bob^* - q_\bob^0} \le \eta/10$ and $\abs{p_\charlie^* - q_\charlie^0} \le \eta/10$,
    \item $\abs{p_\bob^i - q_\bob^i} \le \varepsilon_i$ and $\abs{p_\charlie^i - q_\charlie^i} \le \varepsilon_i$,
    \item $\abs{p_\bob^i - q_\bob^{i-1}} \le \varepsilon_i$ and $\abs{p_\charlie^i - q_\charlie^{i-1}} \le \varepsilon_i$,
    \item $\E\bracS{q_\bob^{i-1} q_\charlie^{i-1}} = \val_P^{\distrpr} \bracS{ i-1, \cS_z^{\cR^*}}$
\end{enumerate}
for all $i \ge 1$. Conditioned on the inequalities 1-4 above being true, we will show that
\begin{align}
    \abs{p_\bob^* p_\charlie^* - \val_P^{\widetilde{\distr}}\bracS{i, \cS_z^{\cR^*}}} \le \eta \label{eq:persistence}
\end{align}
to finish the proof. Indeed, by triangle inequality we have that $\abs{p_\bob^* - q_\bob^i} \le \eta/10 + \sum_{j=1}^i 2\varepsilon_j \le \eta/3$, and similarly $\abs{p_\charlie^* - q_\charlie^i} \le \eta/3$. Since $p_\bob^*, p_\charlie^* \in [0,1]$, this implies that $\abs{p_\bob^* p_\charlie^* - q_\charlie^i q_\charlie^i} \le 2(\eta/3) + (\eta/3)^2 \le \eta$. Taking the expectation, we obtain \cref{eq:persistence}.

\par If $P$ does not have a verifiably polynomial image, note that we can always take $k = 2^n$, i.e. $P$ always has a verifiably \emph{exponential} image. Thus, the proof remains the same, except the runtime of $R$ now depends exponentially on $n$.
    
\end{proof}

\begin{remark} \label{rem:redundant_valest}
We note that the transformation in \Cref{thm:nl_persistence} is not optimal in terms of concrete efficiency. For instance, it is redundant to apply value estimation once in $S$ and twice in $R$. The reason we do this is to use previous work as closely to a black-box as possible, thus keeping the proof simple.
\end{remark}

\subsubsection{Non-Local Classical-to-Quantum Lifting Theorem} \label{sec:nl_lifting_thm}

The remaining steps to reach our non-local lifting theorem (\Cref{thm:nl_clas_to_qtm_lifting}) are essentially the same as those in \cite{BBK22}. For this reason, we do not repeat the detailed proofs here. Nonetheless, we note that the proof has the following (informal) structure: given a $P$-solver, \begin{enumerate}
    \item Obtain a persistent $P$-solver. \label{step:1}
    \item Next, obtain a memoryless $P$-solver, i.e. one that only remembers the number of times it has interacted with the assumption. \label{step:2}
    \item Next, obtain a stateless solver. \label{step:3}
\end{enumerate} Step \ref{step:1} corresponds to \Cref{thm:nl_persistence}, whereas steps \ref{step:2},\ref{step:3} can be summarized below:

\begin{theorem} \label{thm:persistence_to_stateless}
There exists an efficient non-local quantum algorithm $\simulator = \simulator_1 \otimes \simulator_2$ with the following properties. Let $\alice = (\bob, \charlie, \state_0)$ be a $(p, \eta)$-persistent $(\ell, \ell_\bob, \ell_\charlie)$-stateful non-local independent-challenge $P$-solver for a falsifiable non-interactive assumption $P$. Let $\vec{x}_\bob$ and $\vec{x}_\charlie$ each be (independently) sampled from an efficiently samplable distribution over $k$-tuples of $P$-instances and define $\vec{x} := (\vec{x}_\bob, \vec{x}_\charlie)$. Then, there exists a $(p, \eta)$-persistent distribution $\bracC{ \alice^\alpha = (\bob^\alpha, \charlie^\alpha)}_\alpha$ over stateless solvers such that $$\simulator^{\alice}(1^\secparam, 1^\ell, 1^{1/\delta}, \Vec{x}^*) = \simulator_1^{\bob}(1^\secparam, 1^{\ell_\bob}, 1^{1/\delta}, \Vec{x_\bob}^*) \otimes \simulator_2^{\charlie}(1^\secparam, 1^{\ell_\charlie}, 1^{1/\delta}, \Vec{x_\charlie}^*)$$

makes non-adaptive black-box queries to $\bob$ and $\charlie$ and produces a distribution within $\delta$ statistical distance from $\alice^\alpha(1^\secparam, \vec{x})$.

\end{theorem}

\begin{proof}[Proof (sketch)]
Follows from Corollary 6.2 in \cite{BBK22} after making the observation that the simulators in Theorem 5.1 and Theorem 6.1 both preserve non-locality, i.e. they are non-local simulators if given as input non-local solvers.
\end{proof}

The theorem below is a direct consequence of \Cref{thm:nl_persistence,thm:persistence_to_stateless} as per Theorem 7.1 of \cite{BBK22}. We refer the reader to \cite{BBK22} for the proof.

\begin{theorem}[Non-Local Classical-to-Quantum Lifting] \label{thm:nl_clas_to_qtm_lifting}
Let $P = (\gench_P, \ver_P, c_P)$ and $Q = (\gench_Q, \ver_Q, c_Q)$ be non-interactive assumptions. Assume there exists a non-adaptive non-local classical black-box $(\distrpr_P, c_P, \distrpr_Q, c_Q)$-reduction from solving $Q$ to solving $P$, where $\distrpr_P = \cU_{\bit^{d_P} \times \bit^{d_P}}$ is the uniform extension. Then, there exists an inefficient non-local quantum $(\distrpr_P, c_P, \distrpr_Q, c_Q)$-reduction from solving $Q$ to solving $P$. If $P$ has verifiably polynomial image, then the reduction is efficient.
\end{theorem}

\subsection{Applications}

In this section, we present two applications of our main theorem (\Cref{thm:main_lifting_reductions_to_cloning_games,thm:main_thm_it}). The first one is information theoretic single-decryptor encryption for one-bit messages, which relies on simultaneous quantum Goldreich-Levin extraction. The second application involves the effect of changing the challenge distribution in a cloning game to unclonable security. Informally, we show that independent-challenge unclonable security implies unclonable security for correlated distributions for primitives which have negligible trivial success probability.

\newcommand{\vardbtilde}[1]{\tilde{\raisebox{0pt}[0.85\height]{$\tilde{#1}$}}}
\newcommand{\tldabc}{(\vardbtilde{\alice}, \vardbtilde{\bob}, \vardbtilde{\charlie})}

\subsubsection{Single-Decryptor Encryption} \label{sec:sde}

In this section, we analyze a corollary of \Cref{thm:main_thm_it}, which is Simultaneous Quantum Goldreich-Levin Lemma. Then, we show a construction of information theoretic single-decryptor encryption as a corollary.

\begin{definition}[Simultaneous Extraction]
We say that an adversary $\abc$ can \emph{simultaneously extract} a classical function $f(k,x,r)$ given a quantum token $\rho_{k,x}$ using key $k = (k_B, k_C)$ with probability $\delta$ if $\abc$ succeeds in the following experiment with probability $\delta$:
\begin{itemize}
    \item In phase 1, the challenger sends $\rho_{k,x}$ to $\alice$, who applies a CPTP map to split the state into two registers $\bob$ and $\charlie$.
    \item In phase 2, $\bob$ and $\charlie$ can no longer communicate. The challenger samples independent random coins $(r,r') \uniform \cR \times \cR$, then sends $(r,k_B)$ to $\bob$ and $(r', k_C)$ to $\charlie$. Later, $\bob$ outputs a string $z_B$ and $\charlie$ outputs a string $z_C$.
    \item The adversary wins if $z_B = f(k,x,r)$ and $z_C = f(k,x,r')$.
\end{itemize}
It is understood that $x$ above is a random variable.

\end{definition}
\begin{lemma}[Simultaneous Quantum Goldreich-Levin] \label{lem:quantumGL}
Suppose that $\cR = \bit^{n}$ and there exists a (possibly inefficient) adversary $\abc$ which can simultaneously extract $\inner{r}{x}$ from a quantum token $\rho_{k,x}$ using key $k$, with probability $1/2 + \varepsilon$ for a non-negligible function $\varepsilon$. Then, there exists an adversary $\abcprime$ which can simultaneously extract\footnote{We note that the values $(r,r')$ can be ignored when extracting $x$.} $x$ given the same token $\rho_{k,x}$ using key $k$ with non-negligible probability $\varepsilon'$. \\

\end{lemma}

\begin{proof}
    We will interpret the key $k = (k_B, k_C) =: (\ch_\bob, \ch_\charlie)$ as the challenge and $x =: m$ as the message in a cloning game. We will assume here that $k_B, k_C$ are independently distributed, which is sufficient for the application to single-decryptor encryption. See \Cref{sec:app_gl} for a direct proof\footnote{The direct proof also shows that if the simultaneous extractor for $\inner{r}{x}$ is efficient, then so is the simultaneous extractor for $x$.} of \Cref{lem:quantumGL} which covers the general case. Let $\game' = (\setup, \tokengen, \gench', \ver')$ be a (stateful) cloning game, such that: \begin{itemize}
        \item $\rho_{k,x}$ is the token received by $\alice$, $x$ is the message, and $k$ is the challenge received by $\bob$ and $\charlie$.
        \item $\ver'(\sk, x, \ch, \ans)$ accepts if and only if $\ans = x$, i.e. $\game$ is a search game.
    \end{itemize} 

    \noindent Now define a new cloning game $\game = \gamesetup$, where: \begin{itemize}
        \item $\gench(\sk, x)$ computes $\ch \from \gench(\sk, x)$ and samples $r \from \cU_{\cR}$. It outputs $(r, \ch)$.
        \item $\ver(\sk, x, (r, \ch), \ans)$ accepts if and only if $\ans = \inner{r}{x}$.
    \end{itemize}
Suppose that the distribution of $x$, denoted by $\distr_\cM$, is non-trivial in the sense that $\trivprob(\game, \distr_\cM) \le 1/2 + \negl(\secparam)$ and $\game'$ is $\distr_\cM$-evasive, for otherwise the proof is trivial.
\par We will use the contrapositive of \Cref{thm:main_thm_it}, with and $\genchprime = \gench_{\independent}$. Now if $\abc$ exists, then $\game$ does not have information theoretic $(\distr_\cM,\negl)$ independent-challenge unclonable security, where $\cM = \bit^n$ is the message space and $\distr_\cM$ denotes the distribution of $x$. Now we observe that the conditions of \Cref{thm:main_thm_it}, i.e. conditions in \cref{item:1,item:2,item:3} and \cref{item:b} of \Cref{thm:main_lifting_reductions_to_cloning_games}, are all satisfied. Indeed, \cref{item:1,item:2,item:3} are satisfied by assumption. \Cref{item:b} follows from the well-known classical local Goldreich-Levin extraction\footnote{Note that even though the success probability of the extractor may depend on $x$, its description does not.} \cite{GL89}. Therefore, it follows that $\game$ does not have $(\distr_\cM, \negl)$ information theoretic independent-challenge unclonable security, meaning there exists $\abcprime$ which can simultaneously extract $x$ with non-negligible probability as desired.
    
\end{proof}

\paragraph{Search to Decision Transformation.} 

\begin{corollary}[Search to Decision]
\label{cor:search-to-dec:id-to-ind}
    Let $\game = (\setup, \tokengen, \gench, \ver)$ be a cloning search game with extension $\genchprime$ such that (1) $\gench(\sk, m)$ and $\genchprime(\sk, m)$ both do not depend on $m$, (2) $\game$ has statistical correctness, and (3) $\game$ has $(\distr, \varepsilon, \genchprime)$ unclonable security. Define a (stateful) cloning decision game $\game' = (\setup', \tokengen', \gench', \ver')$ as follows: \begin{itemize}
        \item There is no message, i.e. $\cM' = \bracC{\bot}$ is the message space.
        \item $\setup'(1^\secparam)$ computes $\sk \from \setup(1^\secparam)$. $\sk' = \sk$.
        \item $\tokengen'(\sk', m')$ parses the input as $\sk' = \sk$, samples $m \from \distr$, and computes $\rho \from \tokengen(\sk, m)$. It outputs a token $\rho$ and random coins $r_{\tokengen'} = m$
        \item $\gench'(\sk', m'; r_{\gench'})$ parses the input as $\sk' = \sk$. It interprets the random coins as $r_{\gench'} = (r,b, r_\gench)$, where $r \in \cM$, $b \in \bit$, and $r_\gench$ is randomness for $\gench$. Then it computes $\ch \from \gench(\sk, m; r_{\gench})$, and outputs $\ch' = (r, \inner{r}{m} \oplus b, \ch)$.
        \item $\ver'(\sk', m', \ch', \ans', r_{\gench'})$ parses the input as $r_{\gench'} = (r, b, r_\gench)$ and outputs $b' = [\ans' == b]$.
    \end{itemize}
    Then, $\game'$ has statistical correctness and $(\distr_\bot, \sqrt{\varepsilon}/2, \genchprime')$ independent-challenge unclonable security, where the extension $\genchprime'(\sk, m')$ is defined as follows: \begin{enumerate}
        \item Run $(r_B, r_C) \from \gench(\sk, m)$
        \item Sample $r_B', r_C' \uniform \cM$ and $b_B, b_C \uniform \bit$
        \item Output $\brac{(r_B', b_B, r_B), (r_C', b_C, r_C) }$.
    \end{enumerate}
    In particular, if $\game$ has $(\distr, \varepsilon)$ independent-challenge unclonable security, then $\game'$ has $(\distr_\bot, \sqrt{\varepsilon}/2)$ independent-challenge unclonable security.
\end{corollary}

\begin{proof}
    We construct the game $\game' = (\setup', \tokengen', \gench', \ver')$ as follows: 
    
    We define several hybrids: \begin{itemize}
        \item {\bf Hybrid 0:} This is the cloning experiment $\cloninggame^{\genchprime'}_{\game', \distr_\bot}$. Since $\game'$ is a decision game, we have $\trivprob(\game', \distr', \genchprime') \ge 1/2$. Assume for the sake of contradiction that the statement is false, then there exists a cloning adversary $\abcnum{0}$ that succeeds in this experiment with probability $p > \frac{1}{2} + \sqrt{\varepsilon}/2$. 
        \item {\bf Hybrid 1:} In this hybrid, we modify the success condition for the adversary. Instead of outputting the bits $b_B$ and $b_C$, $\bob$ and $\charlie$ are now required to output $\inner{r}{m}$ and $\inner{r'}{m}$, respectively, where $(r,r')$ are the random coins generated for $\ch_\bob'$ and $\ch_\charlie'$, respectively. This hybrid is clearly equivalent to Hybrid 0, since $\inner{r}{m} \oplus b_B$ and $\inner{r'}{m} \oplus b_C$ are known to $\bob$ and $\charlie$, respectively. Thus, there exists $\abcnum{1}$ which succeeds in this hybrid with probability $p$. Specifically, $\bob_1$ and $\charlie_1$ simply run $\bob_0$ and $\charlie_0$; then they XOR the output with the value mentioned above.
        \item {\bf Hybrid 2:} In this hybrid, we truncate the challenges given to $\bob'$ and $\charlie'$. Specifically, instead of getting $\ch_\bob' = (r, \inner{r}{m} \oplus b_B, \ch_\bob)$, $\bob$ will receive $(r, \ch_\bob)$; similarly $\charlie$ will receive $(r', \ch_\charlie)$. This hybrid is equivalent to Hybrid 1, so there exists an adversary $\abcnum{2}$ which succeeds in this hybrid with probability $p$. The reason is that $b_B, b_C$ are uniformly random bits that information theoretically hide the inner products $\inner{r}{m}$ and $\inner{r'}{m}$. In more detail, $\abcnum{2}$ can be constructed as follows:

        \begin{itemize}
            \item Upon receiving a token $\rho$, $\alice_2$ runs $\rho_{BC} \from \alice_1$ and samples random bits $b_B', b_C' \uniform \bit$. Then it sends the bipartite state $\rho_{BC} \otimes \ketbraX{b_B'}_B \otimes \ketbraX{b_C}_C$ to $\bob_2$ and $\charlie_2$.
            \item $\bob_2$, upon receiving $(r, \ch_\bob)$ from the challenger and a state $\rho_B \otimes \ketbraX{b_B'}_B$ from $\alice_2$, runs $\bob_1$ with input $(r, b_B', \ch_\bob)$. $\charlie_2$ is defined similarly.
        \end{itemize}
        
        Since the view of $\abcnum{1}$ exactly matches Hybrid 1, we conclude that $\abcnum{2}$ succeeds in Hybrid 2 with probability $p > \frac{1}{2} + \sqrt{\varepsilon}/2$.
    \end{itemize}
    
    By \Cref{lem:quantumGL}, this implies that there exists a cloning adversary $\abcprime$ which succeeds in $\cloninggame_{\game, \distr}^{\genchprime}$ with probability greater than $4 (\sqrt{\varepsilon}/2)^2 = \varepsilon$, a contradiction.
    
\end{proof}

By plugging in $\game = \game_{\mathsf{BB84}}$ from \Cref{def:bb84_cloning_game}, we get the following corollary:

\begin{corollary} \label{cor:sde}
    There exists a single-decryptor encryption scheme in the plain model with information-theoretic independent-challenge security.
\end{corollary}

\begin{proof}
    It could easily be checked that the resulting game $\game'$ is the cloning game corresponding to the following single-decryptor encryption scheme $(\gen, \tokengen, \enc, \dec)$: \begin{enumerate}
        \item $(\gen(1^\secparam))$: Sample $m, \theta \uniform \bit^\secparam$. Output $\sk = (m,\theta)$
        \item $\tokengen(\sk)$: Parse the input as $\sk = (m, \theta)$ and output the decryption token $\ketbraX{m^\theta}$.
        \item $\enc(\sk, b)$: Sample $r \uniform \bit^\secparam$. Output $ct = (r, \inner{r}{m} \oplus b, \theta)$
        \item $\dec(\rho, ct)$: Parse $ct = (r, b', \theta)$. Measure $\rho$ in basis $H^\theta$ to get $m$. Output $\inner{r}{m} \oplus b'$.
    \end{enumerate}
    
    Note that since $\game_{\mathsf{BB84}}$ is a cloning encryption game, it is secure against any extension, including the independent-challenge case. Therefore, the single-decryptor scheme above is optimally secure against independent challenges by \Cref{cor:search-to-dec:id-to-ind}.
    
\end{proof}

\subsubsection{Relationship Between Challenge Distributions}

\label{sec:cor_ch_dist}
As a second application, we show that when the trivial success probability of a cloning game is negligible, unclonable security for independent challenges implies unclonable security for any challenge distribution, up to a polynomial loss in the success probability.

\begin{corollary} \label{cor:ch_dist}
    Let $\game = \gamesetup$ be a $\distr$-evasive cloning game with $(\distr,\negl)$ independent-challenge unclonable security. Suppose that for any key-message pair $(\sk, m)$, the non-interactive assumption induced by $(\game, \distr, \sk, m)$ has verifiably polynomial image. Then, $\game$ has $(\distr, \negl, \genchprime)$ unclonable security for any extension $\genchprime$. In particular, $\game$ has $(\distr, \negl)$ identical-challenge unclonable security.
\end{corollary}

\begin{proof}
    Follows directly from \Cref{thm:main_lifting_reductions_to_cloning_games} by setting $\game = \game'$. The only non-trivial condition is \ref{item:b}, which follows by \Cref{cor:classical}. Note that the classical reduction is the identity reduction, hence it does not depend on $(\sk, m)$.
\end{proof}

\begin{remark}
    The corollary above can be applied to any unclonable primitive with verifiably polynomial image, where $\trivprob$ is negligible, including copy-protection for functions with output-length $\omega(\log \secparam)$.
\end{remark}

\paragraph{Direct Proof with Concrete Bounds.}

We can in fact show a stronger statement without using the main theorem, namely that the identity reduction works in \Cref{cor:ch_dist} with a cubic loss in success probability.

\begin{proof}[Alternate proof of \Cref{cor:ch_dist}]
Follows directly from \Cref{lem:ind-dep} after setting $\distr$ in the lemma statement to be the output of $\gench'$, setting $\distrtild$ to be the output of the extension $\genchprime$, and setting $B_r$ ($C_r$) to be a POVM element which tests whether $\bob$ ($\charlie$) passes the verification $\ver$ on challenge $r$.
\end{proof}

\section{Relating Unclonable-Search and Unclonable-Indistinguishability} \label{sec:search_to_indist}
\noindent In this section, we will give a relationship between games satisfying unclonable-indistinguishability security~\Cref{def:unclonable_dist_sec} and unclonable-search security~\Cref{def:unclonable_search_sec}. 
\par In~\Cref{sec:searchtoind}, we show that games with unclonable-search property imply games satisfying unclonable-indistinguishability property. In~\Cref{sec:indtosearch}, we show that games with unclonable-indistinguishability property imply games satisfying unclonable-search property. 
\par We start with some simple observations before moving onto our results. 

\paragraph{Trivial Success.} We show upper and lower bounds for trivial success probability of search games.

\begin{lemma}[Trivial Success Probability of Search Games] \label{lem:triv_prob}
Let $\game = \gamesetup$ be a cloning search game with correctness $\delta$ and let $\genchprime$ be an extension. Then, \begin{align*}
    (1 - \sqrt{1 - \delta})\brac{\trivguess(m \given \ch) - \sqrt{1 - \delta}} \le \trivprob(\game, \distr, \genchprime) \le \trivguess(m \given \ch).
\end{align*}
where the variables $(m, \ch)$ are sampled as in \Cref{def:cloning_game}.
\end{lemma}

\begin{proof}[Proof Sketch]
    Consider a $\bob$-trivial attack. For the upper-bound, simply bound the success probability of $\bob$. For the lower bound, note that for $(1 - \sqrt{1 - \sigma})$ fraction of the time, $\charlie$ will succeed with probability at least $1 - \sqrt{1 - \sigma}$ by correctness. Observe that $\bob$ can succeed with probability $\trivguess(m \given \ch_\bob, \aux)$ without the token. Use union bound to complete the proof. Similar for the $\charlie$-trivial attack.
\end{proof}

\begin{corollary} \label{cor:triv_prob_perfect_correctness}
    If $\game$ is a search game with perfect correctness, then $\trivprob(\game, \distr, \genchprime) = \trivguess(m \given \ch)$.
\end{corollary}

\paragraph{Message Hiding.} Clearly, a search game must have the property that the token hides the message, for otherwise cloning would be trivial. We formalize this in the lemma below.

\begin{lemma}[Message Hiding] \label{lem:message_hiding}
    Let $\game = \gamesetup$ be a cloning search game with $(\distr, \varepsilon)$ unclonable security, then for any QPT adversary $\alice$ we have \begin{align*}
        \pr{ m \from \alice(1^\secparam, \rho) \; : \; \substack{ \sk \from \setup(1^\secparam) \\ m \from \distr \\ \rho \from \tokengen(\sk, m) }  } \le \trivprob(\distr, \varepsilon) + \varepsilon(\secparam)
    \end{align*}
\end{lemma}

\begin{proof}
If the statement is false, then $\alice$ can send $m$ to $\bob$ and $\charlie$ in the splitting phase, both of which output $m$ as their answer, hence $\abc$ breaks the $(\distr, \varepsilon)$ unclonable security of $\game$.
\end{proof}

\subsection{Search-to-Indistinguishability}
\label{sec:searchtoind}
We transform games satisfying unclonable-search property to games satisfying unclonable-{\linebreak}indistinguishability property in the following steps. 
\begin{enumerate}
    \item In~\Cref{sec:message:dist}, we present a generic transformation that reduces the unclonable-search security of cloning games, where the message distribution comes from a high entropy distribution, to unclonable-search security, where the message distribution is uniform. 
    \item In~\Cref{sec:generic:aug}, we consider a new notion of security called augmented unclonable security. In this security notion, all the adversaries have oracle access to a point function $P_m(\cdot)$, where $m$ is such that the adversary receives a token generated using $m$. We show that we can generically transform any game satisfying unclonable-search security into one satisfying augmented unclonable-search security. 
    \item In~\Cref{sec:augtoind}, we show how augmented unclonable security implies unclonable-indistinguishability security. 
\end{enumerate}

\subsubsection{Relationship Between Message Distributions}
\label{sec:message:dist}

The following lemma is adapted from Theorem 9 in \cite{BL20}. 
\begin{lemma} \label{lem:min_entropy}
Let $\game = \gamesetup$ be a cloning game. Let $\distr$ be a distribution over the message space $\cM$ with min-entropy $h$, then for any cloning adversary $\abc$ and any extension $\genchprime$ of $\gench$, we have \begin{align*}
    \pr{1 \from \cloninggame^{\genchprime}_{\game, \distr}(1^\secparam, \abc)} \le 2^{\log_2|\cM| - h} \cdot \pr{ 1 \from \cloninggame^{\genchprime}_{\game,\cU_\cM}(1^\secparam, \abc) }.
\end{align*}
\end{lemma}

\begin{proof}
\begin{align*}
    \pr{1 \from \cloninggame^{\genchprime}_{\game, \distr}(1^\secparam, \abc)} &= \sum_{m \in \cM} \pr{m \from \distr} \cdot \pr{1 \from \cloninggame^{\genchprime}_{\game, \distr}(1^\secparam, \abc) \given m} \\
    &\le \sum_{m \in \cM} 2^{-h} \cdot \pr{1 \from \cloninggame^{\genchprime}_{\game, \distr}(1^\secparam, \abc) \given m} \\
    &= 2^{\log_2|\cM|-h} \sum_{m \in \cM} \frac{1}{|\cM|} \cdot \pr{1 \from \cloninggame^{\genchprime}_{\game, \cU_\cM}(1^\secparam, \abc) \given m} \\
    &= 2^{\log_2|\cM| - h} \cdot \pr{ 1 \from \cloninggame^{\genchprime}_{\game,\cU_\cM}(1^\secparam, \abc) }
\end{align*}

\end{proof} 

\subsubsection{Generically Augmenting Security}
\label{sec:generic:aug}

\begin{definition}[Augmented Unclonable Security] \label{def:aug_sec}
Let $\game$ be a cloning search game and $\distr$ be a distribution over the message space $\cM$. Let $\augcloninggame^{\genchprime}_{\game, \distr}$ be the following \emph{augmented cloning experiment}, with the modification highlighted in \highlight{blue}: \begin{itemize}
    \item {\bf \em Setup Phase: }
    \begin{itemize}
        \item All parties get a security parameter $1^\secparam$ as input.
        \item $\referee$ samples a message $m \from \distr$.
        \item $\referee$ computes $\sk \from \setup(1^\secparam)$ and $\rho \from \tokengen(\sk, m)$.
        \item $\cR$ sends $\rho$ to $\alice$.
        \item \highlight{$\abc$ all get oracle access to the point function $P_m(\cdot)$.}
    \end{itemize}
    \item {\bf \em Splitting Phase: } \begin{itemize}
        \item $\alice$ computes a bipartite state $\rho'$ over registers $B,C$.
        \item $\alice$ sends $\rho'[B]$ to $\bob$ and $\rho'[C]$ to $\charlie$.
    \end{itemize}
    \item {\bf \em Challenge Phase: } \begin{itemize}
        \item $\referee$ samples $(r_\bob, r_\charlie) \from \genchprime(\sk, m)$ and computes $\ch_\bob = \genchprime(\sk, m; r_\bob), \ch_\charlie = \genchprime(\sk, m; r_\charlie)$
        \item $\referee$ sends $\ch_\bob$ to $\bob$ and $\ch_\charlie$ to $\charlie$.
        \item $\referee$ sends $\ch_\bob$ to $\bob$ and $\ch_\charlie$ to $\charlie$.
        \item $\bob$ and $\charlie$ send back answers $\ans_\bob$ and $\ans_\charlie$, respectively.
        \item $\referee$ computes bits $b_\bob \from \ver(\sk, m, \ch_\bob, \ans_\bob)$ and $b_\charlie \from \ver(\sk, m, \ch_\charlie, \ans_\charlie)$.
        \item The outcome of the game is denoted by $\augcloninggame^{\genchprime}_{\game, \distr}(1^\secparam,\abc)$, which equals 1 if $b_\bob = b_\charlie = 1$, indicating that the adversary has won, and 0 otherwise, indicating that the adversary has lost.
    \end{itemize}
\end{itemize}

Note that $\augcloninggame^{\genchprime}_{\game, \distr}$ defined the same as $\cloninggame^{\genchprime}_{\game, \distr}$, but $\abc$ additionally get oracle access to $P_m(\cdot)$. We say that $\game$ has \emph{$(\distr, \varepsilon, \genchprime)$ augmented unclonable security} if for all QPT cloning adversaries $\abc$ we have \begin{align*}
    \pr{ 1 \from \augcloninggame^{\genchprime}_{\game, \distr}(1^\secparam, \abc) } \le \trivprob(\game, \distr, \genchprime) + \varepsilon(\secparam).
\end{align*}
\end{definition}

\begin{lemma} \label{lem:aug_sec}
Let $\game = (\setup, \tokengen, \gench, \ver)$ be a cloning search game with message space $\cM$ such that $\trivprob(\game, \cU_\cM)$ is negligible. If $\game$ has $(\cU_\cM, |\cM|^{-\delta})$ unclonable search security for some $\delta > 0$, then $\game$ has $(\cU_\cM, \negl)$ augmented unclonable security.
 
\end{lemma}

\begin{proof}

We will define a sequence of hybrids: 
\paragraph{Hybrid 1:} This is the original augmented cloning experiment $\augcloninggame^{\genchprime}_{\game, \distr}$. Suppose for the sake of contradiction that there exists a QPT adversary $\abc$ which succeeds with non-negligible probability $p$.

\paragraph{Hybrid 2:} In this hybrid, we replace the oracles $P_m(\cdot)$ with $P_S(\cdot)$, where $S \subset \cM$ is a random subset containing $m$ of size $|\cM|^{1 - \delta/2}$. The fact that $\trivprob(\game, \cU_\cM)$ is negligible implies that $\frac{1}{|\cM|}$ is negligible, hence $\frac{|S|}{|\cM|} = \frac{1}{|\cM|^{\delta/2}}$ is also negligible. We claim that $\abc$ succeeds in this hybrid with probability at least $p - \negl(\secparam)$. \\

\noindent Suppose not, we will construct an adversary $\alice'$ which will violate \Cref{cor:subset_hiding}: \begin{itemize}
    \item $\alice'$ picks $m \from \cU_\cM$ and receives oracle access to $\cO$, where either $\cO = P_m$ or $\cO = P_S$ for a random subset $S \subset \cM$ of size $|\cM|^{1 - \delta/2}$ containing $m$.
    \item $\alice'$ then simulates $\augcloninggame^{\genchprime}_{\game, \distr}$ for $\abc$, using $m$ as the message and using $\cO$ as the oracle given to $\abc$. Since $\abc$ are QPT algorithms, $\alice'$ only makes polynomially many queries. $\alice'$ outputs the bit $\augcloninggame^{\genchprime}_{\game, \distr}(1^\secparam, \abc)$.
\end{itemize}
If $\cO = P_m$, then the view of $\abc$ is exactly Hybrid 1, and otherwise it is exactly Hybrid 2. Thus, $\alice'$ has a non-negligible advantage.

\paragraph{Hybrid 3:} In this hybrid, we change the order of sampling. A random subset $S \subset \cM$ of size $|\cM|^{1 - \delta/2}$ is sampled at the beginning of the experiment, and the message later is sampled as $m \from \cU_S$. This is perfectly indistinguishable from Hybrid 2, as the view of $\abc$ is unchanged.

\paragraph{Hybrid 4:} In this hybrid, we fix a particular $S$ such that $\abcprime$ succeeds with probability $p - \negl(\secparam)$ in Hybrid 3, by an averaging argument. \\

\par Now we observe that Hybrid 4 is exactly the experiment $\cloninggame^{\genchprime}_{\game, \cU_S}$. Note that the distribution $\cU_S$ has min-entropy $(1 - \delta/2)\log_2 |\cM|$. Thus, by \Cref{lem:min_entropy}, we have \begin{align*}
    p - \negl(\secparam) \le \pr{ 1 \from \cloninggame_{\game, \cU_S}(1^\secparam, \abc)} \le |\cM|^{\delta/2} \pr{ 1 \from \cloninggame_{\game, \cU_\cM}(1^\secparam, \abc)} \le |\cM|^{-\delta/2} ,
\end{align*}

which contradicts the assumption that $p$ is non-negligible.
\end{proof}

\subsubsection{From Augmented Security to Unclonable-Indistinguishability}
\label{sec:augtoind}
Now, we state our result that shows how to go from a search game in the plain model to a search game in QROM with unclonable indistinguishable security. We first invoke~\Cref{lem:aug_sec} to generically obtain augmented unclonable-search security and we then leverage this notion of security to obtain unclonable-indistinguishable security. 

\begin{theorem}[Search to Indistinguishability] \label{thm:search-to-dec}
Let $\game = (\setup, \tokengen, \gench, \ver)$ be a statistically correct cloning search game with message space $\cM = \bit^\ell$, where $\ell = \poly(\secparam)$, such that $\trivprob(\game, \cU_\cM)$ is negligible and $\gench(\sk, m)$ does not depend\footnote{This requirement can be lifted by extending the definition of stateful cloning games and having $\genchprime$ know the random coins of $\tokengen'$ (in this case $m$). We keep the syntax simple for there is no known application to the more general case.} on $m$. Suppose that $\game$ has $(\cU_\cM, |\cM|^{-\delta})$ unclonable search security for some $\delta > 0$. Let $n = \poly(\secparam)$, and define a cloning search game $\game' = \gamesetupprime$ in QROM as follows:
\begin{itemize}
        \item $\game'$ has message space $\cM' = \bit^n$.
        \item Let $H : \cM \to \cM'$ be a random oracle.
        \item $\setup'(1^\secparam)$ runs $\sk \from \setup(1^\secparam)$. It outputs $\sk' = \sk$.
        \item $\tokengen'(\sk', m')$ parses the input as $\sk' = (\sk, m)$. It samples $m \from \cU_\cM$ and computes $\rho \from \tokengen(sk, m)$, then it outputs the token $\rho' = \rho \otimes \ketbraX{m' \oplus H(m)}$ and random coins $r_{\tokengen'} = m$
        \item $\gench'(\sk', m')$ parses the input as $\sk' = \sk$. It computes $\ch \from \gench(\sk, m)$ (recall that by assumption this does not require knowledge of $m$) and outputs $\ch$.
        \item $\game'$ is a search game, which defines $\ver'$.
    \end{itemize}
Then, $\game'$ has statistical correctness and $\negl$ unclonable indistinguishable security.

\end{theorem}

\begin{proof}
    
    Statistical correctness of $\game'$ follows easily from statistical correctness of $\game$, so it suffices to show $\negl$ unclonable indistinguishable security. Keep in mind that by \Cref{lem:aug_sec}, $\game$ has $\negl$ augmented unclonable security. Suppose there exists a QPT cloning adversary $\abc$ which breaks the unclonable distinguishing security of $\game'$. Let $(m_0', m_1')$ be the messages used by $\abc$. Note that $\trivprob(\game', \distr_{m_0', m_1'}) \ge 1/2 - \negl(\secparam)$ by the statistical correctness of $\game'$ and \Cref{lem:triv_prob}, so that we have \begin{align*} \pr{ 1 \from \cloninggame_{\game', \distr_{m_0', m_1'}}(1^\secparam, \abc) } \ge \frac{1}{2} + \mu(\secparam)
    \end{align*}
    for some non-negligible function $\mu$.
    
    We first define a sequence of hybrids: \begin{itemize}
        \item {\bf Hybrid 0:} The original cloning experiment $\cloninggame_{\game', \distr_{m_0', m_1'}}$. $\abcprime$ succeeds with probability $1/2 + \mu(\secparam)$ in this experiment.
        \item {\bf Hybrid 1:} In this hybrid, we will replace the oracle $H$ for $\alice'$ only with the reprogrammed oracle $H_1$, where \begin{align*}
            H_1(x) := \begin{cases} u, \quad x = m \\ H(x), \quad x \ne m   \end{cases},
        \end{align*}
        where $u \in \bit^n$ is an independent random string. We claim that $\abc$ succeeds in this hybrid with probability at least $\frac{1}{2} + \mu(\secparam) - \negl(\secparam)$.
        
        \par Suppose not, we will construct an adversary $\alice'$ that breaks the message hiding property of $\game$ (\Cref{lem:message_hiding}): \begin{itemize}
            \item $\alice'$ simulates Hybrid 1 for $\alice$, using the token $\rho$ it receives and a fresh random oracle $H'$ it simulates on-the-fly.
            \item Then $\alice'$ measures a random query made by $\alice$ to $H'$.
        \end{itemize}
        Note that if $H'$ was replaced by $H$, the view of $\alice$ would be exactly Hybrid 0. Thus, by \Cref{thm:bbbv}, $\alice'$ outputs the message $m$ with non-negligible probability, which is larger than $\trivprob(\game, \cU_\cM) + |\cM|^{-\delta}$, a contradiction.
        
        \item {\bf Hybrid 2:} In this hybrid, we change the order of sampling with regard to $H(m)$ in the experiment without changing the view of $\abcprime$, so that the probability of success remains the same. More specifically, we sample $r \uniform \bit^n$ and send it to $\alice$ instead of $m' \oplus H(m)$. Then, we sample $b \uniform \bit$ and give both $\bob$ and $\charlie$ the reprogrammed oracle $H_2^b$, where \begin{align*}
            H_2^b(x) := \begin{cases} m_b' \oplus r, \quad x = m \\ H(x), \quad x \ne m   \end{cases}
        \end{align*}
    \end{itemize}
    
    For fixed $\sk, m, \ch_\bob, \ch_\charlie$ and fixed $H_{-m}$, which is the partial random oracle defined on inputs $x \ne m$, consider the following projectors:
    \begin{itemize}
        \item $\Pi_B^b:$ Run $\bob'$ on challenge $\ch_\bob$ with oracle $H_2^b$. Check if the output of $\bob'$ is $m_b'$. Undo the computation.
        \item Similarly define $\Pi_C^b$ for $b \in \bit$.
    \end{itemize}
    
    Without loss of generality assume that the bipartite state output by $\alice$ is a pure state $\ket{\phi_{BC}}$ in Hybrid 2. Let $\Pi_B = (\Pi_B^0 + \Pi_B^1)/2$ and $\Pi_C = (\Pi_C^0 + \Pi_C^1)/2$. We spectrally decompose the state as \begin{align*}
        \ket {\phi_{BC}} = \sum_{i, j} \alpha_{i, j} \ket{\varphi_i}_B \otimes \ket{\psi_j}_C.
    \end{align*}
    
    where $\Pi_B \ket{\varphi_i} = \lambda_i \ket{\varphi_i}$ and $\Pi_C \ket{\psi_j} = \gamma_j \ket{\psi_j}$. \\
    
    \begin{claim} For any polynomial $p(\cdot)$, with overwhelming probability over $\sk, m, \ch_\bob, \ch_\charlie, H_{-m}$, we have \begin{align}
        \sum_{ \substack{i \; : \; |\lambda_i  - 1/2| \ge 1/p(\secparam) \\ j \; : \; |\gamma_j - 1/2| \ge 1/p(\secparam) } } \abs{\alpha_{i,j}}^2  \le \negl(\secparam) \label{eq:bad_weights}
    \end{align}
    
    \end{claim}
    
    \begin{proof}
        Assume the quantity in \cref{eq:bad_weights} is a non-negligible function $w(\secparam)$. We will construct an adversary $\abcprime$ which breaks the $\negl$ augmented security of $\game$: \begin{itemize}
            \item $\abcprime$ get oracle access to $P_m(\cdot)$, where $m \from \distr$.
            \item $\alice'$ gets a quantum token $\rho$. It samples a $2t$-wise independent hash function $H: \bit^\ell \to \bit^n$ to simulate a random oracle and a random string $r \uniform \bit^n$, where $t = \poly(\secparam)$ is an upper-bound on the number of total random oracle queries made by $\abc$. Then, it runs $\alice$ on $\rho \otimes \ketbraX{r}$ to get a bipartite state $\ket{\phi_{BC}}$. It sends this state to $\bob'$ and $\charlie'$. It also sends the description of $H$ to both of them.
            \item $\bob'$ receives a challenge $\ch_\bob$ from the challenger. It implements the operator $\Pi_B$ defined above. Note that $\bob'$ can reprogram the random oracle $H_2$ on input $m$ using and its oracle access to $P_m(\cdot)$. Then, $\bob'$ applies the efficient symmetric approximate threshold measurement $\sati_{(P,Q), \gamma}^{\epsilon, \delta}$ in \Cref{thm:thres_approx} with $P = (\Pi^B_0 + \Pi^B_1)/2$, $Q = I - P$, $\gamma = 3/4p$, $\epsilon = 1/4p$ and $\delta = 2^{-\lambda}$. If the outcome is 0, $\bob'$ aborts. If the outcome is 1, $\bob'$ then runs $\bob$ on the leftover state with $H_2^0$ or $H_2^1$ picked uniformly at random. It measures and outputs a random query $\bob$ makes to the random oracle.
            \item $\charlie'$ is defined the same way as $\bob'$, in the end outputting a random query made by $\charlie$ to the random oracle.
        \end{itemize}
        
        By \Cref{thm:thres_approx} bullet (1), both $\bob'$ and $\charlie'$ will get outcome 1 with non-negligible probability $w - 2\delta$, in which case by bullet (2) the leftover state will be $4\delta$-close to the the following state:
        \begin{align*}
            \sum_{\substack{i: |\lambda_i - 1/2| > 1/4p \\ j: |\gamma_j - 1/2| > 1/4p}} \alpha_{i, j} \ket{\varphi_i}_B \otimes \ket{\psi_j}_C. 
        \end{align*}
        
        Observe that when $\bob$ does not query $m$, it will succeed with probability exactly $1/2$. Therefore, by \Cref{thm:bbbv}, the query weight of $\bob$ on $m$ is non-negligible. Similarly, the query weight of $\charlie$ on $m$ is non-negligible. Therefore, $\abcprime$ succeed with non-negligible probability, a contradiction.
        
    \end{proof}
    
Therefore, for any polynomial $p$, we have that $\ket{\phi_{BC}}$ is negligibly close to the state $\ket{\phi_B} + \ket{\phi_C}$, where \begin{align*}
    \ket{\phi_{B}} = \sum_{\substack{i: |\lambda_i - 1/2| \leq 1/p}} \alpha_{i, j} \ket{\varphi_i}_B \otimes \ket{\psi_j}_C, \quad \ket{\phi_C} = \sum_{\substack{i: |\lambda_i - 1/2| > 1/p \\ j: |\gamma_j - 1/2| \leq 1/p}} \alpha_{i, j} \ket{\varphi_i}_B \otimes \ket{\psi_j}_C,
\end{align*}
    in which case we could bound the success probability of $\abc$ in Hybrid 2 as: \begin{align*}
        & \frac{1}{2} + \mu(\secparam) - \negl(\secparam)    \le (\left| (\Pi^0_B \otimes \Pi^0_C) (\ket{\phi_\Bs} + \ket {\phi_C})\right|^2 + \left| (\Pi^1_B \otimes \Pi^1_C) (\ket{\phi_\Bs} + \ket {\phi_C)}\right|^2)/2 \\
        =& \frac{1}{2} \cdot \left( \langle \phi_{\Bs} | (\Pi^0_B \otimes \Pi^0_C) | \phi_{\Bs}  \rangle +  \langle \phi_{\Bs} | (\Pi^1_B \otimes \Pi^1_C) | \phi_{\Bs}  \rangle  + \langle \phi_{\Cs} | (\Pi^0_B \otimes \Pi^0_C) | \phi_{\Cs}  \rangle  +  \langle \phi_{\Cs} | (\Pi^1_B \otimes \Pi^1_C) | \phi_{\Cs}  \rangle \right) \\
        +& \mathsf{Re}\left(  \langle \phi_{\Bs} | (\Pi^0_B \otimes \Pi^0_C) | \phi_{\Cs} \rangle + \langle \phi_{\Bs} | (\Pi^1_B \otimes \Pi^1_C) | \phi_{\Cs}  \rangle \right) \\
        \leq & \frac{1}{2} \cdot \left( \langle \phi_{\Bs} | (\Pi^0_B \otimes I) | \phi_{\Bs}  \rangle +  \langle \phi_{\Bs} | (\Pi^1_B \otimes I) | \phi_{\Bs}  \rangle  + \langle \phi_{\Cs} | (I \otimes \Pi^0_C) | \phi_{\Cs}  \rangle  +  \langle \phi_{\Cs} | (I \otimes \Pi^1_C) | \phi_{\Cs}  \rangle \right) \\
        +& \mathsf{Re}\left(  \langle \phi_{\Bs} | (\Pi^0_B \otimes \Pi^0_C) | \phi_{\Cs} \rangle + \langle \phi_{\Bs} | (\Pi^1_B \otimes \Pi^1_C) | \phi_{\Cs}  \rangle \right) \\
        = & \langle \phi_{\Bs} | (\Pi_B \otimes I) | \phi_{\Bs}  \rangle + \langle \phi_{\Cs} | (I \otimes \Pi_C) | \phi_{\Cs}  \rangle + \mathsf{Re}\left(  \langle \phi_{\Bs} | (\Pi^0_B \otimes \Pi^0_C) | \phi_{\Cs} \rangle + \langle \phi_{\Bs} | (\Pi^1_B \otimes \Pi^1_C) | \phi_{\Cs}  \rangle \right) \\
        &\le \frac{1}{2} + \frac{1}{p},
    \end{align*}
    where in the last step we used \Cref{lem:jordans}. Since $p$ is arbitrary, this is a contradiction.
    
\end{proof}

As a corollary, we achieve unclonable encryption in QROM using BB84 states. This is an improvement over the main result of \cite{AKLLZ22}, as it can be more easily implemented on near-term quantum computers.

\begin{corollary}[Existence of Unclonable Encryption in QROM Using Prepare-and-Measure States] \label{cor:ue_qrom}
There exists a public-key unclonable encryption scheme with unclonable indistinguishable security in QROM, which uses only prepare-and-measure quantum operations.
\end{corollary}

\begin{proof}
Set $\game = \game_{\mathsf{BB84}}$ from \Cref{def:bb84_cloning_game}, then it satisfies the condition of \Cref{thm:search-to-dec}, so there exists $\game'$ as described in the theorem. Observe that since $\game$ is a cloning encryption game, so is $\game'$. Thus, $\game'$ gives a construction of unclonable encryption in QROM, and the $\negl$ unclonable distinguishing security implies unclonable security of this encryption scheme.

\noindent Finally, we invoke the generic compiler in \cite{AK21} to achieve a public-key scheme\footnote{Although \cite{AK21} only showed that their transformation preserves weak unclonable security, it could easily be checked that it also preserves strong unclonable security.}.
\end{proof}

\subsection{Indistinguishablity-to-Search}
\label{sec:indtosearch}

\begin{theorem} \label{thm:ind_to_search}
    If $\game = \gamesetup$ is a cloning search game with statistical correctness and $\varepsilon$ unclonable indistinguishable security, then $\game$ has $(\distr, 2\varepsilon + \negl)$ unclonable security for any distribution $\distr$ over the message space $\cM$ such that $\game$ is $\distr$-evasive.
\end{theorem}

\begin{proof}
We will give a generic proof that works for an arbitrary extension $\genchprime$, which we will omit for simplicity. Suppose there exists a cloning adversary $\abc$ which breaks $(\distr, \varepsilon)$ unclonable security of $\game$. We will construct an adversary $\abcprime$, which will break $(\distr_{m_0, m}, \varepsilon')$ unclonable security of $\game$, where $\varepsilon' = \varepsilon/2 - \negl$, $m_0 \in \cM$ is a fixed message, and $m$ is a message sampled as $m \from \distr$. By convexity, this would imply that there exists a fixed value of $m$ for which the security is broken, hence finishing the proof. We describe $\abcprime$ below: \begin{itemize}
    \item $\alice'$ is the same as $\alice$.
    \item $\bob'$ runs $\bob$ with the input it receives. If the output is $m$, $\bob'$ outputs $m$; otherwise $\bob'$ outputs $m_0$.
    \item $\charlie'$ is defined similarly to $\bob'$.
\end{itemize}

Note that if the message in the experiment $\cloninggame_{\game, \distr_{m_0, m}}$ above is $m$, then by assumption $\abcprime$ wins with probability $p > \varepsilon$. On the other hand, if the message is $m_0$, then the probability that $\bob'$ outputs $m$ is at most $\trivguess(m \given \ch_\bob) \le \negl(\secparam)$ by \Cref{lem:triv_prob}. Hence, by union bound, $\abcprime$ succeeds with overwhelming probability in this case, and we have \begin{align*}
    \E_{m \from \distr} \pr{ 1 \from \cloninggame_{\game, \distr_{m_0, m}}(1^\secparam, \abcprime) } = \frac{1}{2} \brac{ p + 1 - \negl } \ge \frac{1}{2} + \frac{\varepsilon(\secparam)}{2} - \negl(\secparam)
\end{align*}

\end{proof}
as desired.

\newcommand{\hin}{\cM}
\newcommand{\hout}{\bit^n}
\newcommand{\gin}{\cM'}
\newcommand{\gout}{\cK}

\section{From Search to Decision Games} \label{sec:search_to_dec}
We present a transformation from search games to decision games. We remark that the transformation is tailored to application of copy-protection and thus, the resulting decision game has a specific form. 

\begin{theorem}[Search to Decision] \label{thm:search_to_dec_cp}

Let $\game = (\setup, \tokengen, \gench, \ver)$ be a statistically correct cloning encryption game with message space $\cM = \bit^\ell$, where $\ell = \poly(\secparam)$, such that \begin{itemize}
    \item $\setup(1^\secparam)$ outputs a uniformly random key $\sk \uniform \gout$, where $|\gout|^{-1}$ is negligible in $\secparam$.
    
    \item $\game$ is $\cU_\cM$-evasive.

    \item $\game$ has $(\cU_\cM, |\cM|^{-\delta})$ unclonable search security for some $\delta > 0$
\end{itemize} Define a cloning decision game $\game' = \gamesetupprime$ in QROM as follows: \begin{itemize}
    \item Let $\gin$ be an arbitrary message space.
    \item Let $H: \hin \to \hout$ and $G: \gin \to \gout$ be random oracles, where $n = \Omega(\secparam)$.
    \item $\setup'(1^\secparam)$ outputs $\sk' = \bot$
    \item $\tokengen'(\sk', m')$ samples $m \from \cU_\cM$. It computes $\rho \from \tokengen(G(m'), m)$. It outputs $\rho' = \rho \otimes \ketbraX{H(m)}$.
    \item $\gench'(\sk', m')$ is an algorithm which does not depend on $\sk'$, and outputs a value $\ch' \in \gin$.
    \item $\ver'(\sk', m', \ch', \ans')$ accepts if $\ans' = [\ch' == m']$.
\end{itemize}

    Then, $\game'$ has statistical correctness and $(\distr', \negl)$ independent-challenge unclonable security for any unlearnable distribution $\distr'$ over $\hin$.

\end{theorem}

\begin{proof}
We begin with correctness. Following \Cref{def:correctness}, we define a QPT algorithm $\alice_{\game'}(\rho', \ch')$ as follows: \begin{itemize}
    \item Parse the token as $\rho' = \rho \otimes \ketbraX{y}$, i.e. measure the last register to obtain $y \in \hout$.
    \item Compute $\ch \from \gench(G(\ch'), m)$, recall that this does not require the knowledge of $m$.
    \item Compute $m \from \alice_\game(\rho, \ch)$.
    \item Output $b = [H(m) == y]$.
\end{itemize} 

To analyze correctness, we consider two cases: (1) $\ch' = m'$ and (2) $\ch' \ne m'$. \begin{itemize}
    \item If $\ch' = m'$, then by statistical correctness of $\game$, $\alice_{\game'}$ receives the correct message $m = \sk'$ from $\alice_\game$ above and outputs 1 with overwhelming probability.
    \item If $\ch' \ne m'$, then $\alice_\game$ sees a challenge generated by a random key $\Tilde{\sk} = G(\ch')$ that is independent from $G(m')$, which is the key used to generate the token $\alice_G$ receives. We will show that the probability that $\alice_\game$ outputs the correct message $m$ in this case is negligible. This will then imply that $\alice_\game'$ outputs 0 with overwhelming probability, since the output length of $H$ is $\Omega(\secparam)$. \\
    \par Now suppose that $\alice_\game$ outputs the correct message above with non-negligible probability $p$, we will construct an adversary $\alice'$ that breaks the message hiding property of $\game$ using $\alice_\game$: \begin{itemize}
        \item $\alice'$ receives a token $\rho \from \tokengen(\sk, m)$, where $\sk \uniform \gout$ and $m \from \distr$.
        \item $\alice'$ then samples $\Tilde{\sk} \uniform \gout$ and computes $\ch \from \gench(\sk, m)$. Recall that this does not require the knowledge of $m$.
        \item Next, $\alice'$ computes and outputs $\Tilde{m} \from \alice_\game(\rho, \ch)$.
    \end{itemize}
    By assumption, $\Tilde{m} = m$ with non-negligible probability, since for $\ch' \ne m'$ the value $G(\ch')$ is identically distributed as the value $\Tilde{sk}$ sampled by $\alice'$ above. Therefore, $\alice'$ breaks the message-hiding property of $\game$ given in \Cref{lem:message_hiding}, a contradiction.
\end{itemize}
\ \\
\par Next, we show $(\distr, \negl)$ unclonable security. Note that by \Cref{lem:aug_sec}, $\game$ has $\negl$ augmented unclonable security. We will define a sequence of hybrids:

\begin{itemize}
    \item {\bf Hybrid 0:} This is the original cloning experiment $\cloninggame_{\game', \distr'}^{\independent}$. Suppose for the sake of contradiction that there exists a cloning adversary $\abc$ which succeeds in this hybrid with probability $p = \trivprob(\game', \distr') + \mu(\secparam)$ for a non-negligible function $\mu$.
    
    \item {\bf Hybrid 1:} In this hybrid, we will change the oracle $G$ for $\alice$ only to the punctured oracle $G_m$, defined as \begin{align*}
        G_{m'}(x) := \begin{cases} u, \quad x = m' \\ G(x), \quad x \ne m'  \end{cases},
    \end{align*}
    where $u \in \gout$ is an independent uniformly random value. We claim that $\abc$ succeed in this experiment with probability $p - \negl(\secparam)$. Suppose this is false. Let $\rho_{BC}^j$ be the bipartite state outputs by $\alice$ in Hybrid $j$, then $\tracedist{\rho_{BC}^1}{\rho_{BC}^2}$ must be non-negligible, since the only difference between Hybrids 0-1 is on the random oracle $G$ for $\alice$. Using $\alice$, we will construct an adversary that breaks the unlearnability of $\distr'$: \begin{itemize}
        \item $\alice'$ gets oracle access to $P_{m'}(\cdot)$ (which will not be used), where $m' \from \distr'$. It samples random oracles $G' : \gin \to \gout$ and $H' : \hin \to \hout$, as well as a message $m \from \distr$.
        \item Then, $\alice'$ samples $\sk \uniform \gout$, computes $\rho \from \tokengen(\sk, m)$, and runs $\alice(\rho, H'(m))$.
        \item During the last step, it measures and outputs a random query made by $\alice$ to the oracle $G'$.
    \end{itemize}
    As the view of $\alice$ in Hybrid 1 is perfectly simulated, by \Cref{thm:bbbv}, the query weight of $\alice$ on $m'$ is non-negligible, hence $\alice'$ correctly outputs $m'$ with non-negligible probability, breaking unlearnability.
    
    \item {\bf Hybrid 2:} In this hybrid, we change the order of sampling. A random value $k \from \cU_\gout$ is sampled at the beginning of the experiment and the quantum part of the token received by $\alice$ is calculated as $\rho \from \tokengen(k, m)$. In turn, the random oracle $G$ for $\bob$ and $\charlie$ is replaced by the reprogrammed oracle $G^2_{m'}$ which is defined as \begin{align*}
        G^2_{m'}(x) := \begin{cases} k, \quad x = m' \\ G(x), \quad x \ne m'  \end{cases},
    \end{align*}
    
    This hybrid is perfectly indistinguishable from Hybrid 1, hence $\abc$ succeeds in this hybrid with probability at least $p - \negl(\secparam)$.
    
    \item {\bf Hybrid 3:} In this hybrid, we replace the oracle $H$ for $\alice$ with the punctured oracle $H_m$ defined as \begin{align*}
        H_m(x) := \begin{cases} w, \quad x = m \\ G(x), \quad x \ne m  \end{cases},
    \end{align*}
    where $w \in \hout$ is an independent uniformly random value. We claim that $\abc$ succeeds in this hybrid with probability at least $p - \negl(\secparam)$. Suppose not, then we have $\tracedist{\rho_{BC}^2}{\rho_{BC}^3}$ is non-negligible. We will use $\abc$ to construct an adversary $\alice'$ to break the message hiding property of $\game$: \begin{itemize}
        \item $\alice'$ receives $\rho \from \tokengen(\sk, m)$ from the challenger, where $\sk \uniform \gout$ and $m \from \distr$.
        \item Then $\alice'$ samples $w \uniform \hout$ and simulates random oracles $H' : \hin \to \hout, G' : \gin \to \gout$ and runs $\alice$ on input $\rho \otimes \ketbraX{w}$ with oracles $H', G'$. It measures and outputs a random query made by $\alice$ to the oracle $H'$.
    \end{itemize}
    
    As the view of $\alice$ in Hybrid 3 is perfectly simulated, by \Cref{thm:bbbv}, the query weight of $\alice$ on $m$ is non-negligible, hence $\alice'$ outputs $m$ with non-negligible probability, contradicting the message-hiding property of $\game$ given in \Cref{lem:message_hiding}.

\end{itemize}

    Therefore, we have established that $\abc$ succeeds in Hybrid 3 with probability $\trivprob(\game', \distr') + \eta(\secparam)$ for some non-negligible function $\eta$. Without loss of generality, assume that the bipartite state created by $\alice$ is a pure state, i.e. $\rho_{BC}^3 = \ketbraX{\phi_{BC}}$. Consider the following binary POVM elements: \begin{itemize}
        \item $\Pi_B$: Sample $\ch' \from \gench'(sk', m')$. Run $\bob$ with oracles $H, G^2_{m'}$ and input $\ch'$. Measure if the output of $\bob$ equals $[\ch' == m']$, in which case we will say that $\bob$ \emph{passed} $\Pi_B$.
        \item Similarly define $\Pi_C$ for $\charlie$.
    \end{itemize}
    
    \noindent We write the state received by $\bob$ and $\charlie$ in spectral decomposition as \begin{align*}
        \ketbraX{\phi_{BC}} = \sum_{i,j} \alpha_{i,j} \ket{\varphi_i}_B\ket{\psi_j}_C,
    \end{align*}
    where $\Pi_B \ket{\varphi_i}_B = \lambda_i \ket{\varphi_i}_B$ and $\Pi_C \ket{\psi_j}_C = \gamma_j \ket{\psi_j}_C$. Let $q$ be a polynomial. Let $\trivprob := \trivprob(\game', \distr')$.

    We can bound the probability that $\abc$ succeeds in Hybrid 3 as \begin{align*}
        & p - \negl(\secparam) \le \braket{ \phi_{BC} | (\Pi_B \otimes \Pi_C) | \phi_{BC}} = \sum_{i,j} \abs{\alpha_{i,j}}^2 \lambda_i \gamma_j \\
        &\le \sum_{\substack{i: \; \lambda_i > \trivprob + 1/q \\ j: \; \gamma_j > \trivprob + 1/q}} \abs{\alpha_{i,j}}^2 \lambda_i \gamma_j \\
        &+ \brac{\trivprob + \frac{1}{q}}\sum_{\substack{i: \; \lambda_i \le \trivprob + 1/q \\ j: \; \gamma_j > \trivprob + 1/q}} \abs{\alpha_{i,j}}^2 + \brac{\trivprob + \frac{1}{q}} \sum_{i, j: \; \mu_j \le \trivprob + 1/q} \abs{\alpha_{i,j}}^2 \\
    &\le \trivprob + \frac{1}{q} + \sum_{\substack{i: \; \lambda_i > \trivprob + 1/q \\ j: \; \gamma_j > \trivprob + 1/q}} \abs{\alpha_{i,j}}^2 \lambda_i \gamma_j.
    \end{align*}
    Thus, it suffices to show the following claim to reach a contradiction:
    
    \begin{claim} \label{clm:cp_weights}
    For any polynomial $q$, with overwhelming probability over $\sk' = m, m', H,G,u,k,w,$ we have
    \begin{align*}
        \sum_{\substack{i: \; \lambda_i > \trivprob + 1/q \\ j: \; \gamma_j > \trivprob + 1/q}} \abs{\alpha_{i,j}}^2 \lambda_i \gamma_j \le \negl(\secparam).
    \end{align*}
    \end{claim}
    
    \begin{proof}
    Suppose that there exists $\abc$ which violates \Cref{clm:cp_weights}. Consider the following cloning adversary $\abcnum{1}$ against the $(\distr, \negl)$ augmented security of $\game$: \begin{itemize}
        \item $\abcnum{1}$ get a quantum token $\rho \from \tokengen(k, m)$ from the challenger, where $k \from \cU_{\gout}$ and $m \from \distr$. Also, $\abcnum{1}$ get oracle access to $P_m(\cdot)$.
        \item $\alice_1$ samples\footnote{Recall that this can be efficiently done via $2t$-wise independent hash functions where $t = \poly(\secparam)$ is a query-bound for $\abc$.} random oracles $G': \gin \to \gout$ and $ H': \hin \to \hout$, as well as random strings $m' \from \distr', w \uniform \hout$. It runs $\alice$ on input $\rho \otimes \ketbraX{w}$ with oracles $G', H'$ to obtain a bipartite state $\rho_{BC}$, which it sends to $\bob_1$ and $\charlie_1$. In addition, $\alice_1$ sends $(G', H', m', w)$ to both $\bob_1$ and $\charlie_1$.
        \item In the challenge phase, $\bob_1$ receives $\ch_\bob = k \from \gench(k, m)$ from the challenger and $(G',\allowbreak H',\allowbreak m',\allowbreak w,\allowbreak \rho_{BC}[B])$ from $\alice_1$. Observe that $\bob_1$ can implement $\Pi_B$ using these values. In particular, it can reprogram $H'$ so that it outputs $w$ on input $m$ using its oracle access to $P_m(\cdot)$, and it can reprogram $G'$ to output $k$. With this in mind, $\bob_1$ applies the efficient approximated threshold measurement $\ati_{(P,Q),\gamma_1}^{\epsilon, \delta}$ in \Cref{thm:thres_approx_asymmetric}
    with $P = \Pi_B, Q = I - \Pi_B, \gamma_1 = \trivprob + 3/4q, \epsilon = 1/4q$, and $\delta = 2^{-\secparam}$, with outcome $b_B$. If $b_B = 0$, $\bob_1$ aborts. If $b_B = 1$, then $\bob_1$ runs $\Pi_B$ and measures and outputs a random query made by $\bob$ to the reprogrammed oracle $H'$. $\charlie_1$ is defined similarly to $\bob_1$.
    
    \end{itemize}
    
    \par By \Cref{thm:thres_approx_asymmetric} bullets (1) and (2), with non-negligible probability the bipartite state obtained by $(\bob_1, \charlie_1)$ right before measuring the queries is negligibly close to a state of the form \begin{align}
            \sum_{\substack{i:\; \lambda_i > \trivprob + 1/4q \\ j:\; \gamma_j > \trivprob + 1/4q}} \beta_{i,j} \ket{\varphi_i}_B \otimes \ket{\psi_j}_C. \label{eq:ati_weights}
        \end{align}

    Before we analyze the success probability of $\abcnum{1}$, we define two algorithms $\bob_2, \charlie_2$, where $\bob_2$ is defined the same as $\bob_1$ except it does not reprogram the oracle $H'$ when implementing $\Pi_B$, and similarly for $\charlie_2$. We will refer to this different implementation of $\Pi_B$ (hence a different operator) as $\Pi_B'$. We emphasize that the only difference between $(\bob_1, \Pi_B)$ and $(\bob_2, \Pi_B')$ (similarly between $(\charlie_1, \Pi_C)$ and $(\charlie_2, \Pi_C')$) is the output of the oracle $H'$ given to $\bob$ ($\charlie$) on input $m$. \\
    
    In contrast to \cref{eq:ati_weights}, if we consider $(\alice_1, \bob_2, \charlie_1)$ instead, the state shared by $(\bob, \charlie)$ at the same stage will have negligible weight on eigenstates of $\Pi_B'$ with eigenvalue $\lambda' > \trivprob + 1/8q$, for otherwise one could violate the message hiding property (\Cref{lem:message_hiding}) using $\alice$ and $\bob$ by simulating the view of $\bob$ without the knowledge of $m$, which is information theoretically hidden from $\bob$. Similarly, the state will have negligible weight on eigenstates of $\Pi_C'$ with eigenvalue $\game' > \trivprob + 1/8q$ if we consider $(\alice_1, \bob_1, \charlie_2)$ instead. Therefore, By \Cref{thm:bbbv}, conditioned on \cref{eq:ati_weights}, i.e. conditioned on $b_B = b_C = 1$, $\bob_1$ and $\charlie_1$ will both output $m$ with non-negligible probability, hence breaking $(\distr, \negl)$ augmented security of $\game$.

    \end{proof}

\end{proof}

\begin{corollary} \label{cor:cp_qrom}
    There exists a copy-protection scheme for an arbitrary class of point functions secure in QROM.
\end{corollary}

\begin{proof}
    Follows easily from \Cref{thm:search_to_dec_cp} after setting $\game = \game_{\mathsf{BB84}}$ and $\cM'$ to be the class of point functions represented by the special input $m'$, that is, $m'$ is interpreted as the description of the point function $P_{m'}(\cdot)$.
\end{proof}

\newcommand{\bobcor}[1]{\bob_{#1}}
\newcommand{\charliecor}[1]{\charlie_{#1}}
\newcommand{\gamesetupasym}{(\setup,\allowbreak \tokengen,\allowbreak \gench_\bob,\allowbreak \ver_\bob, \gench_\charlie, \ver_\charlie)}
\newcommand{\moe}{\mathsf{MOE}}
\newcommand{\gamecd}{\game_{\mathsf{BB84-CD}}^{\mathsf{UI}}}
\newcommand{\gamecdsrch}{\game_{\mathsf{BB84-CD}}}

\section{Asymmetric Cloning Games} \label{sec:asym_cg}

Recall that we defined cloning games (e.g. see \Cref{def:unclonable_sec}) such that $\bob$ and $\charlie$ are required to pass the same verification phase, i.e. with respect to the same algorithms $(\gench, \ver)$. This does not capture all unclonable primitives, in particular cryptography with certified deletion. In this section, we extend further the cloning game syntax to include such primitives and give a simpler proof of an existing feasibility result as an application of our framework in the asymmetric setting.

\subsection{Definitions} \label{sec:asym_cg_def}
We start with the formal definition of asymmetric cloning games.
\begin{definition}[Asymmetric Cloning Game] \label{def:cloning_game_asym} An \emph{asymmetric cloning game} consists of a tuple of efficient algorithms $\game = (\setup,\allowbreak \tokengen,\allowbreak \gench_\bob,\allowbreak \ver_\bob, \gench_\charlie, \ver_\charlie)$:

\begin{itemize}
    \item {\bf Key Generation: } $\setup(1^\secparam)$ is a PPT algorithm which takes as input a security parameter $1^\secparam$ in unary. It outputs a secret key $\sk \in \{0,1\}^*$. We will assume without loss of generality that $\sk$ always contains the security parameter $1^\secparam$ in order to simplify the notation below.

    \item {\bf Token Generation:} $\tokengen(\sk, m)$ is a QPT algorithm takes as input a secret key $\sk$ and a message $m \in \{0,1\}^*$. It outputs a quantum token $\rho$. 
    
    \item {\bf Challenge Generation: } $\gench_\bob(\sk, m)$ takes as input a secret key $\sk$ and a message $m$. It outputs a classical challenge $\ch \in \{0,1\}^*$. $\gench_\charlie$ has the same syntax.
    \item {\bf Verification: } $\ver_\bob(\sk, m, \ch,\ans)$ takes as input a secret key $\sk$, a message $m$, a challenge $\ch$, and an answer $\ans$. It outputs either $0$ (reject) or $1$ (accept). $\ver_\charlie$ has the same syntax.
\end{itemize}

\end{definition}

\paragraph{Correctness.} For correctness, we require that the verifications for $\bob$ and $\charlie$ are both individually doable given the entire quantum token. 

\begin{definition}[Correctness of Asymmetric Cloning Games] \label{def:correctness_asym}
Let $\delta_\bob, \delta_\charlie : \Z^+ \to [0,1]$ and $\delta = (\delta_\bob, \delta_\charlie)$. We say that $\game$ has $\delta$-correctness if there exist efficient quantum algorithms $\bobcor{\game}, \charliecor{\game}$ such that for all messages $m \in \cM$:

\begin{align*}
    \pr{ \substack{\ver_\bob(\sk, m, \ch, \ans) = 1} \; : \; \substack{ \sk \from \setup(1^\secparam) \\ \rho \from \tokengen(\sk, m) \\ \ch \from \gench_\bob(\sk, m) \\ \ans \from \bobcor{\game}( \rho, \ch)} } \ge \delta_\bob(\secparam), \quad 
    \pr{ \substack{\ver_\charlie(\sk, m, \ch, \ans) = 1} \; : \; \substack{ \sk \from \setup(1^\secparam) \\ \rho \from \tokengen(\sk, m) \\ \ch \from \gench_\charlie(\sk, m) \\ \ans \from \charliecor{\game}( \rho, \ch)} } \ge \delta_\charlie(\secparam).
\end{align*}

If $\delta = (1,1)$ (or $\delta(\secparam) = (1 - \negl(\secparam), 1 - \negl(\secparam))$), we say $\game$ has \emph{perfect} (or \emph{statistical}) correctness. \\

\end{definition}

\paragraph{Security.} We define security formally so that $\bob$ and $\charlie$ are asked to pass their corresponding verifications. We highlight the part different from \Cref{def:cloning_exp} in \highlight{blue}.

\begin{definition}[Asymmetric Cloning Experiment] \label{def:cloning_exp_asym}
An \emph{asymmetric cloning experiment}, denoted by $\cloninggame_{\game, \distr}$, is a security game played between a referee $\referee$ and a cloning adversary $\abc$. It is parameterized by an asymmetric cloning game $\game = \gamesetupasym$ and a distribution $\distr$ over the message space $\cM$. The experiment is described as follows:
\begin{itemize}
    \item {\bf \em Setup Phase: }
    \begin{itemize}
        \item All parties get a security parameter $1^\secparam$ as input.
        \item $\referee$ samples a message $m \from \distr$.
        \item $\referee$ computes $\sk \from \setup(1^\secparam)$ and $\rho \from \tokengen(\sk, m)$.
        \item $\referee$ sends $\rho$ to $\alice$.
    \end{itemize}
    \item {\bf \em Splitting Phase: } \begin{itemize}
        \item $\alice$ computes a bipartite state $\rho'$ over registers $B,C$.
        \item $\alice$ sends $\rho'[B]$ to $\bob$ and $\rho'[C]$ to $\charlie$.
    \end{itemize}
    \item {\bf \em Challenge Phase: } \begin{itemize}
        \item $\referee$ samples \highlight{$\ch_\bob \from \gench_\bob(\sk, m)$ and $ \ch_\charlie \from \gench_\charlie(\sk, m)$}.
        \item $\referee$ sends $\ch_\bob$ to $\bob$ and $\ch_\charlie$ to $\charlie$.
        \item $\bob$ and $\charlie$ send back answers $\ans_\bob$ and $\ans_\charlie$, respectively.
        \item $\referee$ computes bits \highlight{$b_\bob \from \ver_\bob(\sk, m, \ch_\bob, \ans_\bob)$ and $b_\charlie \from \ver_\charlie(\sk, m, \ch_\charlie, \ans_\charlie)$}.
        \item The outcome of the game is denoted by $\cloninggame_{\game, \distr}(1^\secparam,\abc)$, which equals 1 if $b_\bob = b_\charlie = 1$, indicating that the adversary has won, and 0 otherwise, indicating that the adversary has lost.
    \end{itemize}
\end{itemize}

\end{definition}

\begin{definition}[Asymmetric Trivial Cloning Attack] \label{def:triv_cloning_attack_asym}
We say that $\abc$ is a \emph{{$\bob$-}trivial cloning attack} against an asymmetric cloning experiment $\cloninggame_{\game, \distr}$ if $\alice$ upon receiving a token $\rho$, sends the product state $\ketbraX{\bot} \otimes \rho$ to $\bob$ and $\charlie$. In other words, only $\charlie$ gets the token $\rho$. We denote by $\trivattack{_\bob}(\cloninggame_{\game, \distr})$ the set of $\bob$-trivial attacks against $\cloninggame_{\game, \distr}$. We similarly define $\trivattack{_\charlie}(\cloninggame_{\game, \distr})$ as the set of ${\charlie}$-trivial attacks. 

We denote by $\trivattack(\cloninggame_{\game, \distr})$ the set of trivial attacks against $\cloninggame_{\game, \distr}$. 

Finally, we define \begin{align*}
    \trivattack(\cloninggame_{\game, \distr}) := \trivattack_\bob(\cloninggame_{\game, \distr}) \cup \trivattack_\charlie(\cloninggame_{\game, \distr})
\end{align*}

\end{definition}

\noindent As pointed out in \Cref{rem:triv_success}, the definition above captures mixtures of $\bob$-trivial and $\charlie$-trivial attacks by convexity.

\begin{definition}[Asymmetric Trivial Success Probability for Cloning Games] \label{def:triv_success_prob_asym}
We define the {$\bob$-}trivial success probability of an asymmetric cloning experiment $\cloninggame_{\game, \distr}$ as \begin{align*}
    \trivprob_{\bob}(\game, \distr) := \sup_{\abc \in \trivattack{_\bob}(\cloninggame_{\game, \distr})} \pr{ 1 \from \cloninggame_{\game, \distr}(\abc) }.
\end{align*}
We similarly define $\trivprob_\charlie(\game, \distr)$ as the $\charlie$-trivial success probability of $\cloninggame_{\game, \distr}$. Accordingly, we define the trivial success probability of $\cloninggame_{\game, \distr}$ as \begin{align*}
    \trivprob(\game, \distr) :=& \sup_{\abc \in \trivattack(\cloninggame_{\game, \distr})} \pr{ 1 \from \cloninggame_{\game, \distr}(\abc) } \\
    =& \max \brac{ \trivprob_{\bob}(\game, \distr), \trivprob_{\charlie}(\game, \distr)} ,
\end{align*}
where the last equality follows from the fact that any $\abc \in \trivattack(\cloninggame_{\game, \distr})$ is a convex combination of attacks from $\trivattack_\bob(\cloninggame_{\game, \distr})$ and $\trivattack_\charlie(\cloninggame_{\game, \distr})$.

\end{definition}

\begin{definition}[Asymmetric Cloning Search Game] \label{def:asym_search_game}
    Let $\game = (\setup, \tokengen, \gench, \ver)$ be an asymmetric cloning game such that $\ver_\charlie(\sk, m, \ch, \ans)$ accepts if and only if $\ans = m$. Then, $\game$ is called an \emph{asymmetric cloning search game}.
\end{definition}

\begin{remark} \label{rem:asym_search_game}
    In the definition above, the "search" restriction only applies to $\charlie$. As such, the definition complements the asymmetry between $\bob$ and $\charlie$, while being consistent with \Cref{def:cloning_search_game}.
\end{remark}

With the modified definitions of trivial success, cloning experiment and search game, the notions of unclonable security for asymmetric games are defined similarly to \Cref{def:unclonable_sec,def:unclonable_search_sec,def:unclonable_dist_sec}.

\begin{definition}[(Asymmetric) Unclonable Security] \label{def:unclonable_sec_asym}
Let $\game$ be an asymmetric cloning game, $\distr$ be a distribution over the message space $\cM$, and $\varepsilon : \Z^+ \to [0,1]$. We say that $\game$ has \emph{$(\distr, \epsilon)$ unclonable security} if for all QPT cloning adversaries $\abc$ we have:
\begin{align*}
    \pr{ 1 \from \cloninggame_{\game, \distr}(1^\secparam, \abc) } \le \trivprob(\game, \distr) + \varepsilon(\secparam).
\end{align*}

If $|\cM| = 1$, we will simply write $\varepsilon$ unclonable security.

\end{definition}

\begin{definition}[(Asymmetric) Unclonable Search Security] \label{unclonable_search_sec_asym}
If $\game$ is an asymmetric cloning search game with $(\distr, \varepsilon)$ unclonable security, we additionally say that $\game$ has \emph{$(\distr, \varepsilon)$ unclonable search security}.
\end{definition}

\begin{definition}[(Asymmetric) Unclonable Indistinguishable Security] \label{def:unclonable_dist_sec_asym}
Let $\distr_{m_0, m_1}$ denote the distribution that outputs messages $m_0$ and $m_1$ with probability $1/2$ each. We say that an asymmetric cloning search game $\game$ has \emph{$\varepsilon$ unclonable indistinguishable security} if it has $(\distr_{m_0, m_1}, \varepsilon)$ unclonable search security for any pair of messages $m_0, m_1 \in \cM$.
\end{definition}

\begin{remark} \label{rem:ui_asym}
    Note that \Cref{def:unclonable_dist_sec_asym} could have significantly different flavor compared to its symmetric counterpart (\Cref{def:unclonable_dist_sec}), given that the distinguishing between $m_0$ and $m_1$ need not be simultaneous.
\end{remark}

\begin{remark}[Extended Asymmetric Cloning Games] \label{rem:ext_asym_cloning_game}
One could consider the extended definition for asymmetric cloning games with correlated challenge distributions as in \Cref{sec:prelims_extensions}, but it is unnecessary for the certified deletion setting which we will focus on next.
\end{remark}

\paragraph{Applications.} Asymmetric cloning games provide a framework to analyze unclonable primitives in which the goal of the adversary is to perform two different tasks (as $\bob$ and $\charlie$), each of which require a quantum token, using only one copy of the token. Examples include primitives with certified deletion property, analyzed below in \Cref{sec:deletion_games}, where $\charlie$ is asked to perform the intended use of the primitive and $\bob$ is asked to generate a classical certificate of deletion for the quantum token. Yet another example is secure software leasing \cite{AL20}, where $\bob$ is asked to return a quantum (software) token and $\charlie$ is asked to achieve the functionality of the software.

\fatih{Moved here from \Cref{sec:def_cg}}

\newcommand{\ts}{\mathsf{TS}}

\paragraph{Tokenized Signatures.} Another class of unclonable primitives is \emph{one-time primitives}, in which an honest user consumes the quantum token after using it once. While it is possible to model a one-time primitive as a (regular) cloning game, it is more natural to cast it as an asymmetric cloning game. We will demonstrate this using the example of \emph{tokenized signatures} for single-bit messages. Informally, in a tokenized signature scheme, a quantum signing token $\rho$ can be used to sign one (and only one) bit $b \in \bit$. Formally, a tokenized signature scheme~\cite{BDS16,CLLZ21,Shm22} is a tuple of efficient algorithms $(\gen, \tokengen', \sign, \ver')$: \begin{itemize}
    \item $\gen(1^\secparam)$ takes as input a security parameter and outputs a pair of classical keys $(vk, sk)$.
    \item $\tokengen'(sk)$ takes as input a secret key and outputs a quantum signature token $\rho$.
    \item $\sign(\rho, x)$ takes as input a token $\rho$ and a classical message $x \in \bit$. It outputs a classical signature $\sigma_x$.
    \item $\ver'(vk, x, \sigma)$ takes as input a verification key, a classical message, and a classical signature. It outputs 0 (reject) or 1 (accept).
\end{itemize}

\noindent $(\gen, \tokengen', \sign, \ver')$ defines an asymmetric cloning game $\game_\ts = (\setup,\allowbreak \tokengen,\allowbreak \gench_\bob,\allowbreak \ver_\bob, \gench_\charlie, \allowbreak \ver_\charlie)$, as follows: \begin{itemize}
    \item Similar to quantum money, there is no message, i.e. $m = \bot$.
    \item $\setup(1^\secparam)$ runs $(vk, sk) \from \gen(1^\secparam)$ and outputs $\sk = (vk, sk)$.
    \item $\tokengen(\sk, m)$ parses the input as $\sk = (vk, sk)$, computes $\rho \from \tokengen'(sk)$ and outputs $\rho \otimes \ketbraX{vk}$.
    \item There is no challenge for either $\bob$ or $\charlie$, i.e. $\gench_\bob$ and $\gench_\charlie$ always output $\bot$.
    \item $\ver_\bob(\sk, \bot, \bot, \ans)$ parses the input as $\sk = (vk, sk)$. It computes $b \from \ver'(vk, 0, \ans)$ and accepts if $b = 1$.
    \item $\ver_\charlie(\sk, \bot, \bot, \ans)$ parses the input as $\sk = (vk, sk)$. It computes $b \from \ver'(vk, 1, \ans)$ and accepts if $b = 1$.
\end{itemize}

Informally, $\bob$ is signing the bit $0$ and $\charlie$ is signing the bit $1$ using the signing token.

\paragraph{Correctness.}
We say that the tokenized signature scheme $(\gen, \tokengen', \sign, \ver')$ satisfies correctness if $\game_{\ts}$ has perfect correctness.

\paragraph{Security.}
We say that the scheme has unclonable security if $\game_{\ts}$ has $(\distr_\bot, \varepsilon)$ unclonable security, where $\varepsilon = \negl$ yields optimal security.

\subsection{Deletion Games} \label{sec:deletion_games}

Certified deletion can be defined as a special case of asymmetric cloning games, where $\bob$ gets no challenge.

\begin{definition}[Deletion Game] \label{def:cert_del_game}

A \emph{deletion game} is an asymmetric cloning game $\game = \gamesetupasym$ such that $\gench_\bob$ always outputs $\bot$.

\end{definition}

\par Note that in the definition above, $\bob$ represents $\alice$ generating a classical certificate for deleting $\rho$. Since $\bob$ gets no challenge, $\alice$ and $\bob$ can be effectively considered as one party.

\paragraph{Unclonable Encryption with Certified Deletion.} A notable primitive in this category is unclonable encryption with certified deletion, which we can define as a deletion game $\game_{\mathsf{UE-CD}} = \gamesetupasym$ with the following properties: \begin{itemize}
    \item $\game_{\mathsf{UE-CD}}$ is an asymmetric cloning search game. \pnote{we should remove this} \fatih{It should be fine now I think.}
    \item $\game_{\mathsf{UE-CD}}$ has statistical correctness.
    \item $\gench_\charlie(\sk, m)$ outputs $\sk$ with probability 1.
\end{itemize}

For this primitive, like regular unclonable encryption, we consider two types of security: (1) $(\cU_\cM, \varepsilon)$ unclonable security and $\varepsilon$ unclonable indistinguishable security, the latter of which is stronger.

\subsection{Construction of Unclonable Encryption with Certified Deletion}

\subsubsection{Preliminaries}

We cite two lemmas from literature that we will need in our construction. The first lemma is commonly used to bound the value of monogamy-of-entanglement games.
\begin{lemma}[Lemma 2 in \cite{TFKW13}] \label{lem:tfkw_perm}
Let $A_1, \dots, A_N$ be positive-semidefinite operators over a Hilbert space $\cH$ and let $\bracC{\pi_k}_{k \in [N]}$ be $N$ mutually orthogonal permutations over $[N]$. Then, \begin{align*}
    \opNorm{\sum_{i \in [N]} A_i} \le \sum_{k \in [N]} \max_{i \in [N]} \opNorm{\sqrt{A_i} \sqrt{A_{\pi^k(i)}}  },
\end{align*}
where $\opNorm{\cdot}$ denotes the operator norm, also known as the Schatten-$\infty$ norm.
\end{lemma}

\noindent The second lemma we will need is the local version of quantum Goldreich-Levin (\Cref{lem:quantumGL}, which was known in previous work \cite{AC02,CLLZ21}.

\begin{lemma}[Quantum Goldreich-Levin] \label{lem:quantumGL_single}

Suppose a quantum algorithm $\alice$, given a quantum state $\rho$, a key $k$, and a random string $r \in \bit^n$ can output $\inner{r}{x} (\pmod{2})$ with probability $1/2 + \varepsilon$. Then, there exists a quantum algorithm (extractor) $\alice'$, which, given the same quantum state $\rho$ and the key $k$, can output $x \in \bit^n$ with probability $4\varepsilon^2$.

\end{lemma}

\subsubsection{Achieving Unclonable Search Security}
We first describe a known construction\footnote{This is a simplified version of the construction of \cite{BI20}.} based on BB84 states, denoted as $\game_{\mathsf{BB84-CD}} = \gamesetupasym$: \begin{itemize}
        \item The message space is $\cM = \bit^\secparam$
        \item $\setup(1^\secparam)$ outputs $\sk = \theta \uniform \bit^\secparam$
        \item $\tokengen(\theta, m)$ takes as input $\theta, m \in \bit^\secparam$ and outputs $\rho = \ketbraX{m^\theta}$, where $\ket{m^\theta} = H^\theta \ket{m}$
        \item $\gench_\bob(\theta, m)$ outputs $\bot$ as required by a deletion game.
        \item $\ver_\bob(\theta, m, \bot, \ans_\bob)$ accepts if and only if $\ans_\bob \in \bit^\secparam$ satisfies $\ans_{\bob, i} = m_i$ for all $i \in \theta[1]$, where $\theta[b] := \bracC{i \in [n]: \theta_i = b}$ for $b \in \bit$.
        \item $\gench_\charlie(\theta, m)$ outputs $\theta$ as required by unclonable encryption with certified deletion.
        \item $\ver_\charlie(\theta, m, \ch_\charlie, \ans_\charlie)$ checks if $\ans_\charlie = m$ as required by an unclonable search game. 
    \end{itemize}

\paragraph{Correctness.} It is easy to see that $\gamecdsrch$ satisfies perfect correctness. In order to generate a certificate, $\bob_{\gamecdsrch}(\rho, \bot)$ measures $\rho$ in the Fourier basis and outputs the result, whereas in order to decrypt, $\charlie_{\gamecdsrch}(\rho, \theta)
$ computes $H^\theta \rho H^\theta$, then measures in the computational basis and outputs the result. 
\paragraph{Security.} While it is possible to show that $\game_{\mathsf{BB84-CD}}$ above satisfies $\negl$ unclonable indistinguishable security, the proofs involve either entropic arguments \cite{BI20} or other advanced techniques \cite{BK22}. We will instead follow a different approach which we believe is simpler in many aspects. First, we will show that $\game_{\mathsf{BB84-CD}}$ satisfies $(\cU_\cM, \negl)$ unclonable security by reduction to a monogamy-of-entanglement game following the techniques of \cite{TFKW13}, similar to the first construction of unclonable encryption \cite{BL20}. Then, we will modify the scheme and apply Quantum Goldreich-Levin to achieve $\negl$ unclonable indistinguishable security. We give the formal details below.

\begin{theorem} \label{thm:ue_cd}
    $\game_{\mathsf{BB84-CD}}$ defined above has $(\cU_\cM, \negl)$ information theoretic unclonable search security.
\end{theorem}

\begin{proof}

\par We first define a monogamy-of-entanglement game $\moe = \moe(\secparam)$ for certified deletion which is closely related to $\game_{\mathsf{BB84-CD}}$. Let $\theta[b]$ be defined as above.

$\moe$ is a game between a referee $\referee$ and an adversary $\abc$: \begin{itemize}
    \item $\abc$ and $\referee$ get a security parameter $1^\secparam$ as input.
    \item $\abc$ prepare a bipartite state $\rho_{XA}$ and send $\rho[X]$ (the $X$ register) to $\ch$, where $X = \{0,1\}^\secparam$.
    \item $\alice$ computes a bipartite state $\rho'$ over registers $B,C$, then sends $\rho'[B]$ to $\bob$ and $\rho'[C]$ to $\charlie$.
    
    \item $\referee$ samples $\theta \uniform \bit^\secparam$ and measures the $X$ register in the basis $\bracC{\ket{x^\theta}}_{x \in \bit^\secparam}$, obtaining outcome $x$.
    \item $\bob$ outputs $x_B \in \bit^\secparam$
    \item $\charlie$ gets $\theta$ as input and outputs $x_C \in \bit^\secparam$.
    
    \item The outcome of the game is denoted by $\moe(1^\secparam, \abc)$, which equals 1 if $x_C = x$ and $x_{B,i} = x_i$ for all $i \in \theta[1]$, indicating that the adversary has won, and 0 otherwise, indicating that the adversary has lost.
\end{itemize}

We will begin with showing that the success probability of any adversary in this game is exponentially small. The proof leverages the widely used techniques of \cite{TFKW13}.
\begin{claim}

For any (unbounded) adversary $\abc$, we have \begin{align*}
    \pr{ 1 \from \moe(1^\secparam, \abc)} \le \brac{\frac{1}{2} + \frac{1}{2\sqrt[4]{2}}}^\secparam 
\end{align*}
\end{claim}

\begin{proof}

We can write the winning probability of $\abc$ as follows:

\begin{align*}
    p_{\text{win}} := \pr{ 1 \from \moe(1^\secparam, \abc)} &= \E_{\theta} \sum_{\substack{x,x': \\ x_i' = x_i, \forall i \in \theta[1]}} \tr\bracS{ \brac{\ketbraX{x^\theta} \otimes B_{x'} \otimes C^\theta_x }\rho } \\
    &\le \frac{1}{2^\secparam} \opNorm{ \sum_{\substack{\theta,x,x': \\ x_i' = x_i, \forall i \in \theta[1]}} \ketbraX{x^\theta} \otimes B_{x'} \otimes C^\theta_x },
\end{align*}

\noindent where $\bracC{B_{x'}}_{x' \in \bit^\secparam}$, as well as $\bracC{C^\theta_x}_{x \in \bit^\secparam}$ for any $\theta \in \bit^\secparam$, is a POVM. By a standard purification argument, we can w.l.o.g. assume that the POVM's are projective measurements. Next, we can apply \Cref{lem:tfkw_perm} to get \begin{align}
    p_{\text{win}} \le \frac{1}{2^\secparam} \sum_{k \in \bit^\secparam} \max_{\theta \in \bit^\secparam} \opNorm{\Pi^\theta \Pi^{\pi_k(\theta)}}, \label{eq:successprob}
\end{align}

where \begin{align*}
    \Pi^\theta := \sum_{\substack{x,x': \\ x_i' = x_i, \; \forall i \in \theta[1]}} \ketbraX{x^\theta} \otimes B_{x'} \otimes C^\theta_x,
\end{align*}
and $\pi_k : \bit^\secparam \to \bit^\secparam$ are $2^\secparam$ mutually orthogonal permutations to be determined later. For convenience, we define the index sets $S[\theta, \theta'] := \bracC{i \in [\secparam] : \theta_i = \theta_i'}$ ('same' indices) and $D[\theta, \theta'] := \bracC{i \in [\secparam] : \theta_i \ne \theta_i'}$ ('different' indices). Also let $\theta' = \pi_j(\theta)$ and $s = |S[\theta, \theta']|$ for short-hand notation. Now, we have
\begin{align*}
    \Pi^{\theta'} \le \Bar{Q} := \sum_{\substack{x,x': \\ x_i' = x_i, \; \forall i \in {\theta'}[1]}} \ketbraX{x^{\theta'}} \otimes B_{x'} \otimes \identity,
\end{align*}

hence
\begin{align}
    &\opNorm{\Pi^\theta \Pi^{\theta'}}^2 \le \opNorm{\Pi^\theta \Bar{Q}}^2 = \opNorm{\Pi^\theta \Bar{Q} \Pi^\theta} \nonumber \\
    &= \opNorm{ \sum_{\substack{x,y,z,x',y',z': \\ x_i' = x_i, z_i' = z_i, \;  \forall i \in {\theta}[1] \\ y_i' = y_i, \;  \forall i \in \theta'[1] } } \ketbraX{x^\theta} \cdot \ketbraX{y^{\theta'}} \cdot \ketbraX{z^\theta} \otimes B_{x'}B_{y'}B_{z'} \otimes C^\theta_x C^\theta_z }\nonumber \\
    &= \opNorm{ \sum_{\substack{x,y,x': \\ x_i' = x_i, \;  \forall i \in {\theta}[1] \\ x_i' = y_i, \;  \forall i \in \theta'[1] } } \ketbraX{x^\theta} \cdot \ketbraX{y^{\theta'}} \cdot \ketbraX{x^\theta} \otimes B_{x'} \otimes C^\theta_x } \nonumber \\
    &= \opNorm{ \sum_{\substack{x,y,x': \\ x_i' = x_i, \;  \forall i \in {\theta}[1] \\ x_i' = y_i, \;  \forall i \in \theta'[1] } } \abs{\braket{x^\theta | y^{\theta'}}}^2 \ketbraX{x^\theta} \otimes B_{x'} \otimes C^\theta_x } \label{eq:bound2x} \\
    &= 2^{s-\secparam} \opNorm{ \sum_{\substack{x,y,x': \\ x_i' = x_i, \;  \forall i \in {\theta}[1] \\ x_i' = y_i, \;  \forall i \in \theta'[1] \\ x_i = y_i, \; \forall i \in S[\theta, \theta'] } } \ketbraX{x^\theta} \otimes B_{x'} \otimes C^\theta_x } \label{eq:bound1x}
\end{align}

Above in \cref{eq:bound2x} we used the fact that the inner product $\abs{\braket{x^\theta | y^{\theta'}}}$ vanishes unless $x_i = y_i$ for all $i \in S[\theta, \theta']$, in which case it is contributed a factor of $2^{-1/2}$ for every index $i \notin S[\theta, \theta']$. Now in \cref{eq:bound1x}, every term in the sum with distinct $(x,x')$ is orthogonal, so we only need to count the number of $y$ values for given $(x,x')$. Specifically, we need the number of $y \in \bit^\secparam$ such that $y_i = x_i'$ for $i \in \theta'[1]$ and $y_i = x_i$ for $i \in S[\theta, \theta']$. Note that these two conditions never contradict because $x_i' = x_i$ for all $i \in \theta[1]$, and $S[\theta, \theta'] \cap \theta'[1] \subseteq \theta[1]$. Hence, $y$ has $|S[\theta, \theta'] \cup \theta'[1]|$ coordinates fixed. The number of free coordinates, then is given by $\secparam - |S[\theta, \theta'] \cap \theta'[1]| = |\theta[1] \setminus \theta'[1]|$, which is the number of indices where $\theta_i' = 0, \theta_i = 1$. Thus, \cref{eq:bound1x} equals $2^{-t}$, where $t$ is the number of indices $i$ such that $\theta_i = 0, \theta'_i = 1$, then we can bound \cref{eq:successprob} as 
\begin{align*}
    p_{\text{win}} \le \frac{1}{2^\secparam} \sum_{k \in \bit^\secparam} \max_{\theta \in \bit^\secparam} \opNorm{\Pi^\theta \Pi^{\pi_k(\theta)}} \le \frac{1}{2^\secparam} \sum_{k \in \bit^\secparam} \max_{\theta \in \bit^\secparam}  2^{-t/2} 
\end{align*}

\par Furthermore, observe that by symmetry, we can achieve the same bound for $t$ representing the number of indices $i$ with $\theta_i=1, \theta_i' = 0$ as well for any $\theta, \theta'$. Thus, without loss of generality we could assume $t \ge |D[\theta, \theta']|/2$ by choosing the optimal order in the analysis. Finally, we choose $\pi_k(\theta) = \theta \oplus k$ to be the cyclic permutations, which yields $\binom{\secparam}{k}$ permutations with $D[\theta, \theta'] = k$. Thus, \begin{align*}
    p_{\text{win}} \le \frac{1}{2^\secparam} \sum_{k=0}^\secparam \binom{\secparam}{k} 2^{-k/4} = \brac{\frac{1}{2} + \frac{1}{2\sqrt[4]{2}}}^\secparam \approx 0.920^\secparam
\end{align*}

as desired.

\end{proof}

\begin{remark}
\cite{CV21} bounds an easier version of $\moe$ by $0.85^{\secparam/2} \approx 0.924^\secparam$, where $\bob$ and $\charlie$ need to guess indices in $\theta[0]$ and $\theta[1]$ after learning $\theta$, respectively. Because $\theta$ is unknown, $\alice$ cannot trivially succeed by splitting the qubits between them.
\end{remark}

We finish the proof of the theorem by giving a reduction from unclonable security of $\game_{\mathsf{BB84-CD}}$ to hardness of $\moe$. Since $0.920^\secparam$ is negligible in $\secparam$, this suffices.

\begin{claim} \label{clm:cd_to_moe}
Suppose there exists an adversary $\abc$ breaking $(\cU_{\cM}, \varepsilon)$ unclonable search security of $\game_{\mathsf{BB84-CD}}$, then there exists $\abcprime$ such that \begin{align*}
    \pr{ 1 \from \moe(1^\secparam, \abc)} > \varepsilon.
\end{align*}
\end{claim}

\begin{proof}
Define $\abcprime$ as follows: \begin{itemize}
    \item $\abcprime$ creates $\secparam$ EPR pairs, i.e. the bipartite state $\ket{\psi} = 2^{-\secparam/2} \sum_{x \in \bit^\secparam}{\ket{x}\ket{x}}$ over registers $XA$, and send the $X$ register to the referee $\referee$.
    \item $\alice'$ applies the same splitting channel as $\alice$. Similarly, $\bobprime$ and $\charlieprime$ apply the same measurements as $\bob$ and $\charlie$.
\end{itemize}
We will show that $\abcprime$ defined above succeeds in $\moe$ with probability equal to the success probability of $\abcprime$. Recall that for any $\theta \in \bit^\secparam$, we have $\ket{\psi} = 2^{-\secparam/2} \sum_{x \in \bit^\secparam}{\ket{x^\theta}\ket{x^\theta}}$. Hence, \begin{align*}
    &\pr{ 1 \from \moe(1^\secparam, \abcprime)  } = \E_{\theta} \tr \bracS{ \sum_{\substack{x,x': \\ x_i' = x_i, \; \forall i \in {\theta'}[1]}} \brac{\ketbraX{x^\theta} \otimes B_{x'} \otimes C^\theta_x } (I_X \otimes \alice) \brac{ \ketbraX{\psi} }} \\
    &=  2^{-\secparam} \E_{\theta} \tr \bracS{ \sum_{\substack{x,x': \\ x_i' = x_i, \; \forall i \in {\theta'}[1]}} \sum_{y,z} \brac{\ketbraX{x^\theta} \otimes B_{x'} \otimes C^\theta_x } (I_X \otimes \alice) \brac{ \ket{y^\theta} \bra{z^\theta} \otimes \ket{y^\theta} \bra{z^\theta} }} \\
    &= 2^{-\secparam} \E_{\theta} \tr \bracS{ \sum_{\substack{x,x': \\ x_i' = x_i, \; \forall i \in {\theta'}[1]}} \sum_{y,z} \brac{\ketbraX{x^\theta} \otimes B_{x'} \otimes C^\theta_x } \brac{ \ket{y^\theta} \bra{z^\theta} \otimes \alice \brac{\ket{y^\theta} \bra{z^\theta}} }} \\
    &= 2^{-\secparam} \E_{\theta} \tr \bracS{ \sum_{\substack{x,x': \\ x_i' = x_i, \; \forall i \in {\theta'}[1]}} \brac{\ketbraX{x^\theta} \otimes B_{x'} \otimes C^\theta_x } \brac{ \ket{x^\theta} \bra{x^\theta} \otimes \alice \brac{\ket{x^\theta} \bra{x^\theta}} }} \\
    &=\E_{\theta, x} \tr \bracS{ \sum_{\substack{x,x': \\ x_i' = x_i, \; \forall i \in {\theta'}[1]}} \brac{ B_{x'} \otimes C^\theta_x } \brac{ \alice \brac{\ket{x^\theta} \bra{x^\theta}} }} \\
    &= \pr{ 1 \from \cloninggame_{\game_{\mathsf{BB84-CD}}, \cU_\cM}(1^\secparam,\abc) },
\end{align*}
which completes the proof.

\end{proof}

\end{proof}

\subsubsection{Achieving Unclonable-Indistinguishable Security}
Now, we give the construction of unclonable-indistinguishable secure unclonable encryption with certified deletion, denoted as $\gamecd = (\gamesetupasym)$ below: \begin{itemize}
    \item We consider single-bit messages, i.e. $\cM = \bit$
    \item $\setup(1^\secparam)$ samples $\theta \uniform \bit^\secparam$ and $r \uniform \bit^\secparam \setminus\{0^\secparam\}$ independently. It outputs $\sk = (\theta, r)$.
    \item $\tokengen(\sk, m)$ parses the input as $\sk = (\theta, r)$. It samples $x \uniform \bracC{ x' \in \bit^\secparam \; : \; \inner{r}{x'} = m }$, where we treat $\bit^\secparam$ as $\F_2^\secparam$ for the inner product $\inner{\cdot}{\cdot}$. It outputs $\rho = \ketbraX{x^\theta}$.
\end{itemize}

\begin{theorem} \label{thm:ue_cd_ui}
    $\gamecd$ above is an unclonable encryption scheme for single-bit messages with certified deletion, which has $\negl$ information theoretic unclonable indistinguishable security.
\end{theorem}

\begin{proof}
    Correctness is easy to see. For security, we define a sequence of hybrids: \begin{itemize}
        \item {\bf Hybrid 0:} The asymmetric cloning experiment $\cloninggame_{\gamecd,\; \cU_{\bit}}$. 
        \par Suppose for the sake of contradiction that there exists an adversary $\abc$ which succeeds in this experiment with probability $1/2 + \varepsilon$, where $\varepsilon$ is non-negligible in $\secparam$.
        \item {\bf Hybrid 2:} In this hybrid, we sample $r \uniform \bit^\secparam$, i.e. we allow $r = 0^\secparam$. \par Since the probability of this is negligibly small, the success probability of $\abc$ in this Hybrid is $\varepsilon - \negl(\secparam)$.
        \item {\bf Hybrid 2:} In this hybrid, instead of sampling $m$ uniformly and sampling $x$ conditioned on $\inner{r}{x} = m$, we sample $x \uniform \bit^\secparam$ uniformly and set $m = \inner{r}{x}$. In other words, we remove $m$ from the experiment and ask $\charlie$ to output $\inner{r}{x}$ in order to pass verification. \par This hybrid is statistically indistinguishable from Hybrid 0, since $\abs{\pr{\inner{r}{x} = 0} - 1/2} \le \negl(\secparam)$ for $r,x \uniform \bit^\secparam$. Thus, $\abc$ succeeds in this Hybrid with probability $\varepsilon - \negl(\secparam)$. Let $p_\bob$ be the probability that $\bob$ passes verification, and let $p_\charlie$ be the probability that $\charlie$ passes verification conditioned on $\bob$ passing verification. Then, $p_\bob p_\charlie \ge \varepsilon - \negl(\secparam)$, so that $p_\bob$ and $p_\charlie$ are both non-negligible in $\secparam$.
        \item {\bf Hybrid 3:} This is the asymmetric cloning experiment $\cloninggame_{\game_{\mathsf{BB84-CD}}, \; \cU_{\bit^\secparam}}$. \par Observe that the only difference between Hybrids 2 and 3 is the verification phase for $\charlie$. Accordingly, we define an adversary $(\alice, \bob, \charlie')$ for this Hybrid, where $\charlie'$ is the extractor guaranteed by \Cref{lem:quantumGL_single}, applied with respect to the mixed state received by $\charlie$ in Hybrid 1 conditioned on $\bob$ passing verification. The guarantee of \Cref{lem:quantumGL_single} states that if $p_\charlie'$ is the probability of $\charlie'$ passing verifiction conditioned on $\bob$ passing verification, then $p_\charlie' \ge 4p_\charlie^2$, which is non-negligible. Therefore, $(\alice, \bob, \charlie')$ succeeds in Hybrid 3 with non-negligible probability $p_\bob p_\charlie'$, contradicting \Cref{thm:ue_cd}.
    \end{itemize}
\end{proof}

\begin{remark}
For simplicity, we give a construction for single-bit messages, but a standard hybrid argument can be used to show that bitwise encryption works for multi-bit messages in the certified deletion setting. \fatih{I believe this should be easy.}
\end{remark}

\printbibliography

\newpage 

\appendix

\appendix

\section{Alternate Proof of Simultaneous Quantum Goldreich-Levin} \label{sec:app_gl}

\par Below we give an alternate\footnote{Kundu and Tan have independently generalized the Goldreich-Levin technique to the non-local (simultaneous) setting \cite{KT22}. The authors apply this technique to achieve a weaker form of unclonable encrpytion in the plain model, whereas we apply it to achieve single-decryptor encryption with unclonable security against independently generated ciphertexts.} proof of \Cref{lem:quantumGL} which does not use the lifting theorem. 
\begin{proof}[Direct proof of \Cref{lem:quantumGL}] We will adapt the proof\footnote{See Lemma B.12 in \cite{CLLZ21}.} of \cite{CLLZ21}, originally due to \cite{AC02}, to the simultaneous case. We can assume that the token $\rho_{k,x} = \ketbraX{\psi_{k,x}}$ is a pure state, for the mixed state case follows by convexity. Observe that we can defer any measurements made by $\abc$ until the very end. Accordingly, we can model $\alice$ as a unitary map $\Phi$ which acts as: \begin{align*} \Phi \ket{\psi_{k,x}} \ket{0^m}_{aux} = \ket{\varphi_{k,x}}_{BC},
\end{align*}

where $\ket{\varphi_{k,x}}_{BC}$ is a bipartite state shared by $\bob$ and $\charlie$. In phase 2, the key $k$ as well as the random coins $(r,r')$ are revealed and unitary maps $(U^{k,r}_B, U^{k,r'}_C)$ are applied by $\bob$ and $\charlie$, respectively. The resulting state then is given by \begin{align*} \left( U^{k_B,r}_B \otimes U^{k_C,r'}_C \right) \ket{\varphi_{k,x}}_{BC}
&= \bigg( \alpha_{k,x,r,r'} \ket{\inner{r}{x}}_B \ket{\inner{r'}{x}}_C \ket{\phi_{k,x,r,r'}^0}_{BC}
+ \beta_{k,x,r,r'} \ket{\inner{r}{x}}_B \ket{\overline{\inner{r'}{x}}}_C \ket{\phi_{k,x,r,r'}^1}_{BC} \\
&+ \theta_{k,x,r,r'} \ket{\overline{\inner{r}{x}}}_B \ket{\inner{r'}{x}}_C \ket{\phi_{k,x,r,r'}^2}_{BC}
+ \gamma_{k,x,r,r'} \ket{\overline{\inner{r}{x}}}_B \ket{\overline{\inner{r'}{x}}}_C \ket{\phi_{k,x,r,r'}^3}_{BC} \bigg) \\
&=: \ket{\Gamma_{k,x,r,r'}}, \end{align*}
where $\ket{\phi^j_{k,x,r,r'}}$ is a normalized state for $j \in \bracC{0,1,2,3}$ and $\alpha_{k,x,r,r'}$ is the coefficient corresponding to the case of the adversary succeeding, so that we can express the assumption as \begin{align*} \E_{k,x,r,r'} \abs{\alpha_{k,x,r,r'}}^2 \ge \frac{1}{2} + \varepsilon
\end{align*}
and hence \begin{align}
    &\E_{k,x,r,r'} \abs{\alpha_{k,x,r,r'}}^2 - \abs{\beta_{k,x,r,r'}}^2 - \abs{\theta_{k,x,r,r'}}^2 + \abs{\gamma_{k,x,r,r'}}^2 \ge \E_{k,x,r,r'} \brac{2\abs{\alpha_{k,x,r,r'}}^2 - 1} \ge 2\varepsilon \label{eq:coeff_bound}
\end{align}

We now describe the new adversary $\abcprime$: \begin{itemize}
    \item Given $\ket{\psi_{k,x}}$ in phase 1, $\alice'$ acts the same as $\alice$, i.e. it applies $\Phi$, obtaining the state $\ket{\varphi_{k,x}}_{BC}$.
    \item After receiving the key $k$ from the challenger (ignoring the random coins received), $\bob'$ prepares a uniform superposition over $r \in \cR$ and applies the unitary $U^{k_B}_B$, where we define $U^{k_E}_E$ as $U^{k_E}_E \ket{r}\ket{\varphi} = \ket{r}U^{k_E,r}_E \ket{\varphi}$ for $E \in \bracC{B,C}$. Then, $\bob'$ applies a $Z$ gate to the register storing the inner product $\inner{r}{x}$, and applies $(U^{k_B}_B)^\dagger$ to its state. Finally, $\bob'$ measures the register storing the random coins $r$ in the Fourier basis and outputs the result.
    \item $\charlie'$ is defined in a similar fashion.
\end{itemize}
Next, we will analyze the evolution of the state shared by $\bob'$ and $\charlie'$ step by step. Since the actions of $\bob'$ and $\charlie'$ commute, we can synchronously track their operations. After the first step, the state is given by \begin{align*} \left(U^{k_B}_B \otimes U^{k_C}_C \right) \frac{1}{{|\cR|}} \sum_{r,r' \in \cR} \ket{r}_B\ket{r'}_C \ket{\varphi_{k,x}}_{BC} =  \frac{1}{{|\cR|}} \sum_{r,r' \in \cR} \ket{r}_B\ket{r'}_C \ket{\Gamma_{k,x,r,r'}}. 
\end{align*}

Next, $\bob'$ and $\charlie'$ each apply a $Z$ gate to their register storing the inner product, which results in the state \begin{align*} &\frac{1}{{|\cR|}} \sum_{r,r' \in \cR}  \ket{r}_B\ket{r'}_C (-1)^{\inner{r}{x} \oplus \inner{r'}{x}} \bigg( \alpha_{k,x,r,r'} \ket{\inner{r}{x}}_B \ket{\inner{r'}{x}}_C \ket{\phi_{k,x,r,r'}^0}_{BC} \\
&- \beta_{k,x,r,r'} \ket{\inner{r}{x}}_B \ket{\overline{\inner{r'}{x}}}_C \ket{\phi_{k,x,r,r'}^1}_{BC}
- \theta_{k,x,r,r'} \ket{\overline{\inner{r}{x}}}_B \ket{\inner{r'}{x}}_C \ket{\phi_{k,x,r,r'}^2}_{BC} \\
&+ \gamma_{k,x,r,r'} \ket{\overline{\inner{r}{x}}}_B \ket{\overline{\inner{r'}{x}}}_C \ket{\phi_{k,x,r,r'}^3}_{BC} \bigg) \\
&=: \frac{1}{{|\cR|}} \sum_{r,r' \in \cR} \ket{r}_B\ket{r'}_C \ket{\Gamma'_{k,x,r,r'}},
\end{align*}
with \begin{align*} \inner{\Gamma_{k,x,r,r'}}{\Gamma'_{k,x,r,r'}} = (-1)^{\inner{r}{x} \oplus \inner{r'}{x}}\left(|\alpha_{k,x,r}|^2 - |\beta_{k,x,r}|^2 - |\theta_{k,x,r}|^2 + |\gamma_{k,x,r}|^2\right).
\end{align*}
Now $\bob'$ and $\charlie'$ uncompute the unitary $U^{k_B}_B \otimes U^{k_C}_C$, and the state becomes \begin{align*}
&\left(U^{k_B}_B \otimes U^{k_C}_C\right)^{\dagger} \frac{1}{{|\cR|}} \sum_{r \in \cR} \ket{r}_B\ket{r'}_C \ket{\Gamma'_{k,x,r,r'}} =  \frac{1}{{|\cR|}} \sum_{r \in \cR} \ket{r}_B\ket{r'}_C \left(U^{{k_B},r}_B \otimes U^{{k_C},r'}_C\right)^{\dagger} \ket{\Gamma'_{k,x,r,r'}} \\
&= \frac{1}{{|\cR|}} \sum_{r,r' \in \cR} \ket{r}_B\ket{r'}_C \left((-1)^{\inner{r}{x} \oplus \inner{r'}{x}}\left(|\alpha_{k,x,r,r'}|^2 - |\beta_{k,x,r,r'}|^2 - |\theta_{k,x,r,r'}|^2 + |\gamma_{k,x,r,r'}|^2\right) \ket{\varphi_{k,x}}_{BC} + \ket{\mathsf{err}_{k,x,r,r'}} \right), 
\end{align*}
where $\ket{\mathsf{err}_{k,x,r,r'}}$ is a subnormalized state orthogonal to $\ket{\varphi_{k,x}}_{BC}$. \\

\par Next, $\bob'$ and $\charlie'$ each apply a Quantum Fourier Transform (QFT) on their random coins, resulting in the state \begin{align*}
    \frac{1}{{|\cR|^{2}}} &\sum_{r,r' \in \cR} \sum_{y,z \in \cR} (-1)^{\inner{r}{y} \oplus \inner{r'}{z}}\ket{y}_B\ket{z}_C \bigg((-1)^{\inner{r}{x} \oplus \inner{r'}{x}}\big(|\alpha_{k,x,r,r'}|^2 - |\beta_{k,x,r,r'}|^2 \nonumber \\
    &- |\theta_{k,x,r,r'}|^2 + |\gamma_{k,x,r,r'}|^2\big) \ket{\varphi_{k,x}}_{BC} + \ket{\mathsf{err}_{k,x,r,r'}} \bigg).
\end{align*}
Note that the coefficient of $\ket{x}_B\ket{x}_C\ket{\varphi_{k,x}}$ equals \begin{align*}
    \frac{1}{|\cR|^{2}} \sum_{r,r' \in \cR} |\alpha_{k,x,r,r'}|^2 - |\beta_{k,x,r,r'}|^2 - |\theta_{k,x,r,r'}|^2 + |\gamma_{k,x,r,r'}|^2,
\end{align*}
so the probability that $\bob'$ and $\charlie'$ both output $x$ is lower bounded by \begin{align*}
    \pr{y = z = x} &\ge \E_{k,x}  \abs{ \frac{1}{|\cR|^{2}} \sum_{r,r' \in \cR} |\alpha_{k,x,r,r'}|^2 - |\beta_{k,x,r,r'}|^2 - |\theta_{k,x,r,r'}|^2 + |\gamma_{k,x,r,r'}|^2 }^2 \\
    &= \E_{k,x} \abs{ \E_{r,r'} |\alpha_{k,x,r,r'}|^2 - |\beta_{k,x,r,r'}|^2 - |\theta_{k,x,r,r'}|^2 + |\gamma_{k,x,r,r'}|^2 }^2 \\
    &\ge \abs{ \E_{k,x,r,r'} |\alpha_{k,x,r,r'}|^2 - |\beta_{k,x,r,r'}|^2 - |\theta_{k,x,r,r'}|^2 + |\gamma_{k,x,r,r'}|^2 }^2  \\
    &\ge 4\varepsilon^2,
\end{align*}
where we used Cauchy-Schwartz Inequality and \cref{eq:coeff_bound}.


\end{proof}

\end{document}